\documentclass[12pt]{article}
 \usepackage{xcolor}
\usepackage{amsfonts,amsmath}
\usepackage[mathscr]{eucal}
\usepackage{amssymb}
\usepackage{bm}
\usepackage{bbold}
\usepackage{amsthm}
\usepackage{graphicx,subfigure}
\usepackage{dsfont}

\usepackage{ytableau}

\usepackage{overpic}

\usepackage{hyperref}

\usepackage[titletoc,toc,title]{appendix}

\newcommand{\appendixnumberline}[1]{Appendix #1.\space}
\newcommand{\vecind}[1]{\mbox{\scriptsize \boldmath $#1$}}
\newcommand{\pd}{\partial}

\let\oldappendix\appendix
\makeatletter
\renewcommand{\appendix}{%
  \addtocontents{toc}{\let\protect\numberline\protect\appendixnumberline}%
  \renewcommand{\@seccntformat}[1]{
  \bfseries Appendix \Alph{section}. }%
  \oldappendix
}
\makeatother

\theoremstyle{plain}
\newtheorem{thm}{Theorem}

\definecolor{navycol}{RGB}{100,150,160}
   \definecolor{pinkcol}{RGB}{242,55,55}
   \definecolor{bluecol}{RGB}{50,205,50}
   \definecolor{bluecol}{RGB}{30,144,255}
   \definecolor{yellowcol}{RGB}{255,252,134}
   \definecolor{lbluecol}{RGB}{214,252,168}
   
   \usepackage{empheq}

\def\theequation{\arabic{section}.\arabic{equation}}
\newcommand{\be}{\begin{eqnarray}}
\newcommand{\red}{\color{red}}
\newcommand{\blue}{\color{blue}}
\newcommand{\ee}{\end{eqnarray}}
\newcommand{\nn}{\nonumber \\}
\newcommand{\lb}{\label}
\newcommand{\p}[1]{(\ref{#1})}

\newcommand{\ga}{\lower.7ex\hbox{$
\;\stackrel{\textstyle>}{\sim}\;$}}
\newcommand{\vecg}[1]{\mbox{\boldmath $#1$}}
\newcommand{\la}{\lower.7ex\hbox{$
\;\stackrel{\textstyle<}{\sim}\;$}}

\newcommand{\bc}{\begin{center}}
\newcommand{\ec}{\end{center}}

\usepackage{setspace}
\usepackage{tcolorbox}
\tcbuselibrary{theorems}

\newtcbtheorem
  []
  {theorem}
  {Theorem}
  {
    colback=bluecol!20!,
    colframe=blue!50!,
    fonttitle=\sffamily,
  }
  {def}

  \newtcbtheorem
  []
  {lemma}
  {Lemma}
  {
    colback=yellowcol!30!,
    colframe=pinkcol!60!,
    fonttitle=\sffamily,
  }
  {def}

\let\OLDthebibliography\thebibliography
\renewcommand\thebibliography[1]{
  \OLDthebibliography{#1}
  \setlength{\parskip}{4pt}
  \setlength{\itemsep}{0pt plus 0.3ex}
}

\begin{document}

\begin{titlepage}

\vspace*{0.2cm}

\renewcommand{\thefootnote}{\star}
\begin{center}

{\LARGE\bf  Witten index in $4d$ supersymmetric gauge theories}\\

\vspace{0.5cm}

\vspace{1.5cm}
\renewcommand{\thefootnote}{$\star$}

\quad {\large\bf Andrei~Smilga} 
 \vspace{0.5cm}

{\it SUBATECH, Universit\'e de Nantes,}\\
{\it 4 rue Alfred Kastler, BP 20722, Nantes 44307, France;}\\
\vspace{0.1cm}

{\tt smilga@subatech.in2p3.fr}\\

\end{center}
\vspace{0.2cm} \vskip 0.6truecm \nopagebreak

\begin{abstract}
\noindent  
We present a review of Witten index calculations in different supersymmetric gauge theories in four dimensions: supersymmetric electrodynamics, pure ${\cal N} = 1$ supersymmetric Yang-Mills theories and also SYM theories including matter multiplets --- both with chirally symmetric and asymmetric content.
 \end{abstract}

\vspace{1cm}
\bigskip

\newpage

\end{titlepage}

\tableofcontents

\setcounter{footnote}{0}

\setcounter{equation}0

\section{Introduction} 

Supersymmetric theories in four dimensions are characterized by the algebra\footnote{The supercharges may carry an additional flavor index; in this case we are dealing with an {\it extended} supersymmetric theory. But we focus in this review on  ${\cal N} = 1$ theories involving only one pair of complex supercharges.}
\be
\lb{SUSY-alg}
\{\hat Q_\alpha, \hat {\bar Q}_{\dot\alpha} \} \ =\ 2(\sigma^\mu)_{\alpha\dot\alpha} \hat P_\mu,
 \ee
where $\hat Q_\alpha$ and $\hat {\bar Q}_{\dot\alpha}$  are  Hermitially conjugated supercharges belonging to the representations $(\frac 12, 0)$ and $(0, \frac 12)$ of the Lorentz group and $\hat P_\mu$ is the operator of 4-momentum. The anticommutators $\{\hat Q_\alpha, \hat Q_\beta\}$ and $\{\hat{\bar Q}_{\dot\alpha}, \hat {\bar Q}_{\dot\beta}\}$ vanish.

In all such theories, the spectrum of the excited states in the Hamiltonian $\hat P_0$ is split into degenerate quartets. Each such quartet includes a state $\Psi$ annihilated by both supercharges $\hat {\bar Q}_{\dot\alpha}$. It is {\it bosonic} in a sense that it carries an integer spin.  Then there are two {\it fermionic} states $\Psi_\alpha 
= \hat Q_\alpha \Psi$ carrying a half-integer spin and a state $\Psi' = \hat Q_1 \hat Q_2 \Psi$. The latter state is annihilated by $\hat Q_\alpha$ and is also bosonic.  The states in the quartets have positive energies.

But the spectrum may include also the {\it vacuum} states that are annihilated by all the supercharges and, in a consequence of the algebra \p{SUSY-alg}, have zero energies. If such states are present, people say that supersymmetry is not broken. There are also theories where the zero-energy states are absent and, for any eigenstate $\Psi$ of the Hamiltonian, $\hat Q_\alpha \Psi$ and 
 $\hat {\bar Q}_{\dot\alpha} \Psi$ do not vanish simultaneously. If this happens, supersymmetry is spontaneously broken.

The numbers of bosonic and fermionic states in any excited quartet or in a set of such quartets is  equal. But it need not to be so for the vacuum states.  The  Witten index is defined as the difference between the number of bosonic and fermionic vacuum states:
\be
\lb{IW}
I_W \ =\ n_B^{(0)} -  n_F^{(0)}\,.
 \ee
The importance of this object stems from the fact that it is {\it invariant} under a smooth deformation of the theory that keeps the algebra \p{SUSY-alg}. Indeed, some vacuum states may be shifted from zero and acquire a positive energy, but, bearing in mind the degeneracy mentioned above, this is possible only for a {\it quartet} of states leaving the difference \p{IW} intact.

This suggests a clever way to evaluate the index in a complicated theory \cite{Wit82}: deform it to make it as simple as possible and find the index in the deformed theory. In the original theory, it must be the same. 

If the index thus determined is nonzero, one can be sure that supersymmetry is not broken. If the index is zero, one can say that supersymmetry is {\it probably}  broken, but cannot be sure of it --- counterexamples are known. An additional analysis is required in this case. Note that to understand whether supersymmetry is broken or not in a given theory is important for hypothetic phenomenological applications. In fact, this is the main motivation of Witten index studies.

One of the ways to deform the theory is to put it in a small spatial box. For some theories, one can get rid of the spatial field dependence whatsoever and reduce a field theory problem to a quantum mechanical problem, which is, of course, much simpler.

Bearing in mind the double  degeneracy of the excited states (four-fold for $4D$ theories and their reductions), the Witten index \p{IW} can be represented as a {\it supersymmetric partition function}, 
\be
\lb{SUSY-Z}
I_W \ =\ {\rm Tr} \left\{ (-1)^{\hat F} e^{-\beta \hat H} \right\}\,
 \ee
where $\hat F$ is the  operator of the fermion charge, which can be well defined in almost all theories\footnote{There are exceptions, but they are beyond the scope of our review. The {\it fermion parity} $(-1)^F$
is well defined in all supersymmetric systems.}, being even for bosonic and odd for fermionic states, and $\beta$ is an arbitrary real number having the meaning of inverse temperature. The RHS of Eq.\p{SUSY-Z} can be represented as a
functional integral:

\be
   \lb{path-Grassmann}
   I_W   \ \propto  \ 
   \int \prod_{j\tau} dq_j(\tau) \prod_{a\tau}   d\psi_a(\tau) d \bar \psi_a (\tau) 
 \exp\left\{\!\! - \!\!\int_0^\beta \!\! L_E[q_j(\tau), \psi_a(\tau), \bar \psi_a (\tau)]  d\tau \right \}\!,
    \ee
where  $L_E$ is the Euclidean Lagrangian depending on the  bosonic ($q_j$) and fermionic ($ \psi_a,  \bar \psi_a$) dynamic variables in the reduced mechanical system, which satisfy    the  periodic boundary conditions in the imaginary time $\tau$:
 \be
 \lb{bc-SUSY}
 q_j(\beta) = q_j(0); \quad \psi_a(\beta) = \psi_a(0); \quad  \bar \psi_a(\beta) = \bar \psi_a(0)\,.
  \ee
 For many  systems (though not for all of them --- see the discussion in   Sec. 6),  the higher Fourier harmonics in the expansion,
\be
q_j(\tau) \ =\ \sum_n q_j^{(n)} e^{2\pi i n \tau/\beta}, \qquad \psi_a(\tau) \ =\ \sum_n \psi_a^{(n)} e^{2\pi i n \tau/\beta},
\ee
can be, for small $\beta$,  effectively integrated over, and
the functional integral is then reduced for to an {\it ordinary} phase space integral  over the {\it constant} field configurations \cite{Cecotti}:
\begin{empheq}[box=\fcolorbox{black}{white}]{align}
\lb{IW-int-final}
I_W \ =\ \lim_{\beta \to 0} \int \prod_j \frac {dp_j dq_j}{2\pi} \prod_a d\psi_a d \bar \psi_a \,  e^{- \beta H(p_j, q_j; \,\bar\psi_a, \psi_a) } \, ,
 \end{empheq}
Calculating this integral allows one to evaluate the index in the original theory.

Finally, we would like to mention that the notion of the Witten index represents, if you will, a supersymmetric avatar of some notions known in pure mathematics.  A very similar object is   the so-called {\it equivariant index} introduced by Cartan back in 1950 \cite{Cartan}. The Witten index is also closely related to the Atiyah-Singer index \cite{At-Sin}.
 Supersymmetric language allows one to describe in a rather transparent way some known facts of differential geometry and derive also new results in this field \cite{Wit-Morse, glasses}. 

But all these fascinating issues are not the subject of our review. We will limit the scope of our discussion to supersymmetric gauge theories in four dimensions. With rare exceptions, we will report the results on the evaluation of the index in these theories performed at the end of the last millennium, before 2000. But to the best of our knowledge, a detailed review, where all these results were comprehensively described and explained, has not been written yet. And better late than never...

We use the notations  that are close to the notations of Wess and Bagger \cite{WB} with the main distinction being the  signature of the Minkowski metric, which we choose $(+---)$ rather than $(-+++)$. We define $A^\mu = (A_0, \vecg{A}) \Rightarrow A_\mu = (A_0, -\vecg{A})$, but $\pd_\mu = (\pd_0, \vecg{\pd})$.

The fermion fields are the Weyl 2-component spinors $\psi_\alpha$ and the complex conjugated $\bar\psi_{\dot \alpha}$. The spinor indices are lifted and lowered according to 
\be \lb{low-raise}
X_\alpha \ =\ \varepsilon_{\alpha\beta} X^\beta, \quad X^\alpha \ =\ \varepsilon^{\alpha\beta} X_\beta
 \ee
with $\varepsilon_{\alpha\beta} = -\varepsilon^{\alpha\beta}$ and $\varepsilon_{12} = 1$, and the same for the dotted indices of complex conjugated spinors.
The same concerns the superspace coordinates $\theta_\alpha$ and $\bar \theta_{\dot \alpha}$.
It is convenient to introduce the shortcuts
\be
\lb{prod-spin}
\eta \chi = \eta^\alpha \chi_\alpha =  \chi^\alpha\eta_\alpha, \qquad  \bar\eta \bar\chi = \bar \eta_{\dot\alpha} \bar\chi^{\dot\alpha} =  \bar\chi_{\dot\alpha}  \bar\eta^{\dot\alpha}
\ee
for the products of two Grassmann spinors. 
We will also meet a
matrix-valued  4-vector,
 \be
\lb{sigma-mu}
\sigma^\mu_{\alpha\dot\alpha} \ =\ (\mathbb{1}, \vecg{\sigma})_{\alpha\dot\alpha}, 
 \ee
where $\vecg{\sigma}$ are the Pauli matrices.

That were 4-dimensional notations. We will often perform a dimensional reduction reducing a supersymmetric field theory  to a supersymmetric quantum mechanical (SQM) system. In that case, the convenient notations are slightly different. 

 In (0+1) dimensions, there is no point to distinguish the dotted and undotted indices. The ${\cal N} = 1$ supersymmetric algebra \p{SUSY-alg} goes over into an {\it extended}  SQM algebra: 
 \be
\lb{SQM-alg}
\{\hat Q_\alpha, \hat Q_\beta\} \ = \ \{\hat {\bar Q}^\alpha, \hat {\bar Q}^\beta\} \ =\  0, \nn
\{\hat Q_\alpha, \hat {\bar Q}^{\beta} \} \ =\ 2\delta_\alpha^\beta  \hat H \,,
 \ee
where $\hat {\bar Q}^{\alpha} = (\hat Q_\alpha)^\dagger$.

\vspace{1mm}

In the next two sections, I will expose the results of Witten's analysis \cite{Wit82}.

\setcounter{equation}0
\section{Supersymmetric electrodynamics}
The superfield Lagrangian includes a real vector supermultiplet $V(x, \theta, \bar \theta)$ and a couple of left chiral superfields 
 \be 
S \ =\ s(x_L)  + \sqrt{2} \,\theta \psi(x_L) + \theta^2 F(x_L), \nn 
T \ =\ t(x_L)  + \sqrt{2} \,\theta \xi(x_L) + \theta^2 G(x_L)
 \ee
of opposite electric charges (with $x^\mu_L = x^\mu - i\theta \sigma^\mu \bar \theta$). It reads
 \begin{empheq}[box=\fcolorbox{black}{white}]{align}
\lb{L-SQED}
{\cal L}_{SQED} \ =\ \frac 14 \!\!\int d^2\theta d^2 \bar \theta (\overline S e^{eV} S + \overline T e^{-eV} T) + \nn
\left( \frac 18 \int  d^2\theta \,W^\alpha W_\alpha   + \frac m2 \int \!\! d^2\theta \,ST \ +\ {\rm c.c.} \right),
 \end{empheq}
where 
\be
\lb{W-alpha}
W_\alpha \ =\ \frac 18 \bar D \bar D D_\alpha V, \qquad \overline W_{\dot\alpha} \ =\  \frac 18 D  D \bar D_{\dot\alpha} V 
 \ee
with
\be
&&D_\alpha \ =\ \frac \pd{\pd \theta^\alpha} - i (\sigma_\mu)_{\alpha \dot \alpha} \bar \theta^{\dot\alpha} \pd_\mu, \nn 
&&\bar D_{\dot\alpha} \ =\ -\frac \pd{\pd \theta_{\dot\alpha}} + i \theta^\alpha (\sigma_\mu)_{\alpha \dot \alpha} \pd_\mu \,.
\ee
$m$ is the mass of the matter fields and $\pm e$ are their electric charges.
In the Wess-Zumino gauge,
 \be
\lb{WZ-gauge}
V = - 2\theta \sigma^\mu \bar \theta \, A_\mu(x) + 2i (\theta \theta) \bar \theta \bar \lambda(x) -
2i (\bar\theta \bar\theta) \theta \lambda(x) +  (\theta \theta) (\bar\theta \bar\theta) D(x)
  \ee
and
\be
\lb{W-alpha-WZ}
W_\alpha \ =\  i\lambda_\alpha(x_L) + i \theta^\beta F_{\alpha\beta}(x_L)- \theta_\alpha D(x_L) + \theta^2 [\sigma^\mu \partial_\mu \bar \lambda(x_L)]_\alpha\,.  
\ee
  Integrating \p{L-SQED} over $d^2\theta$ and $d^2 \bar \theta$ with the convention
\be \int d^2\theta \, \theta^2 \ =\ \int d^2\bar\theta \, \bar\theta^2 \ =\ 2
\ee
and excluding   the auxiliary fields $F,F^*, G, G^*$ and $ D$, we arrive at the component Lagrangian:
\be
\lb{L-SQED-comp}
&&{\cal L} \ =\ - \frac 14 F_{\mu\nu}^2 +  i \lambda \sigma^\mu \partial_\mu \bar \lambda + i\psi \sigma^\mu (\partial_\mu + ie A_\mu ) \bar \psi +
i\xi \sigma^\mu (\partial_\mu - ie A_\mu ) \bar \xi +
\nn
&& |(\partial_\mu - ie A_\mu)  s|^2 \, + \,  |(\partial_\mu + ie A_\mu)  t|^2  
-m^2( s^* s +  t^* t) - m(\psi \xi + \bar \psi \bar \xi) \nn
&& +  ie\sqrt{2}[ (\bar \xi \bar\lambda) t - (\bar \psi \bar \lambda) s - (\xi \lambda) t^*  + (\psi \lambda) s^* ] - \frac {e^2}2 (s^*  s - t^* t)^2\,.
\ee

The physical fields in this theory are the gauge potentials [their dynamical components are $A_{j=1,2,3}(\vecg{x}, t)$],  photino $\lambda_\alpha$,  matter scalars $s,t$ and matter fermions $\psi_\alpha, \chi_\alpha$. 
This Lagrangian was first derived in the pioneer paper \cite{Golfand}. The first two lines in \p{L-SQED-comp} describe interactions of the massive matter fields with the gauge field $A_\mu$, while  the third line (it is indispensible to make the Lagrangian  supersymmetric) includes the Yukawa terms and the quartic scalar potential. 
Note that this Lagrangian is invariant under a discrete (charge conjugation) symmetry:
 \be
 \lb{C-conj}
{\rm C}: \qquad A \to -A, \, \ \lambda \to -\lambda, \, \ \psi \leftrightarrow \xi, \,\  s \leftrightarrow t\,.
 \ee
Now, we deform the theory puttting  it  in a finite spatial box of size $L$, imposing on all the fields the periodic boundary conditions,
\be
\lb{period}
A_j(x+L, y,z) &=&  A_j(x,y+L,z)  =  A_j(x, y,z+L) =  A_j(x, y,z), \nn
s(x+L, y,z) &=&  s(x,y+L,z)  =  s(x, y,z+L) =  s(x, y,z),
\ee
etc., and expanding them into the Fourier series. The characteristic excitation energy of all  nonzero Fourier harmonics is then of order $1/L$. For the zero Fourier modes of the matter fields, it is of order of their mass $m$, which we assume to be large enough (at the level  $\sim 1/L$ or larger). We assume also the coupling constant $e^2$ to be small. Under these conditions, the low-energy dynamics of the theory is determined by  the  zero Fourier modes of the gauge field and its fermion superpartner.

Note that, at finite $L$, the constant potentials $A_j^{(0)}$ cannot be completely gauged away, as they can in the infinite volume. The admissible gauge transformations,
\be
\lb{large-gauge-QED}
\!\!\!\!\!\!\!\! A_j \to A_j + \partial_j \chi(\vecg{x}),  \quad (s, \psi) \to (s,\psi)e^{-ie \chi(\vecind{x})}, \quad (t,\xi) \to (t,\xi)e^{ie \chi(\vecind{x})}, 
\ee
should respect the periodicity conditions \p{period}. This implies\footnote{Implies in a certain sense --- see the discussion in   Sec.~2.1.} that the admissible domain of $A_j^{(0)}$ is the dual torus, $A_j^{(0)} \in 
[0, \frac {2\pi}{eL})$. \lb{domain} As we will shortly see, the characteristic excitation energy for this mode is 
 \be
\lb{low-En-scale}
E_{\rm char} \sim e^2/L\,,
 \ee
which is small compared to $1/L$ and $m$.
  
The low-energy dynamics of the system is then described by the effective Lagrangian including only three bosonic variables $c_j = A_j^{(0)}$ and two Grassmann variables $\eta_\alpha = 
\lambda_\alpha^{(0)}$. 
In the physical slang (which goes back to the classic paper \cite{BO} devoted to  the spectrum of the hydrogen molecule --- the positions of the electrons  were treated there as {\it fast} variables and the position of the protons as {\it slow} ones), the  variables $c_j$ and $\eta_\alpha$ are slow and all other variables are fast to be ``integrated out".

The effective Lagrangian  has a very simple form:\footnote{Just delete in \p{L-SQED-comp} all the terms including the matter fields, neglect nonzero modes of $A_j$ and $\lambda_\alpha$ and multiply the result by the spatial volume $V = L^3$ (the $L$ in the RHS is not the Lagrangian, of course).}
  \be
L^{\rm eff} \ = \ L^3 \left(\frac {\dot c_j^2}2  + i \eta \dot {\bar \eta} \right)\,.
 \ee
The corresponding quantum effective Hamiltonian, 
 \begin{empheq}[box=\fcolorbox{black}{white}]{align}
\lb{Heff-SQED}
\hat H^{\rm eff} \ = \ \frac {\hat P_j^2}{2L^3} \ =\ -\frac 1{2L^3}\triangle_c\,,
 \end{empheq}
describes free motion over the dual torus of length $a = 2\pi/(eL)$. The characteristic momentum (not the physical momentum!) for the low-energy states is therefore $P_{\rm char} \sim eL$. Bearing this in mind, we reproduce the estimate \p{low-En-scale}.

The Hamiltonian \p{Heff-SQED} is supersymmetric: it satisfies the algebra \p{SQM-alg} with the supercharges 
\be
\lb{superzar}
\hat Q_\alpha \propto (\sigma_j)_\alpha{}^\beta \eta_\beta \hat P_j, \qquad  \hat {\bar Q}^\alpha \propto (\sigma_j)_\beta{}^\alpha \bar \eta^\beta \hat P_j \,.
 \ee
 It acts on the wave functions $\Psi(c_j, \eta_\alpha)$. 
The full spectrum includes four sectors:
\be
\lb{sectors}
\!\!\!\!\!\Psi_0(\vecg{c}, \eta_\alpha) = f_0(\vecg{c}), \qquad \Psi_{1,2}(\vecg{c}, \eta_\alpha) = \eta_{1,2} f_{1,2}(\vecg{c}), \qquad
\Psi_3(\vecg{c}, \eta_\alpha) = \eta_1\eta_2 f_3(\vecg{c}) \,.
 \ee
We  impose on all $f_k(\vecg{c})$ periodic  boundary conditions:\footnote{More general boundary conditions will be considered in   Sec. 2.1.}
   \be
\lb{bc-wave}
f_k (\vecg{c} + a\vecg{n}) = f_k (\vecg{c})
 \ee 
with  integer $\vecg{n}$. Then 
 \be
f_k(\vecg{c}) \ =\ \sum_{\vecind{m}} b_{k\vecind{m}} e^{ieL \vecind{c} \cdot{\vecind{m}}}
 \ee
with integer $\vecg{m}$. 
The spectrum splits into degenerate quartets with energies
 \be
E_{\vecind{m}} \ =\ \frac {e^2}{2L} \vecg{m}^2\,.
 \ee
It involves four zero-energy states:
\be
\Psi_0^{(0)} = 1, \qquad \Psi_1^{(0)} = \eta_1, \qquad \Psi_2^{(0)} = \eta_2, \qquad \Psi_3^{(0)} = \eta_1 \eta_2 \,.
 \ee
The states $\Psi_{0,3}^{(0)}$ are bosonic, while the states $\Psi_{1,2}^{(0)}$  are fermionic. It follows that the Witten index vanishes:
 \be
I_W \ =\ 2-2 = 0\,.
 \ee
Generically, this would not allow us to conclude anything about the vacuum structure of the original undeformed Hamiltonian. The supersymmetric quartet of lowest states  could in principle move up from zero. They do not do so, however, and there is a {\it special reason} for that.

As was mentioned, the Lagrangian \p{L-SQED-comp} has a discrete symmetry \p{C-conj}. This allows us to consider alongside with the supersymmetric partition function \p{SUSY-Z} an object 
  \be
\lb{IC}
I_C \ =\ {\rm Tr} \left\{\hat C (-1)^{\hat F} e^{-\beta \hat H}\right\}
 \ee
 where the states with negative $C$-parity enter the sum with an extra minus. For the effective Hamiltonian \p{Heff-SQED}, the symmetry \p{C-conj} boils down to
\be
\lb{C-eff}
\vecg{c} \to - \vecg{c}, \qquad \eta_\alpha \to -\eta_\alpha\,.
 \ee
The states  $\Psi_{0,3}^{(0)}$ are even under this transformation, whereas the states  $\Psi_{1,2}^{(0)}$ are odd. As a result, the contribution of the vacuum states in $I_C$ is $2 + 2 = 4 \neq 0$.

And the excited states of the Hamiltonian do not contribute there. The reason is simple: the supercharges \p{superzar} commute with $\hat C$ and hence all the states in a quartet including the states that are related  by the action of the supercharges have the same $C$-parity and do not contribute in  \p{IC}.

We can now proceed with the same reasoning as for the ordinary Witten index. Given that only the vacuum states contribute in $I_C$, its value, being integer, cannot depend on smooth deformations, and we may say that $I_C = 4$ also for the original undeformed system. And this implies that the latter has four supersymmetric zero-energy vacuum states and supersymmetry is not broken. 

The following important remark is, however, in order. The validity of this calculation for the original field theory may look not so obvious because the result \p{Heff-SQED} for the effective Hamiltonian is valid only in the {\it leading} Born-Oppenheimer approximation. There are perturbative corrections which become large in the regions close to the corner of the box where  the dimensionless parameter
 \be
\lb{kappa-BO}
\kappa  \ = \ \frac 1{e|L\vecg{c} |^3}
 \ee
ceases to be small. These corrections 
have been evaluated  \cite{BOcorners} in the first order in $\kappa$ both for SQED and for the pure SYM theory to be discussed in the next section. They do not bring about an effective potential, but modify the kinetic term in \p{Heff-SQED} to endow it with a nontrivial metric. The full supersymmetric effective Hamiltonian acquires
the form
\be 
\lb{Heff-corrections}
\hat  H^{\rm eff} \ = \ -\frac 12 f \stackrel{\longrightarrow}{\partial^2_k} f + i \varepsilon_{kpl} \, \hat{\bar \eta} \sigma_l \eta \,f (\pd_p f) \pd_k + \frac 16 f (\pd_k^2 f)\, 
 (\hat{\bar\eta} \sigma_l \eta)^2,
\ee
where we have set $L = 1$, $\hat {\bar \eta}^\alpha = \pd/\pd \eta_\alpha$. The derivatives $\pd_k$  in the first term as well as $\pd/\pd \eta_\alpha$  act on whatever they find on the right, and $f(\vecg{c}) = 1 - \kappa/4 + o(\kappa)$. 
 One can observe that the account of the corrections does not change the estimate for the index.

Furthermore, one can argue that the results are not changed in any order in $\kappa$. Indeed, though we do not know the behaviour of the vacuum wave functions at the vicinity of the corner, this region represents a small fraction $\sim e^2$ of the full configuration space, and we can evaluate quite reliably the form of the wave functions in its main part, which are just constants to the leading order, and they stay close to constants when the corrections in $\kappa$ are taken into account. 
One can say that the BO corrections {\it deform} the effective Hamiltonian, and the Witten index must stay invariant under such deformation.

\subsection{Universes}

There is a subtlety which we have ignored up to now in our analysis. 
On p. \pageref{domain} we wrote that 
the admissible domain of $A_j^{(0)}$ is 
\be
\lb{dual}
A_j^{(0)} \in 
\left[0, \frac {2\pi}{eL}\right)\,.
\ee 
 In a certain sense, it is correct. Indeed, the field configurations $A_j$ and $A_j +  2\pi n_j /(eL)$ are related by  admissible gauge transformations that respect the periodicity condition \p{period}. If we require for the wave function to be invariant under gauge transformations, the conditions \p{bc-wave} follow and the statement above  is correct without reservations.
However, strictly speaking, the Gauss law constraints, which should be imposed on wave functions, only imply that the latter are invariant under {\it infinitesimal} gauge transformations and hence under topologically trivial finite gauge transformations that can be reduced to trivial ones by a set of continuous deformations. And a gauge transformation of the type
 \be
\lb{omega-nj}
\hat \omega(n_j): && A_j \to A_j + \frac {2\pi}{eL}n_j, \nn
&& (s, \psi) \to e^{-2\pi i \vecind{n}\cdot \vecind{x} /L}(s, \psi), \qquad (t, \xi) \to e^{2\pi i  \vecind{n} \cdot \vecind{x} /L}(t, \xi) \,,
 \ee
which can bring any $U(1)$ gauge field to a field whose zero Fourier mode lies in the dual torus \p{dual}, is not topologically trivial   but represents an uncontractible loop. 

In such a case, we may impose  generalized constraints:
\be
\lb{sdvig}
\hat \omega(n_j) \Psi  \ =\ e^{ i  \vecind{\vartheta}\cdot \vecind{n}} \Psi, \qquad \qquad \vartheta_j \in [0,2\pi)\,.
 \ee
The wave functions in a {\it universe}\footnote{A more traditional term is {\it superselection sector}, but recently the name ``universe" has gone into use  (see e.g. Ref. \cite{Zohar}) and we are following the crowd.}  with given $\vartheta_j$  do not ``talk" to the functions in some other universe in a sense that the result of the action of a local physically relevant operator  on a function   from a given  universe belongs to the same universe.

It follows from \p{omega-nj}, \p{sdvig} that the effective wave functions $\Psi^{\rm eff}(c_j)$ in the universe $\vecg{\vartheta}$ satisfy the constraints
\be
\lb{cons-shift}
\Psi^{\rm eff}_{\vecind{\vartheta}} \left( c_j + \frac {2\pi}{eL}n_j \right) \ =\ e^{ i \vecind{\vartheta} \cdot \vecind{n}} \, \Psi^{\rm eff}_{\vecind{\vartheta}} (c_j) \,.
 \ee
Then any such function has the form
 \be
\Psi^{\rm eff}_{\vecind{\vartheta}} (c_j)  \ =\ e^{ie L \vecind{\vartheta} \cdot \vecind{c}/(2\pi)} \, \Psi^{\rm eff}_{\vecind{0}} (c_j)\,,
  \ee
where $ \Psi^{\rm eff}_{\vecind{0}} (c_j)$ is periodic on the dual torus. The ground states in the spectrum of the effective Hamiltonian \p{Heff-SQED} have the energy
\be
\lb{En-SQED-vartheta}
E_0(\vecg{\vartheta}) \ =\ \frac {e^2 \vecg{\vartheta}^2}{8\pi^2L}\,.
 \ee
Thus, the true supersymmetric vacua dwell in the universe $\vecg{\vartheta} = 0$. Their wave functions are invariant under all gauge transformations --- topologically trivial or not. The results $I_W = 0$ and $I_C = {\rm Tr}\{\hat C(-1)^F e^{-\beta \hat H}\} = 4$  for the index remain intact. 

Looking at \p{En-SQED-vartheta}, a reader may think that the system has continuous spectrum, recall that we said on several occasions that the notion of index is ill-defined in such systems, and get confused.

However, the spectrum becomes continuous only if the spectra of all the universes are brought together in a common heap. One should not do that. A correct interpretation 
of Eqs. \p{cons-shift} -- \p{En-SQED-vartheta} is the following:

\vspace{1mm}

{\it When we put SQED in a finite box and impose periodic boundary conditions, infinitely many quantum systems corresponding to infinitely many universes can be defined.
In all the universes with nonzero $\vecg{\vartheta}$, supersymmetry is spontaneously broken (indeed, in that case, the ground states are not annihilated by the action of 
the supercharges). In the universe $\vecg{\vartheta} = 0$, supersymmetry is not broken
 and the spectrum includes two bosonic and two fermionic zero-energy states.}

\subsection{Fayet-Illiopoulos model}

Let us modify the model \p{L-SQED} by adding to the Lagrangian the extra term \cite{FI}
 \be
\lb{FI} 
{\cal L}_{FI} \ =\ \frac \xi 4 \int d^4\theta \,V\,.
 \ee
This model is still  gauge-invariant and supersymmetric, but supersymmetry is {\it spontaneously broken} in this case.
Indeed, excluding the auxiliary fields, we arrive at the following expression for the classical potential:
 \be
V(s,s^*,t,t^*) \ =\ m^2(s^* s + t^* t) + \frac {[\xi + e(s^* s - t^* t)]^2}2 \,.
 \ee
If $\xi \neq 0$, this potential is positive definite, the vacuum energy is not zero, and supersymmetry is broken.

This conforms with the Witten index analysis. As we have seen, in the undeformed theory, the index is zero, but supersymmetry  is  still not broken because of the nonzero value of  the generalized index \p{IC}. However, the latter does not have much meaning in the deformed theory, because the FI term {\it breaks} the charge conjugation symmetry \p{C-conj}. We are only left with the ordinary Witten index, it is zero, and there is no wonder that supersymmetry  breaks down. 

\setcounter{equation}0
\section{Supersymmetric Yang-Mills theory: unitary  gauge groups}
We start our discussion of non-Abelian gauge theories with supersymmetric YM theories not including matter superfields and described by the superfield Lagrangian
 \be
\lb{SYM}
{\cal L}_{\rm SYM} \ =\  \frac 1{4} {\rm Tr} \int d^2\theta \, \hat W^\alpha \hat W_\alpha \ + \ {\rm c.c.} 
 \ee 
with 
\be
\hat W_\alpha \ =\ \frac 1{8g} \bar D^2 (e^{-g\hat V} D_\alpha e^{g\hat V} )
 \ee
($\hat V = V^a t^a$ and $\hat W_\alpha = W_\alpha^a t^a$). In components:
\begin{empheq}[box=\fcolorbox{black}{white}]{align}
\lb{SYM-comp}
{\cal L}_{\rm SYM} \ =\ - \frac 14  F^a_{\mu\nu}  F^a_{\mu\nu}  + i  \lambda^a \sigma^\mu {\cal D}_\mu {\bar \lambda^a} + \frac 
{D^a D^a}2 \,,
 \end{empheq}
 where $ F_{\mu\nu}^a = \pd_\mu A_\nu^a - \pd_\nu A_\mu^a + g f^{abc} A_\mu^b A_\nu^c$ and
 ${\cal D}_\mu \bar\lambda^a =  \partial_\mu \bar \lambda^a + g f^{abc} A^b_\mu \bar \lambda^c$. 
$D^a$ are the auxiliary fields that are equal to zero due to equations of motion.

 We also assume in this section that the gauge group ${\cal G}$ is $SU(N)$. More complicated systems based on other classical Lie groups will be considered later. 

To regularize the system in the infrared and make the spectrum of the Hamiltonian discrete, 
we proceed as before and put the system in a finite  spatial box. One can then impose periodic boundary conditions like in \p{period}, but in the non-Abelian case, one can also impose the so-called {\it twisted} boundary conditions to be discussed in the second part of the section. We say right away that the two
ways of handling the system bring about the same physical answer: $I_W  = N$ and supersymmetry is {\it not} broken.  

\subsection{Periodic boundary conditions}
We take the Lagrangian \p{SYM-comp}, eliminate the auxiliary field $D^a$  using the equation of motion $D^a =0$, impose the boundary conditions
\be
\lb{period-nab}
\!\!\!\! A_j^a(x+L, y,z) &=&  A_j^a(x,y+L,z)  =  A_j^a(x, y,z+L) =  A_j^a(x, y,z), \nn
\!\!\!\!  \lambda_\alpha^a(x+L, y,z) &=&  \lambda_\alpha^a(x,y+L,z)  =  \lambda_\alpha^a(x, y,z+L) =  \lambda_\alpha^a(x, y,z),
\ee
and expand $A_j^a(\vecg{x})$ and  $\lambda_\alpha^a(\vecg{x})$ in the Fourier series.

In the Abelian case, the vacuum dynamics was associated with the constant gauge field modes, because the configurations $A_j(x,y,z) = c_j$ (and only them) corresponded to the zero field strength and zero vacuum energy. One could interprete the motion over the dual torus discussed in the previous section as the motion on the {\it moduli space} of classical vacua. 

What are classical vacuum field configurations in the non-Abelian case? Locally, any configuration $\hat A_j(\vecg{x})$ with zero field strength (a {\it flat connection} in the mathematical language) can be represented as
\be
\lb{flat-A}
\hat A_j = -\frac ig (\pd_j U) U^{-1}\,,
 \ee
where $U(\vecg{x})$ is an element of the gauge group $\cal G$. To find a global solution to the equation 
\p{flat-A}, we pick out a particular point 
in our cube, say, the vertex $(0,0,0)$ and define a set of 
holonomies (Wilson loops along nontrivial cycles of the torus)
 \be
 \label{hol}
\Omega_1 \ =\ P \exp \left\{ i g \int_0^L \hat A_1(x,0,0) \, dx \right\},
\nonumber \\ 
\Omega_2 \ =\ P \exp \left\{ i g \int_0^L \hat A_2(0,y,0) \, dy \right\},
\\ 
\Omega_3 \ =\ P \exp \left\{ i g \int_0^L \hat A_3(0,0,z) \, dz \right\},
\nonumber 
  \ee
where $P\exp\{\cdots\}$ is a path-ordered exponential --- a product of an infinite number of factors:
 \be
\lb{Omega-prod}
\Omega_1 \ =\ \lim_{n \to \infty} \prod_{l=0}^{n-1} \left[\mathbb{1} + i g \frac Ln \hat A_1 \left(\frac {lL}n,0,0\right)  \right]
\ee
and similarly for $\Omega_2$ and $\Omega_3$.

Let us prove a simple theorem (we follow Ref. \cite{KRS}).

\begin{thm}
For a periodic flat connection, the holonomies \p{hol} commute. Inversely,
 for any set of commuting matrices $\Omega_{1,2,3}$ belonging to a  connected, simply connected group $\cal G$, a periodic flat 
connection on the 3-torus exists such that $\Omega_j$ are the holonomies \p{hol}.  
\end{thm}
\begin{proof}
The proof of the direct theorem is simple.
Infinitesimally,
$$
\mathbb{1} + \frac ig A_j(\vecg{x}) dx^j  \ =\  \mathbb{1} + U^{-1}(\vecg{x}) \,\pd_j U(\vecg{x}) \,dx^j $$
and hence $ U(x+dx, y, z) \ =\ U(\vecg{x}) \left[\mathbb{1} + \frac ig A_j(\vecg{x}) dx^j \right] $, etc. 
 Multiplying over all the factors in Eq. \p{Omega-prod}, we deduce 
 that the function $U(x,y,z)$ 
 satisfies the boundary conditions. 
\be
 \label{bcU}
U(x+L,y,z) \ = U(x,y,z)  \Omega_1\,, \nn
U(x,y+L,z) \ = U(x,y,z)  \Omega_2\,, \nn
U(x,y,z+L) \ = U(x,y,z)  \Omega_3 
 \ee
with constant commuting $\Omega_i$  (commutativity of $\Omega_i$ is 
{\it necessary} for the matrix $U$ to be uniquely defined). 

\vspace{1mm}

 To prove the second part of the theorem, we have to construct 
the matrix $ U(x,y,z)$ satisfying the boundary conditions \p{bcU} for an arbitrary commuting triple $(\Omega_1, \Omega_2, \Omega_3)$. We do so in several steps.

If one chooses $U(0,0,0) = 1$, the matrices 
 $\Omega_j$ are the holonomies (\ref{hol}).
 We construct 
the matrix $ U(x,y,z)$ in several steps. 
 \begin{itemize}
\item
At the first step, we define
 \be
U(x,0,0) \ =\ \exp\left\{i T_1 \frac xL \right\}, \nonumber \\
U(0,y,0) \ =\ \exp\left\{i T_2 \frac yL \right\}, \\
U(0,0,z) \ =\ \exp\left\{i T_3 \frac zL \right\}, \nonumber 
  \ee
where $\Omega_j = \exp\{i T_j\}$. (The choice of $T_j$ once $\Omega_j$ 
are given is not unique, but it is irrelevant. Take {\it some} set of the 
logarithms of the holonomies $\Omega_i$.) Having done this, we can extend the 
construction over all other edges of the 3-cube so that the boundary 
conditions (\ref{bcU}) are fulfilled. For example, we define
 $$
U(L,y,0) \ =\  \exp\left\{i T_2 \frac yL \right\} \Omega_1 ,\ \ \ 
U(x,L,0) \ =\  \exp\left\{i T_1 \frac xL \right\}  \Omega_2
 $$
[then $U(L,L,0) = \Omega_2 \Omega_1 =  \Omega_1 \Omega_2$], etc.
\item With $U(x,y,z)$ defined on the edges of the cube in hand, we can 
continue $U$ also to the {\it faces} of the cube due to the fact that, 
according to our assumption, $\pi_1({\cal G}) = 0$ i.e. any loop in the group is
contractible. Let us do that first for 3 faces adjacent to the vertex (0,0,0).
\item With $U(x,y,0),\ U(x,0,z)$, and $U(0,y,z)$ in hand, we can find 
$U(x,y,z)$ on the other 3 faces of the cube:
\be 
&&U(x,y,L)  =  U(x,y,0)\Omega_3 ,\qquad U(x,L,z)  =  U(x,0,z) \Omega_2, \nn 
&&U(L,y,z)  =  U(0,y,z)  \Omega_1\,.
 \ee

\item With $U(x,y,z)$ defined on the surface of the cube, we can continue 
it into the interior using the fact that  $\pi_2({\cal G}) = 0$ for all simple Lie groups.
 \end{itemize}

By construction, $U(x,y,z)$ satisfies the boundary conditions (\ref{bcU})
and hence $A_j(x,y,z)$ are periodic.
\end{proof}
The skeleton construction above  is rather common in homotopy theory and
can be found also in physical literature (see e.g. Ref. \cite{Baal}).

Now, for the unitary and also for the symplectic groups (but not for higher ortogonal and exceptional groups --- see the next section), three commuting group elements  always belong to a {\it maximal torus} 
such that we can represent $\Omega_j = e^{i T_j}$ with commuting $T_j$ that belong to a Cartan subalgebra. 
The function $U(x,y,z)$ satisfying the boundary conditions \p{bcU} may then  be chosen as
\be
\lb{U-for-flat}
 U(x,y,z) \ =\ \exp\left\{ i T_1 \frac xL\right\} \exp\left\{ i T_2 \frac yL\right\} \exp\left\{ i T_3 \frac zL\right\}.
\ee
 According to \p{flat-A}, this gives constant  gauge potentials $\hat A_j^{(0)} = T_j/(gL)$. 

By conjugation we bring the Cartan subalgebra $\mathfrak{C}$, where $T_j$ belong, to a convenient form $\mathfrak{C}_0$. For $SU(N)$, it is a set of  diagonal matrices  of order $N$ with zero trace. This is a vector space of dimension  $r = N-1$. The standard basis in this space are the matrices
$\{t^3, t^8, \ldots, t^{N^2-1}\}$ satisfying  the standard normalization  Tr$\{t^a t^b\} = \delta^{ab}/2$.
The moduli space of classical vacua is parameterized by 
 $N-1$ three-vectors: 
\be
\lb{Cartan}
\vecg{c}^1 = \vecg{A}^{3(0)},  \quad \vecg{c}^2 = \vecg{A}^{8(0)},  \ldots  
 \ee
In a similar way as it was the case in the Abelian theory, the range of $c_j^a$ is finite. Indeed, two different sets of $c_j^{a = 1, \ldots, r}$ that correspond to the same holonomies $\Omega_j$ can be obtained from one another by a gauge transformation that respects the boundary conditions \p{bcU}. Besides, there are {\it Weyl reflections} acting on the three copies of the Cartan subalgebra. The admissible $c_j^a$ lie in the {\it Weyl alcove}.

Let us explain how it works in the simplest $SU(2)$ example. The slow variables are $c_j = A_j^{3(0)}$ and their superpartners $\eta_\alpha = \lambda_\alpha^{3(0)}$. The wave function must be invariant under the periodic gauge transformations,
 \be
\lb{shift}
t^3 c_j \ \to \omega t^3 c_j \omega^{-1} - \frac ig  (\pd_j \omega) \omega^{-1},  \qquad t^3 \eta_\alpha \ \to \omega t^3 \eta_\alpha \omega^{-1}
 \ee
with 
\be
\lb{shift2}
\omega(\vecg{x}) \ =\ \exp\left\{ \frac {4\pi i x_j n_j t^3}L  \right\}, \qquad \qquad n_j \in \mathbb{Z}\,.
 \ee
In contrast to the Abelian gauge transformations \p{omega-nj}, the transformation \p{shift}, \p{shift2} is topologically trivial --- the function \p{shift2} represents a map $S^1 \to SU(2)$ and $\pi_1[SU(2)] = 0$.

The Weyl reflection in this case is  simply
 \be
\lb{Weyl2}
c_j \ \to \ -c_j, \qquad \eta_\alpha \ \to \ -\eta_\alpha\,.
 \ee
Each such reflection is realized as a constant gauge transformation,
\be 
\omega = e^{i \pi t^1} = 2i t^1 \qquad {\rm or} \qquad \omega = e^{i \pi t^2} = 2i t^2\,.
 \ee
The wave function must be invariant under any of these transformations.
Hence it is completely defined by its values in the following Weyl alcove:
\be
0 \leq c_j \leq \frac {2\pi}{gL}\,.
 \ee
For $SU(N)$, the wave functions are invariant under the shifts of $c_j^a$ similar to \p{shift}, \p{shift2} with a change
\be
\lb{shiftN}
\omega(\vecg{x}) \ =\ \exp\left\{ \frac {4\pi i x_j n_j t^*}L\right\}\,,
 \ee
where  $t^*$ are the elements of the Cartan subalgebra corresponding to the  coroots of $su(N)$ and normalized in the same way as $t^3$. 
The wave functions are also  invariant under Weyl reflections ($N!$ of them) permuting the entries of the diagonal matrix representing an element of
$\mathfrak{C}_0$. The Weyl alcove for $SU(3)$ is depicted (in an ortonormal basis) in Fig. \ref{Weyl}. 

 \begin{figure} [ht!]
      \bc
    \includegraphics[width=0.5\textwidth]{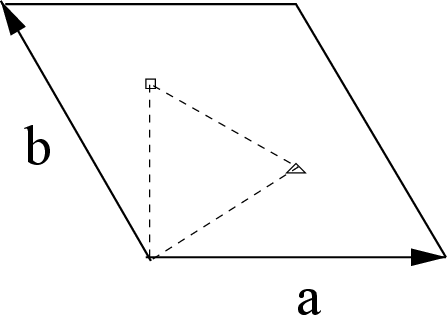}                  
     \ec
    \caption{The dashed triangle represents the Weyl alcove for $SU(3)$. Its area is 6 times smaller than the area of the solid rhombus formed  by the elements of $\mathfrak{C}_0$,
${a} = 2\pi/(gL)\, {\rm diag}(1,-1,0)$ and ${b} = 2\pi/(gL)\, {\rm diag}(0,1,-1)$, proportional to the simple coroots. The vertices of the alcove $\Box$ and $\triangle$ are proportional to the {\it fundamental coweights}: $\Box = \, 2\pi/(3gL){\rm diag}(1,1,-2)$, $\triangle = 2\pi/(3gL)\, {\rm diag}(2,-1,-1)$.}
 \label{Weyl}
    \end{figure}   

If the coupling constant $g$ is small, $\vecg{c}^{a = 1, \ldots, r}$  represent slow variables, which determine together with their superpartners $\eta_\alpha^a$ the vacuum dynamics  of our system.  In the lowest BO order and in the orthonormal basis, the effective Hamiltonian  has a very simple form:
 \begin{empheq}[box=\fcolorbox{black}{white}]{align}
\lb{Heff-SYM}
\hat H^{\rm eff} \ = \ - \frac 1{2L^3} \sum_{a=1}^r  \triangle_{c^a} \,,
 \end{empheq}
where $r = N-1$ is the rank of $SU(N)$.
The zero-energy eigenfunctions of \p{Heff-SYM} represent polynomials of $\eta^a_\alpha$ whose order does not exceed $2r$. However, not all such polynomials are admissible. The wave functions should be invariant under  Weyl reflections.
Imposing this requirement, we are left with only $N$ different vacuum functions:
 \be
\lb{wave-funct}
\Psi(\vecg{c^a}, \eta^a) \ =\ 1, \eta^a \eta^a, \ldots, (\eta^a \eta^a)^{N-1}\,.
 \ee
We derive 
 
 \begin{empheq}[box=\fcolorbox{black}{white}]{align}
I_W \ = \ N \,.
 \end{empheq}

For the SYM theory based on a symplectic gauge group $Sp(2r)$, a similar analysis gives the value of the index $I_W = r+1$.

\subsubsection{Instantons}

We continue the discussion of $SU(N)$.

 In Theorem 6.1, we proved the {\it existence} of  a periodic flat connection for any set of holonomies \p{hol}. But, obviously, there are infinitely many such connections related to one another by a gauge transformations
$$ U(x,y,z) \ \to \  U(x,y,z)   \tilde U(x,y,z)$$
 with periodic $ \tilde U(x,y,z)$. Due to trivial $\pi_1[SU(N)]$ and $\pi_2[SU(N)]$, we can continuously deform $\tilde U(x,y,z)$ so that $\tilde U = \mathbb{1}$ on the whole surface of the cube. But $\pi_3[SU(N)] = {\cal Z}$ is nontrivial. This means that topologically nontrivial $\tilde U(x,y,z)$ exist that cannot be continuously deformed to 
$\tilde U = \mathbb{1}$ in the interior of the cube. 

As a result, admissible $U(x,y,z)$ and the corresponding periodic $\hat A_j(x,y,z)$ are split into distinct topological classes characterized by a degree $q$ of mapping $S^3 \to SU(N)$. This degree may be expressed as the {\it Chern-Simons} topological invariant,
\begin{empheq}[box=\fcolorbox{black}{white}]{align}
\lb{CS}
q \ =\ \frac {g^2}{8\pi^2} \int_{\rm cube\ interior} {\rm Tr} \left\{A \wedge dA - \frac {2ig}3 A \wedge A \wedge A  \right\},
 \end{empheq}
where $A = \hat A_j dx^j$.

Let $\omega$ be a topologically nontrivial gauge transformation relating the trivial classical vacuum $\hat A_j = 0$ to a nontrivial one with $q=1$. Then, similar to what we had in the Abelian theory, the whole Hilbert space is split into the {\it universes}. All wave functions in a universe characterized by the {\it vacuum angle} $\vartheta \in [0,2\pi)$ satisfy the property
 \be
\lb{shift-q-integer}
\omega \Psi_\vartheta(A, \lambda) \ =\ e^{i\vartheta} \Psi_\vartheta(A,\lambda)\,.
 \ee 
There are two essential differences between the Abelian and non-Abelian universes.
\begin{enumerate}
\item The Abelian universes are defined only at finite volume --- they are related to nontrivial cycles wound on the torus. But non-Abelian universes are there also in the infinite volume --- for the YM theory living on $\mathbb{R}^3$ with the spatial infinity treated as one point (topologically, it is $S^3$). An example of a  topologically nontrivial vacuum carrying Chern-Simons charge $q=1$ in the $SU(2)$ theory on $\mathbb{R}^3$ is 
\be
\lb{vac-q}
\hat A_j \ =\ \frac 1g \exp \left\{i \pi  f(r) \sigma_j x_j  \right\}
\ee  
with $f(0) = 1$ and lim$_{r \to \infty} \, rf(r) = 1$ ($\sigma_j$ are the Pauli matrices).

The tunneling trajectories between the vacua of different charges $q$ are \cite{Jackiw-Rebbi} the  {\it instantons} --- topologically nontrivial solutions of the Euclidean Yang-Mills equations of motion\cite{BPST}.\footnote{
In a certain gauge, the instanton interpolating between the trivial vacuum and a vacuum  with unit Chern-Simons charge
has the form 
 \be
\lb{inst}
A_\mu^a \ =\ \frac 1g \frac {2 \eta^a_{\mu\nu}x_\nu}{x^2 + \rho^2}\,,
  \ee
where $\rho$ (the size of the instanton) is an arbitrary real number and $\eta^a_{\mu\nu}$ are the so-called `t Hooft symbols:
\be
\eta^a_{00} = 0, \quad \quad \eta^a_{jk} = \varepsilon_{ajk}, \quad \quad \eta^a_{j0} = -\eta^a_{0j} = \delta_{aj}\,.
\ee
The configuration \p{inst}   carries a unit {\it Pontryagin number}:
\be
\lb{Pontr}
q^{\rm Pontr} \ =\ \frac {g^2}{8\pi^2} \int_{\mathbb{R}^4} {\rm Tr}\, \{F\wedge F\}  \ =\ 1
 \ee 
with
$F = \frac 12 \hat F_{\mu\nu} dx^\mu \wedge dx^\nu$.}

\item In the Abelian case, the supersymmetric vacua dwell in the universe $\vecg{\vartheta} = 0$. In other universes, only the states with positive energies are present.
Similarly, in the  pure Yang-Mills theory or in QCD with massive quarks, the energies of the states in the universe with nonzero $\vartheta$ exceed by a fixed amount $\Delta E \sim c(1 - \cos \vartheta)$ the energies of the analogous states in the universe $\vartheta = 0$.  But in any theory with {\it massless} fermions including the SYM theory,  $c$ turns to zero due to an exact fermion zero mode in the instanton background \cite{Hooft}. As a result, the spectra in all the universes coincide. Supersymmetry is not broken in any of them. In a finite box with periodic boundary conditions, each such spectrum involves $N$ vacuum states. 
\end{enumerate}

\subsection{Twisted boundary conditions}

To define a precise quantum problem for a field theory placed on the torus, we need to specify the boundary conditions. In the analysis above, we imposed the periodic boundary conditions \p{period-nab}. This gave us the moduli space \p{Cartan} of classical vacua. The existence of $N$ quantum vacuum states was derived by studying the spectrum of the effective Hamiltonian that described the moduli space dynamics. This analysis was quite sufficient to arrive at the essential conclusion  that the supersymmetry is not spontaneously broken in the original infinite-volume SYM field theory, but it is interesting to explore what happens if  another way to regularize the theory in the infrared, another kind of boundary conditions are chosen. 

We set  for simplicity $L=1$ and consider the following boundary conditions \cite{twisted-Hooft}:
\be
\lb{bcA-twist}
\hat A_j(x+1,y,z) \ =\ P \hat A_j(x,y,z) P^{-1}, \nn
 \hat A_j(x,y+1,z) \ =\ Q \hat A_j(x,y,z) Q^{-1}, \nn
\hat A_j(x,y,z+1) \ =\ \hat A_j(x,y,z) \,,
 \ee
where $P,Q$ are constant $SU(N)$ matrices forming the so-called {\it Heisenberg pair} --- they satisfy the condition
\be
\lb{Heisen}
QP  \ = \ \epsilon PQ, 
 \ee
 with 
$\epsilon \mathbb{1}$ representing 
 an element of the group center. 
 For example, we can choose
 \be
  \label{PQ}
 P= \ \left( \begin{array}{lllcl}
1 & 0 &   0 & 0 & \ldots  \\
0 & \epsilon &  0 & 0 &  \ldots  \\
0 & 0 & \epsilon^{2} & 0 &   \ldots  \\
{\vdots} & {\vdots} & {\vdots} & {\vdots} & {\vdots}  \end{array} \right), \qquad 
Q= \ \left( \begin{array}{lllcl}
0 & 1 &   0 & 0 & \ldots  \\
0 &  0 &  1 & 0 &  \ldots  \\
0 & 0 & 0 & 1 &   \ldots  \\
{\vdots} & {\vdots} & {\vdots} & {\vdots} & {\vdots}\\
1 & 0 & 0 & 0 & {\ldots}  \end{array} \right).
  \ee
with $\epsilon = e^{2i\pi/N}$.

We are looking for the classical gauge field configurations with zero energy.  The boundary conditions  for $U(\vecg{x})$ that we should impose in order to respect the twisted boundary conditions \p{bcA-twist} for the flat gauge potentials \p{flat-A} have the form
\be
 \label{bcU-twist}
U(x+1,y,z) \ = \  P U(x,y,z) P^{-1} \epsilon^{k_x} \,, \nn
U(x,y+1,z) \ =   \  Q U(x,y,z) Q^{-1} \epsilon^{k_y}\,, \nn
U(x,y,z+1) \ = \ U(x,y,z)  \epsilon^{k_z}
\ee
with  integer $\vecg{k}$.   An essential difference with \p{bcU} is that the extra  factors $\epsilon^{k_{x,y,z}}$  that appear in  the RHS of Eq.\p{bcU-twist} are not arbitrary commuting $SU(N)$ matrices, as was the case in \p{bcU}, but belong to the {\it center} of the group.  As a result, we do not have a continuous moduli space of classical vacua anymore. The classical vacua now represent  {\it isolated points} (up to topologically trivial gauge transformations)
 in the field configuration space. 

Let $\vecg{k} = (1,0,0)$. In this case, the matrix $U$ satisfying \p{bcU-twist} can be chosen to be constant: $U = Q$. If $\vecg{k} = (0,1,0)$, we may choose $U = P^{-1}$. Such gauge transformations leave the classical vacuum $\hat A_j = 0$ intact and do not bring about anything new. But the transformations \p{bcU-twist} with 
$\vecg{k} = (0,0,1)$ give us something interesting.

An explicit solution for $U(x,y,z)$ to the equations
\be
\lb{bc-U}
\hat U(x+1,y,z) \ =\ P U(x,y,z) P^{-1}, \nn
  U(x,y+1,z) \ =\ Q U(x,y,z) Q^{-1}, \nn
U(x,y,z+1) \ =\ \epsilon U(x,y,z) 
 \ee
was found in Ref.  \cite{Selivanov}. It reads
\be
\label{ansatz}
U\ =\ e^{{2\pi i z} \hat T(x,y)},
\ee
where $\hat T(x,y)$ is the following Hermitian $su(N)$ matrix:
\begin{empheq}[box=\fcolorbox{black}{white}]{align}
\label{fund}
T_j^k(x,y)\ =\ \frac 1N \delta_j^k- \psi_j(x,y) \psi^{\dagger k}(x,y),
\end{empheq}
with $\psi_j$ being a complex unitary vector of the following form:
\be
\label{psi-j}
\!\!\!\!\! \psi_j(x,y) \ = \ \frac{
\exp\left\{-\pi \left( \frac {y+j-1}N \right)^2 + 2\pi ix \left(  \frac {y+j-1}N \right)\right\} \theta\left(x+ i \frac{y+j-1}{N} \right)}
{\sqrt{\sum_{k=1}^N  \exp\left\{-2\pi \left( \frac {y+k-1}N \right)^2 \right\} \left|\theta\left(x+ i \frac {y+k-1}N \right)\right|^2}  }\,. 
\ee
Here $\theta(w)$ is a  Jacobi $\theta$ function (see  Ref. \cite{Mumford} for an extensive review):

\be
\lb{Jacobi-theta}
\theta(w) \ =\ \sum_{n = -\infty}^\infty \exp\{- \pi n^2 + 2\pi i n w\}\,.
 \ee

Indeed, the function \p{Jacobi-theta}  satisfies the boundary conditions 
 \be
\theta(w+1) \ =\ \theta(w)\,, \nn
\theta(w+i) \ =\ e^{\pi(1-2iw)} \theta(w)\,.
\ee
Hence the vector \p{psi-j} satisfies the conditions
\be
\label{twistp}
\left\{ \begin{array}{c} \psi(x+1,y)\ =\ e^{2\pi i y/N} P \psi(x,y) \\ 
\!\!\!\!\!\!\!\!\!\!\!\!\!\!\psi(x,y+1)\ =\  Q \psi(x,y) \end{array} \right. .
\ee
It follows that the matrix $\hat T$  satisfies the conditions
\be
\label{twistT}
\left\{ \begin{array}{c} \hat T(x+1,y)\ =\ P \hat T(x,y) P^{-1} \\ 
\hat T(x,y+1)\ =\  Q \hat T(x,y) Q^{-1} \end{array} \right. ,
\ee
and the same for $\exp\{2\pi i z \hat T\}$. The appearance of Jacobi functions is natural in the problem with toroidal geometry.

Now note that the matrix \p{fund} can be brought by a conjugation to the form 
$$ \hat T \stackrel{\rm conj}\longrightarrow  \hat T_0 \ =\  \frac 1N  {\rm diag}(1,\ldots, 1, 1-N)\,.$$
This is achieved by rotating $\psi_j \to \psi_j^{(0)} = \delta_{jN}$. The matrix $\hat T_0$ is a {\it fundamental coweight} of $su(N)$, so that 
$$ \exp\{2\pi i \hat T_0\} =  \exp\{2\pi i \hat T\} = \epsilon \mathbb{1}\,.$$
It follows that the matrix \p{ansatz} satisfies all the conditions \p{bc-U}.

The flat connection \p{flat-A} based on \p{ansatz} has a nonzero Chern-Simons charge \p{CS}. An explicit calculation \cite{Selivanov} gives the value
\be
q \ =\ \frac 1N\,.
 \ee
Thus, this classical vacuum plays on the torus the same role as the topologically nontrivial classical vacua on $S^3$ with a unit Chern-Simons number, which we discussed at the end of the previous subsection. Tunneling Euclidean trajectories discovered by 't Hooft \cite{twisted-Hooft}, which relate the trivial vacuum $\hat A_j = 0$ to  nontrivial ones,    are called {\it torons}.  In contrast to ordinary instantons, their Pontryagin number is an integer multiple of $1/N$
and may be fractional.

 The toron dynamics can be treated in a similar way as we treated the instanton dynamics a couple of pages before. Let $ \hat {\tilde \omega}$ be an operator of a topologically nontrivial gauge transformation that shifts the Chern-Simons number by $1/N$. Then the spectrum is split into different universes. The wave functions in each such universe, represented  schematically as 
 \be
|\vartheta\rangle \ =\ \sum_q e^{iq\vartheta} |q \rangle
 \ee
with 
$$ q \ =\ \ldots, - \frac 1N, 0, \frac 1N, \ldots, $$
 satisfy the property
 \be
\hat {\tilde \omega} \Psi_\vartheta \ =\ e^{i\vartheta/N} \Psi_\vartheta\,.
 \ee
If only integer values of $q$ were admissible (as was the case for the problem with periodic boundary conditions), the parameter $\vartheta$ would belong to the interval $[0, 2\pi)$.  But if the fractional $q$ are also present, as it is the case for twisted boundary conditions, $\vartheta \in [0, 2\pi N)$. There is a single quantum vacuum state in each universe.

We have to warn the reader: the interpretation suggested above is not the standard one. Following \cite{Wit82}, people usually only write \p{shift-q-integer} with the gauge transformation shifting the Chern-Simons number by 1. Then there are less universes, $\vartheta \in [0,2\pi)$ rather than $\vartheta \in [0,2\pi N)$, but in each universe, there are $N$ distinct vacuum states. And this conforms with the counting $I_W = N$ derived above for periodic boundary conditions. 

There is a certain reason to do so. Putting theory in a finite box is a way to regularize it in infrared, but what interests us more is the dynamics of SYM field theory in the infinite volume. In the infinite volume, we know classical vacuum field configurations \p{vac-q} with an integer Chern-Simons number and the corresponding Euclidean instanton solutions \p{inst} with integer Pontryagin number. But we are not aware of similar nice invinite-volume configurations with fractional Chern-Simons or Pontryagin number.\footnote{Well, there are some configurations with fractional topological charge, the so-called {\it merons} \cite{merons}, but they are not so nice, being singular and having an infinite action. Their role in Yang-Mills dynamics is not clear.}

But let us separate two issues:
{\it (i)} the dynamics of the theory on a finite torus with twisted boundary conditions and {\it (ii)} the dynamics of the theory  in infinite volume.

As far as the twisted torus is concerned, I do not see any specifics of the fractional $q$ configurations compared to the configurations with integer $q$. The universe interpretation with $\vartheta \in [0,2\pi N)$ seems in this case to be the most natural.

And what really happens in the infinite volume limit is a difficult and not so clear\footnote{At least, not so clear to your author.} question.
First of all, the torons appear only if twisted boundary conditions are imposed. With periodic conditions, there is no trace of torons and the Chern-Simons number is always integer. And we believe that two ways of regularization should give the same dynamics in the infinite volume limit.

Further, we know that in SYM theory the {\it gluino condensate} $\langle \lambda^a_\alpha \lambda^{a\alpha} \rangle $ is formed \cite{gluino-cond}. 
It follows from the fact that the correlator 
\be
\lb{corr-gluino-SUN}
C(x_1,\ldots,x_N) \ =
\ \langle \lambda^{\alpha a} \lambda^a_\alpha(x_1) \ldots \langle \lambda^{\alpha a} \lambda^a_\alpha(x_N) \rangle
\ee 
evaluated in the instanton background is a real nonzero constant not depending on $x_k$ --- this is a corollary of certain supersymmetric Ward identities. 
 As this is true both at small separations $x_j - x_k$ (where the instanton calculation is reliable) and also at large separations, it implies by cluster decomposition that $\langle \lambda^a_\alpha \lambda^{a\alpha} \rangle$   
 is a nonzero complex constant whose phase may acquire $N$ complex values, $\phi_k = 2\pi k/N$. 

But does it mean, as people usually say, spontaneous breaking of discrete chiral symmetry $Z_N$ (after $U(1)$ chiral symmetry of the Lagrangian \p{SYM-comp} is broken explicitly by instanton effects)? Or do the condensates with different phases $\phi_k$ live in different universes?..

\subsection{Domain walls}

Physically, spontaneous breaking of a discrete symmetry means coexistence of different phases separated by domain walls in the same physical space. Unfortunately, SYM theory does not describe nature, and we cannot stage an experiment where these domain walls would or would not be observed. We have to find theoretical arguments pro or contra their existence.

One can recall in this regard a very confusing story about would-be spontaneous breaking of center symmetry in hot Yang-Mills theory. Many people claimed that such breaking associated with different phases of the Polyakov loop, 
\be
P(\vecg{x}) \ =\ \frac 1{N_c} {\rm Tr} \exp \{ ig \beta \hat A_0 (\vecg{x}) \}, \qquad \qquad \beta = 1/T,
\ee
really occurs. Even the domain walls separating the regions with different phases of $P(\vecg{x})$ were ostensibly found \cite{Pisarski}.

But it is not correct \cite{no-bubbles}. Physical spontaneous breaking of $Z_N$ does not occur in YM theory at high temperature. The Polyakov loop [in contrast to the correlator $\langle P(\vecg{x}) P^*(0) \rangle$] is not a physically observable quantity. It plays the same role in hot YM theory as the dual plaquette variables $\eta_j$ in the Ising model.\footnote{We have deviated from the main subject of the review, but I cannot resist to comment here that the deconfinement phase transition, if it takes place in hot YM theory (which, I think, it does) is associated not with this nonexistent spontaneous breaking, but probably with percolation of color fluxes \cite{Kiskis}.}

\vspace{1mm}

And how can one confirm or disprove the existence of domain walls between the regions with different phases of the gluino condensate in SYM theory?

 The only instrument that I am aware of is the Veneziano-Yankielowicz effective Lagrangian \cite{Ven-Yan}, which is formulated 
in terms of the composite colorless chiral superfield 
\be
\lb{S=WW}
S\  =\  -2{\rm Tr} \{ \hat W^\alpha \hat W_\alpha \}\,.
 \ee
The lowest component of $S$ is the bifermion operator $\lambda^{\alpha a} \lambda_\alpha^a$.  Its vacuum expectation value is the gluino condensate discussed
above.

 The VY Lagrangian has the form 
\begin{empheq}[box=\fcolorbox{black}{white}]{align}
\lb{VY}
{\cal L}^{VY} \ =\ \alpha \int d^4\theta\, (\overline{S} S)^{1/3} + \beta \left( \int d^2\theta \,S \ln \frac S{\Lambda^3} + {\rm h.c.}    \right),
 \end{empheq}
where $\alpha$ is an arbitrary numerical coefficient, the coefficient $\beta$ has some particular, irrelevant for our purposes numerical value, and $\Lambda$ is a dimensionful parameter of  quantum SYM theory --- the scale where its effective coupling becomes strong.

The SYM Lagrangian \p{SYM} is invariant under the scale transformations and also under the chiral transformations,
 \be
\lb{sup-chiral}
\hat W_\alpha \ \to e^{i\gamma} \hat W_\alpha , \qquad \qquad \theta_\alpha \ \to \ e^{i\gamma}\theta_\alpha \,.
 \ee
In supersymmetric theory, these two transformations have the same nature and the corresponding currents enter the same supermultiplet. These symmetries are, however, anomalous --- they are broken by quantum effects. The WY effective Lagrangian takes account of this anomaly. The first (kinetic) term in Eq.~\p{VY} is invariant under \p{sup-chiral}, but the second (potential) term is not. Its variation is proportional to 
$\int d^2 \theta \, S - \int d^2 \bar\theta \, \overline{S}$, and this is exactly what the anomalous Ward identity for chiral rotations gives.

The Lagrangian \p{VY} is similar to the Lagrangian 
\be
\lb{L-WZ}
{\cal L}_{\rm WZ} \ =\ \frac 14 \int d^2\theta d^2 \bar \theta\ \overline \Phi \Phi + \left(\frac 12 \int  d^2\theta \,{\cal W}(\Phi) + {\rm c.c.}    \right)
 \ee
 of the Wess-Zumino model \cite{WZ} (it is especially clear if one expresses it in terms of the superfield $\Phi = S^{1/3}$ of canonical dimension 1). The Witten index for the WZ model can be easily evaluated by the Cecotti-Girardello method. For the polynomial superpotential ${\cal W}(\Phi)$, the index is equal to the degree of the polynomial ${\cal W}'(\Phi)$. 

In our case, $${\cal W}'(\Phi) \sim \Phi^2 \ln \frac \Phi{\tilde\Lambda} \qquad \quad (\tilde \Lambda =  \Lambda e^{-1/3})$$
 is not a polynomial, but includes an extra logarithmic factor. It has an unphysical cut to be handled somehow, but its presence does not change the estimate for the index. The classical vacua are
\be
\phi = 0, \qquad  {\rm and} \qquad \phi = \tilde \Lambda\,.
\ee
The index is equal to 2. 

This does not conform with the evaluations above. The main pecularity is the presence  of the {\it chirally symmetric} vacuum \cite{Kov-Shif} $\langle \lambda^{a \alpha}\lambda^a_\alpha \rangle = \langle \phi \rangle^3 = 0$ on top of the chirally asymmetric one with
$\langle \lambda^{a \alpha}\lambda^a_\alpha \rangle =  \tilde\Lambda^3$.

Do not worry, this is not a paradox. The matter is that the VY Lagrangian is {\it not} derived following a correct  Born-Oppenheimer (people also call it {\it Wilsonian}) procedure,
where slow variables on which the effective Lagrangian depends are carefully separated from the fast ones to be integrated out. 

The VY Lagrangian describes {\it some} features of the low-energy dynamics of the SYM theory, but does not, unfortunately, describe some other features. In particular, it does not seem to describe well the vacuum structure of the theory.   In particular, it has only one asymmetric vacuum at $\phi = \tilde \Lambda$ and not $N$ vacua that the original theory [in the standard description with $\vartheta \in [0,2\pi)$] must have.

In Ref. \cite{Kov-Shif}, an amendment of  the Lagrangian \p{VY} was suggested. It includes the standard instanton 
 parameter $\vartheta_{\rm inst} \in [0,2\pi)$ and has a discrete $Z_N$ chiral symmetry. The effective potential includes $N$ sectors. There is a chirally asymmetric vacuum in each sector {\it and} the extra chirally symmetric minimum in the origin (common for all sectors)    --- there is no way to get rid of it in this approach. 

The  Kovner-Shifman effective potential for $SU(3)$  in the universe $\vartheta_{\rm inst} = 0$  is represented in Fig. \ref{vac-sect}. The dashed lines on the plane $s = \phi^3$ delimit 3 sectors. In the right sector, the potential is given by the VY expression,
\be
V_0(s, s^*) \ =\ (s s^*)^{2/3} \ln \frac s {\tilde \Lambda^3} \ln \frac  {s^*}{\tilde \Lambda^3}\,.
\ee
In the upper sector, the potential is
 \be
V_1(s, s^*) \ =\ V_0(e^{-2i\pi/3}s, \, e^{2i\pi/3}s^*)\,.
\ee
In the lower sector:
 \be
V_2(s, s^*) \ =\ V_0(e^{2i\pi/3}s, \, e^{-2i\pi/3}s^*)\,.
\ee
On the border  arg($s$) $ = \pi/3$, the potentials $V_0$ and $V_1$ coincide. At  arg($s$) $ = -\pi/3$, $V_0 = V_2$. And at arg($s$) $ = \pi$, $V_1 = V_2$. But the derivatives of the potential are not continuous.

This potential has 3 chirally asymmetric minima at $s = \tilde \Lambda^3$ and $s=  \tilde \Lambda^3 e^{\pm 2i\pi/3}$, which are marked by boxes, and the chirally symmetric minimum at $s =0$, which is marked by a blob.
 \begin{figure} [ht!]
      \bc
    \includegraphics[width=0.5\textwidth]{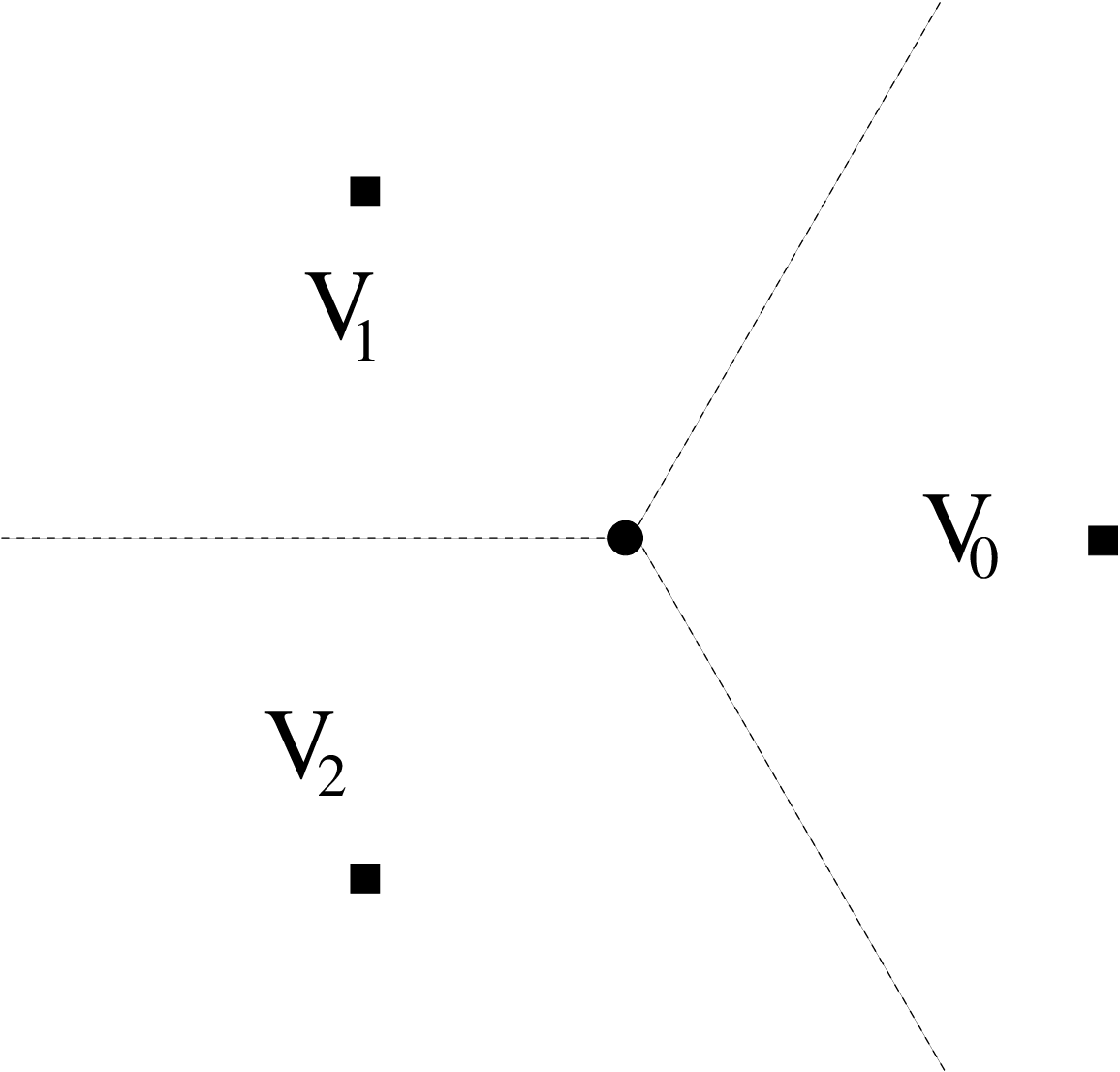}                  
     \ec
    \caption{Kovner-Shifman effective potential. $N=3$ and $\vartheta_{\rm inst} = 0$.}
 \label{vac-sect}
    \end{figure}   

The question that we originally posed was whether domain walls  between the different chirally asymmetric vacua exist. And the analysis in the framework of VY-KS effective potential gives a negative answer: there are domain walls connecting a box to the blob, but there are no walls between different boxes \cite{KSS}. 

For more detailed discussion of this walls vs. torons controversy, we address the reader to   Sec.~7 of Ref. \cite{KSS}. The question still remains open. But whatever is true: is discrete chiral $Z_N$ symmetry  in SYM theory with $SU(N)$ gauge group spontaneously broken or are configurations with fractional topological charge also relevant in infinite volume, one thing is clear: supersymmetry is {\it not} broken there and an exact supersymmetric vacuum (or several vacua) exist.

\section{SYM theories with other gauge groups}
\setcounter{equation}0

In the original paper \cite{Wit82}, the reasoning of   Sec.~3.1, which allowed one to evaluate the Witten index on the torus with periodic boundary conditions for SYM theory based on an $SU(N)$ group, was extended to all other simple gauge groups. The effective wave function on a small torus had the same form as in Eq. \p{wave-funct} with the only difference that the highest power of $\eta^a \eta^a$ was  equal to the rank $r$ of the group. This led to the result
 \be
\lb{r+1}
I_W \ =\ r+1 \,.
 \ee 
However, as was mentioned above, there is another way to evaluate the index. It is based on the standard scenario with spontaneous breaking of discrete chiral symmetry in the infinite volume and the vacuum angle restricted to the interval $\vartheta_{\rm inst} \in [0,2\pi)$.
For a generic group, the instanton admits $2h^\vee$ gluino zero modes, where  $h^\vee$ is the  {\it dual ``ve
 number}, which coincides with the 
Casimir eigenvalue $c_V = T^a T^a$ when a proper normalization for the adjoint generators $T^a$ is chosen. As was the case for $SU(N)$, supersymmetric Ward identities dictate to the Euclidean correlator
\be
\lb{corr-gluino-gen}
C(x_1,\dots,x_{h^\vee}) \ =\  \langle \lambda^a \lambda^a(x_1)\ldots   \lambda^a \lambda^a(x_{h^\vee}) \rangle    
\ee
 not to depend on the coordinates. Evaluating it
 in the instanton background, we obtain a nonzero real constant. And this implies that  
 $\langle \lambda^{a\alpha} \lambda^a_\alpha \rangle$ is a complex constant whose phase may acquire $h^\vee$ complex values, so that
\be
\lb{vee}
 I_W \ =\ h^\vee \,.
 \ee

For the  unitary and also for the symplectic groups, the estimates \p{r+1} and \p{vee} coincide. But for higher orthogonal groups $SO(N \geq 7)$ and all exceptional groups, they do not. 
For $SO(N \geq 7)$, $h^\lor = N-2  >  r+1$.
Also for exceptional groups $G_2, F_4, E_{6,7,8}$, the index (\ref{IW}) is 
larger  than Witten's original estimate (see Table 1).

\begin{table}
\caption{Vacuum counting for  exceptional groups}

\begin{center}
\begin{tabular}{||l|c|c|c|c|c||} \hline
group $\cal G$& $G_2$ & $F_4$ & $E_6$ & $E_7$ & $E_8$ \\ \hline
r + 1    & 3 & 5 & 7 & 8 & 9 \\ \hline
 $h^\lor$ & 4&  9 & 12 & 18 & 30 \\ \hline
mismatch & 1& 4 & 5 & 10 & 21 \\ \hline
\end{tabular}
\end{center}
\end{table}

The origin of this troublesome mismatch remained unclear until 1997 when Witten realized \cite{Witten-O7} that one of the assumptions on which the estimate \p{r+1} was based --- that {\it a triple of commuting holonomies $\Omega_j$ can always be brought by conjugation to the maximal torus} --- is correct for the unitary and symplectic groups, but is {\it wrong} for higher orthogonal and exceptional groups. 

To understand why, note first that, for an orthogonal group (which is not simply connected), even a {\it pair} of commuting group elements cannot always be conjugated to a maximal torus. Consider for example $SO(3)$. The elements 
\be
\lb{g-O7}
g_1 = {\rm diag}(-1,-1, 1) \qquad {\rm and} \qquad g_2 = {\rm diag}(-1, 1, -1)
\ee
 obviously commute, but their ``logarithms", proportional to the
generators of the correponding rotations $T^3$ and $T^2$, do not. 

However, for a simply connected group, a pair of commuting elements $\Omega_1, \Omega_2$ can always be conjugated to the maximal torus. This follows from the so-called Bott theorem: {\it a centralizer\footnote{I.e. the subgroup of ${\cal G}$ that commutes with $\Omega$.} ${\cal G}_\Omega$ of any element $\Omega$
 in a simply connected Lie group ${\cal G}$ is connected.} 
We
first put $\Omega_1$ on the torus of ${\cal  G}$ and then conjugate $\Omega_2$ to the torus of the centralizer ${\cal G}_{\Omega_1}$. The connectedness of the latter dictates that the two tori coincide.

One can  note in this regard that the preimages of the commuting in $SO(3)$ elements $g_1$ and $g_2$ in the universal  simply connected covering $SU(2)$ are  $\sigma^3$ and $\sigma^2$, and they do not commute,
 representing a  Heisenberg pair as in Eq. \p{Heisen}.

But exceptional commuting {\it triples} that cannot be conjugated to the maximal torus may exist also for simply connected groups.
To give the simplest example of such a triple, consider the group $Spin(7)$. An element of $Spin(7)$ can be represented as 
 \be
\lb{g-Spin7}
\Omega \in Spin(7) \ =\ \exp\{\omega_{\mu\nu} \gamma_\mu \gamma_\nu/2\}\,,
 \ee 
with skew-symmetric $\omega_{\mu\nu}$. Here $\gamma_\mu$ are the gamma matrices $8\times 8$ satisfying the Clifford algebra $\gamma_\mu \gamma_\nu + \gamma_\nu \gamma_\mu = 2 \delta_{\mu\nu}$, $\mu,\nu = 1,\ldots, 7$.

Consider the triple \cite{Witten-O7,KRS} 
 \begin{empheq}[box=\fcolorbox{black}{white}]{align}
\lb{triple}
\Omega_1 \ =\ \exp \left\{ \frac \pi 2 (\gamma_1 \gamma_2 + \gamma_3 \gamma_4) \right\} \ =\ \gamma_1 \gamma_2 \gamma_3 \gamma_4\,, \nn
\Omega_2 \ =\ \exp \left\{ \frac \pi 2 (\gamma_1 \gamma_2 + \gamma_5 \gamma_6) \right\} \ =\ \gamma_1 \gamma_2 \gamma_5 \gamma_6\,, \nn
\Omega_3 \ =\ \exp \left\{ \frac \pi 2 (\gamma_1 \gamma_3 + \gamma_5\gamma_7) \right\} \ =\ \gamma_1 \gamma_3 \gamma_5 \gamma_7\,.
 \end{empheq}
The corresponding triple in $SO(7)$ is 
\be
\label{triple-O7}
\tilde \Omega_1 \ =\ {\rm diag} ( -1,-1, -1, -1, 1, 1, 1 ), \nn
 \tilde \Omega_2 \ =\ {\rm diag} (-1,-1,1,1,-1,-1,1 ), \nn
 \tilde \Omega_3 \ =\ {\rm diag} (-1, 1, -1, 1, -1, 1, -1). 
  \ee
One can easily check that all $\Omega_j$  commute, while the exponents in the LHS of Eq.\p{triple} do {\it not}. The representation
\p{g-Spin7} is not unique, but one can prove the following simple theorem. 
\begin{thm}
A representation $\Omega_j = e^{iS_j}$ with $S_j \in spin(7)$ and  $[S_j, S_k] = 0$ does not exist. 
\end{thm}
\begin{proof}
Suppose, such a triple $S_j$ exists. The property $[S_j, \Omega_k] = 0$  should then also hold. Let $S = i \omega_{[\mu\nu]} \gamma_\mu \gamma_\nu$. The requirement $[S, \Omega_j] = 0$ gives
 \be
&&\omega_{15} T_{2345} + \omega_{16} T_{2346} + \omega_{17} T_{2347} - \omega_{25} T_{1345} -  
\omega_{26} T_{1346} - \omega_{27} T_{1347} \nn
+  && \omega_{35} T_{1245} +   \omega_{36} T_{1246} +  \omega_{37} T_{1247}    - \omega_{45} T_{1235} -  \omega_{46} T_{1236} -  \omega_{47} T_{1237}  \ =\ 0\,, \nn
&&\omega_{13} T_{2563} + \omega_{14} T_{2564} + \omega_{17} T_{2567} - \omega_{23} T_{1563} -  
\omega_{24} T_{1564} - \omega_{27} T_{1567} \nn
+ &&  \omega_{53} T_{1263} +   \omega_{54} T_{1264} +  \omega_{57} T_{1267}    - \omega_{63} T_{1253} -  \omega_{64} T_{1254} -  \omega_{67} T_{1257}  \ =\ 0\,, \nn
&&\omega_{12} T_{1352} + \omega_{14} T_{1354} + \omega_{16} T_{1356} - \omega_{32} T_{1572} - \omega_{34} T_{1574} -  
\omega_{36} T_{1576}  \nn 
+ &&\omega_{52} T_{1372} + \omega_{54} T_{1374} +   \omega_{56} T_{1376} -  \omega_{72} T_{1352}    - \omega_{74} T_{1354} -  \omega_{76} T_{1356}  \ =\ 0\,, \nn
 \ee
where $T_{2345} = \gamma_2 \gamma_3 \gamma_4 \gamma_5$  etc. It follows that all $\omega_{\mu\nu}$ vanish.
\end{proof}

The underlying reason of nontriviality of the triple \p{triple} is the fact that the centralizer of the element $\Omega_1$ (or $\Omega_2$
or $\Omega_3$) in $Spin(7)$ is not simply connected. And it is not just not simply connected, but not simply connected in not a simple way, if you will.

I mean the following. For a unitary or symplectic group, a centralizer of any element represents a product of a semi-simple simply connected group and some number of $U(1)$ factors. For example, the centralizer of an element of $SU(3)$ can be $SU(3)$, $SU(2)\times U(1)$ or $[U(1)]^2$. But the  centralizer of $\Omega_1$  in $Spin(7)$ has a more complicated structure.  $\Omega_1$  commutes with three generators $\gamma_5 \gamma_6$, $\gamma_5 \gamma_7$ and $\gamma_6 \gamma_7$ of the subgroup $Spin(3) \subset Spin(7)$ and also with six generators $\gamma_1 \gamma_2$, etc. of the subgroup   $Spin(4) \subset Spin(7)$. It may seem that the centralizer is $Spin(3) \times Spin(4)$, but it is not exactly so because the product $Spin(3) \times Spin(4)$ is not a subgroup of $Spin(7)$.

Indeed,
both $Spin(4)$ and $Spin(3)$ have non–trivial centers. The center of $Spin(3) \equiv SU(2)$
is $Z_2$ --- an  element corresponding to rotation by $2\pi$ in, say, (56) plane.
The center of $Spin(4) \equiv SU(2) \times SU(2)$ is $Z_2 \times Z_2$. It has a diagonal element $z$ representing
the product of non–trivial center elements in each $SU(2)$ factor:
 \be
\!\!\!\! z  =  \exp \left\{ \frac \pi 2 (\gamma_1 \gamma_2 - \gamma_3 \gamma_4) \right\} 
 \exp \left\{ \frac \pi 2 (\gamma_1 \gamma_2 + \gamma_3 \gamma_4) \right\} = \exp\{\pi \gamma_1 \gamma_2 \} \,. 
\ee
Now note that $\exp \{\pi \gamma_1 \gamma_2\} = \exp \{\pi \gamma_5 \gamma_6\}  =  - \mathbb{1}$ in the full $Spin(7)$ group. Hence the true centralizer of $\Omega_1$ 
 is $[Spin(4) \times Spin(3)]/Z_2 \equiv [SU(2)]^3/Z_2$ where the factorization is done
over the diagonal center element $(-\mathbb{1}, -\mathbb{1}, - \mathbb{1})$ in $[SU(2)]^3$.  The matrices $\Omega_2$ and $\Omega_3$ correspond to the elements $(i\sigma^2, i\sigma^2, i\sigma^2)$ and  $(i\sigma^3, i\sigma^3, i\sigma^3)$ , i.e. they
form a Heisenberg pair in each $SU(2)$ factor. They still commute in the centralizer and hence in
$Spin(7)$ just because the factorization over the common center should be done. 
The elements $\Omega_2,\,\Omega_3$ 
cannot be conjugated to the maximal torus in the centralizer and hence the whole triple \p{triple} cannot be conjugated to the maximal torus in $Spin(7)$.

Let us call the group elements whose centralizers involve a not simply connected semi-simple factor {\it exceptional}. Up to a conjugation, there is only one exceptional element in Spin(7). To understand a ``scientific" reason for the existence of such elements,  some group-theoretic study is required, in particular --- a study of the {\it Dynkin diagrams} of $Spin(7)$ and then of other orthogonal and exceptional groups.

\begin{figure}
\bc
    \includegraphics[width=0.4\textwidth]{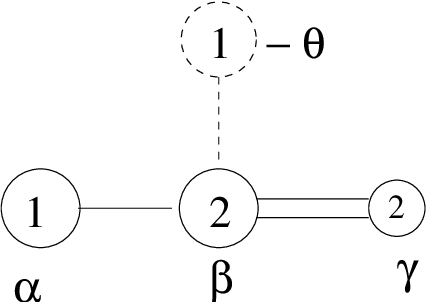}                  
     \ec
\caption{Dynkin diagram for $Spin(7)$ with the highest root and the Dynkin labels.}
\label{Dyn7}
\end{figure}

The extended Dynkin diagram for $Spin(7)$ is depicted in Fig. \ref{Dyn7}. There are three simple roots:
$\alpha$, $\beta$,
and $\gamma$.\footnote{See Appendix.} In a conveniently chosen orthonormal basis of $\mathfrak{C}$, $e_1 = i\gamma^1\gamma^2/2, \,e_2 = i\gamma^3\gamma^4/2, \, e_3 =  i\gamma^5\gamma^6/2$, the roots are\footnote{It is convenient in this case to choose the normalization where the length of the long roots is $\sqrt{2}$, while the length of the short root is 1.} 
$\alpha = (1,-1,0),\, \beta = (0,1,-1), \, \gamma = (0,0,1)$. The highest root is $\theta = \alpha + 2\beta + 2\gamma = (1,1,0)$.
The root vectors are
\be
&&E_{\pm \alpha} \ =\ \frac 14 [\gamma^2 \gamma^3 +  \gamma^4 \gamma^1  \pm i (\gamma^1 \gamma^3 + \gamma^2 \gamma^4)], \nn
&&E_{\pm \beta} \ =\ \frac 14 [\gamma^4 \gamma^5 +  \gamma^6 \gamma^3  \pm i (\gamma^3 \gamma^5 + \gamma^4 \gamma^6)], \nn
&&E_{\pm \gamma} \ =\ \frac 12 (\gamma^6 \pm i\gamma^5)\gamma^7\,.
 \ee
The coroots are
\be
\lb{coroots}
&&\alpha^\vee \ =\ [E_\alpha, E_{-\alpha}] \ =\ \frac i2 (\gamma^3 \gamma^4 - \gamma^1 \gamma^2), \nn
&&\beta^\vee \ = \ [E_\beta, E_{-\beta}] \ =\ \frac i2 (\gamma^5 \gamma^6 - \gamma^3 \gamma^4), \nn
&&\gamma^\vee \ =\ [E_\gamma, E_{-\gamma}] \ =\  -i \gamma^5 \gamma^6\,.
\ee
Note now that the Dynkin label for  the root $\beta$ is equal to 2. And the existence of the exceptional group element $\Omega_1$ with a nontrivial centralizer is related to {\it this} fact. Actually, $\Omega_1$ can be represented as 
\be
\Omega_1 \ =\ e^{i \pi \omega_\beta}\,,
 \ee
where $\omega_\beta = \theta^\vee = \frac i2 (\gamma^1 \gamma^2 + \gamma^3 \gamma^4)$ is the  fundamental coweight of the root $\beta$
(i.e. $[\omega_\beta, E_\alpha] = [\omega_\beta, E_\gamma] = 0$ and $[\omega_\beta, E_\beta] = E_\beta$) .

The centralizer of $\Omega_1$ can be found by inspecting the extended Dynkin diagram. Just delete the node $\beta$ and look at what is left. One sees three disconnected circles giving $[SU(2)]^3$, but, as was explained above, this should be factorized over the diagonal $Z_2$ center element. 
Note that it is only the node $\beta$ which plays a special role giving rise to an exceptional element. Similarly constructed elements $e^{i \pi \omega_\alpha}$ and 
 $e^{i \pi \omega_\gamma}$ are not exceptional. The former is just an element of the center, and its centralizer is the whole $Spin(7)$, while the centralizer of the latter is  $Spin(6) \equiv SU(4)$, which is simply connected.

We mentioned that the underlying reason for the existence  of exceptional elements is the presence of the roots whose Dynkin label is greater than one.
The root $\beta$ has this property, but the same is true for root  
 $\gamma$ and one may wonder why there is no associated exceptional element. We will not prove it here, addressing the reader to Ref. \cite{Kac1}, but it turns out  that the exceptional elements are only generated by the roots $\alpha_j$ where not just $a_j$, but also
 \be
a_j^\vee \ =\ a_j \, |{\alpha}_j|^2
 \ee
(with the length of the long roots normalized to 1) is greater than 1. The root $\gamma$ is short, $a_\gamma^\vee = 1$ and there is no associated exceptional element.   

That is why there are no exceptional triples for the groups $SU(N)$ and $Sp(2r)$. The former has only unit Dynkin labels, while the latter has Dynkin labels $a_j = 2$, but they all sit on the short nodes so that $a_j^\vee = 1$ everywhere. 

For the groups $Spin(7), Spin(8)$ and $G_2$ where only one root has $a_j^\vee > 1$, the exceptional element and the corresponding exceptional triple are isolated. This gives a single extra vacuum state. But for the groups $Spin(N \geq 9)$, $F_4$ and $E_{6,7,8}$, there are many such roots. This brings about extra {\it moduli spaces} of exceptional triples and of classical vacua. For higher exceptional groups, they  have a rather complicated structure. The spectra of effective Hamiltonians on these moduli spaces  lead to appearance of extra quantum vacua so that the final counting is $I_W = h^\vee$ in all the cases.
 
All the details of this counting are spelled out in the original papers \cite{Kac1,Keurentjes}. But to give the reader an idea of how it is done, we will briefly explain here the essentials.

\begin{itemize}
\item {\it Orthogonal groups}. The Dynkin diagrams for $Spin(9)$ and $Spin(10)$ include two nodes (call them $\alpha$ and $\beta$) with $a^\vee_\alpha = a_\beta^\vee = 2$. This brings about a 1-parametric moduli space of exceptional elements\footnote{{\it 2-exceptional} in the terminology of Ref. \cite{Kac1}.} $\Omega = {\rm exp} \{2\pi i (s_\alpha \omega_\alpha + s_\beta \omega_\beta)\}$ with $s_{\alpha} \geq 0$, $s_{\beta} \geq 0$ such that $2(s_\alpha + s_\beta) = 1$. That gives a 3-parametric moduli space of exceptional triples and hence of classical vacua. The effective Hamiltonian in this moduli space includes 3 bosonic variables  $c_j$ and their fermion superpartners $\eta_\alpha$.

 The problem is basically the same as for the ordinary moduli space for $SU(2)$, where we had two vacuum states with the wave functions $\Psi = 1$ and $\Psi = \eta^\alpha \eta_\alpha$. In the case under consideration, we get two {\it extra} vacuum states. Adding to this $r+1$ vacuum states associated with the maximal torus, we obtain 7 quantum vacua for $Spin(9)$ and 8 quantum vacua for $Spin(10)$.

For $Spin(11)$ and $Spin(12)$, there are 3 nodes with $a^\vee = 2$. This gives a 2-parametric moduli space of exceptional elements and 6-parametric moduli space of classical vacua, similar to the ordinary moduli space for $SU(3)$. We obtain 3 extra quantum vacua and all together $(r+1)+3$ vacua, which gives $I_W = 9$ for $Spin(11)$ and $I_W = 10$ for $Spin(12)$. 

This reasoning is easily generalized for  $Spin(N)$ with higher $N$, and we derive 
\be
\lb{ind-ON} 
I_W[Spin(N)] \ =\ N-2 \ = \ h^\vee\,.
\ee

\begin{figure}
\bc
    \includegraphics[width=0.4\textwidth]{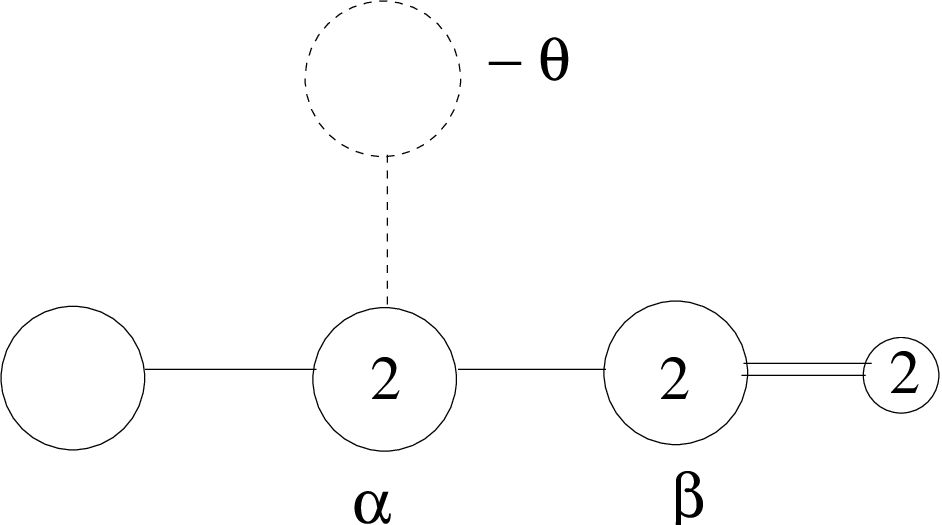}                  
     \ec
\caption{Extended Dynkin diagram for $Spin(9)$. Only the labels $a_j > 1$ are displayed.}
\label{DynO9}
\end{figure}

\begin{figure}
\bc
    \includegraphics[width=0.27\textwidth]{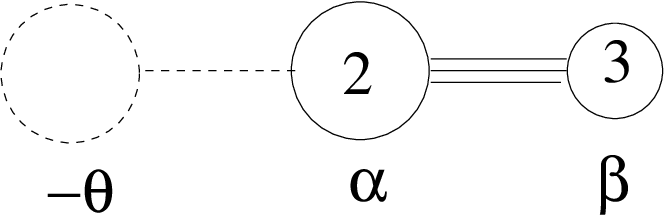}                  
     \ec
\caption{Extended Dynkin diagram for $G_2$.}
\label{DynG2}
\end{figure}

\item $G_2$. Its Dynkin diagram (see Fig. \ref{DynG2})  includes a long node $\alpha$ with the label $a_\alpha^\vee = a_\alpha = 2$ and the short node of length $1/\sqrt{3}$ with the label $a_\beta^\vee = a_\beta/3 = 1$. There is an isolated exceptional element $\Omega = \exp\{i\pi \omega_\alpha\}$ and the associated exceptional triple.\footnote{The latter can in fact  be lifted to $Spin(7) \supset G_2$ giving \p{triple-O7} \cite{KRS}.}
This gives an extra vacuum state. All together:
 \be
\lb{ind-G2} 
I_W[G_2] \ =\ 3+1  = 4 \ =\ \ h^\vee(G_2)\,.
\ee

\item $F_4$. The Dynkin diagram is shown in Fig. \ref{DynF4}.

\begin{figure}
\bc
    \includegraphics[width=0.5\textwidth]{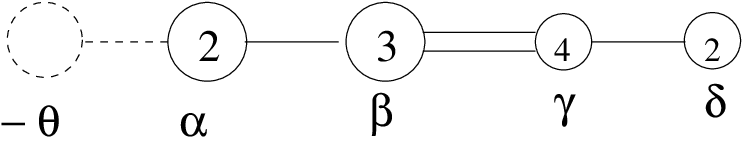}                  
     \ec
\caption{Extended Dynkin diagram for $F_4$.}
\label{DynF4}
\end{figure}
The node $\beta$ carries the Dynkin label $a_\beta = a_\beta^\vee = 3$. Up to a conjugation, there is only one associated  3-exceptional element $\Omega_1 = \exp\{ 2\pi i \omega_\beta/3\}$.  Its centralizer is $[SU(3) \times SU(3)]/Z_3$. Two other components of the triple represent  the Heisenberg pairs \p{Heisen} in each $SU(3)$ factor of the centralizer: $\Omega_2 = (P,P)$ and $\Omega_3 = (Q,Q)$. Now, $Z_3$ has two nontrivial elements and hence there are two different Heisenberg pairs: the pair \p{PQ} and its permutation.  Correspondingly, there are two isolated exceptional triples and two isolated vacua.

The nodes $\alpha$ and $\gamma$ carry the labels $a_\alpha^\vee = a_\gamma^\vee =2$. As was the case for $Spin(9)$ and $Spin(10)$, this brings about 2 extra quantum states. We have all together 4 extra vacua and the total count is 
\be
I_W[F_4] \ =\ (r+1) + 2_\beta + 2_{\alpha\gamma}  \ = \ 9 \ =\ h^\vee[F_4]\,.
\ee

\item $E_6$. The Dynkin diagram is shown in Fig. \ref{DynE6}.

\begin{figure}
\bc
    \includegraphics[width=0.5\textwidth]{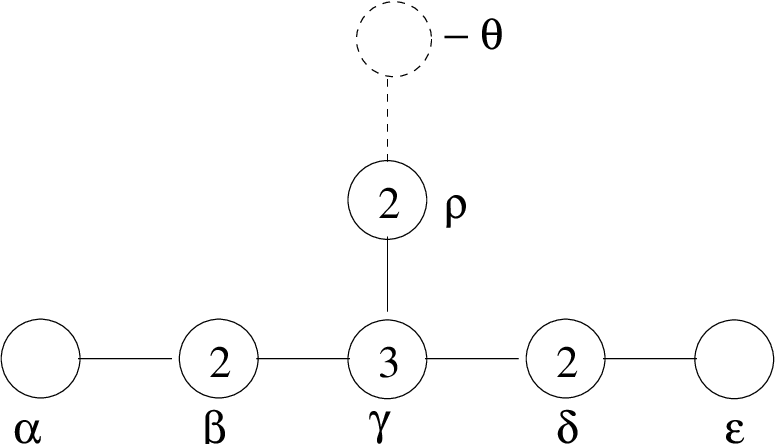}                  
     \ec
\caption{Extended Dynkin diagram for $E_6$.}
\label{DynE6}
\end{figure}
There are 2 extra isolated states associated with the node $\gamma$ and 3 quantum states coming from the moduli space associated with the nodes $\beta, \delta, \rho$. The total count is 
\be
I_W[E_6] \ =\ (r+1) + 2_\gamma + 3_{\beta\delta\rho}  \ = \ 12 \ =\ h^\vee[E_6]\,.
\ee

\item $E_7$. The Dynkin diagram is shown in Fig. \ref{DynE7}.

\begin{figure}
\bc
    \includegraphics[width=0.7\textwidth]{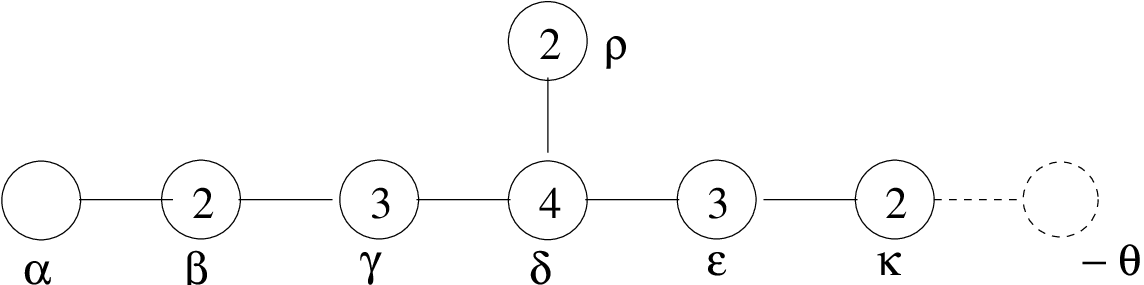}                  
     \ec
\caption{Extended Dynkin diagram for $E_7$.}
\label{DynE7}
\end{figure}

The nodes $\beta, \kappa, \rho, \delta$ bring about a 3-parametric moduli space of 2-exceptional elements. This  gives 4 extra states.
 
Two nodes $\gamma,\epsilon$ with $a_\gamma^\vee = a_\epsilon^\vee = 3$ give a 1-parametric moduli space of 3-exceptional elements having the form
$\Omega_1 = {\rm exp}\{2\pi i (s_\gamma \omega_\gamma + s_\epsilon \omega_\epsilon)\}$ with $3(s_\gamma + s_\epsilon) = 1$. A centralizer of each such element
is $[SU(3)]^3/Z_3 \times U(1)$. This gives two distinct 3-parametric moduli spaces of commuting triples.
 Each such moduli space brings about 2 quantum vacua, and the total contribution to the vacuum counting coming from these nodes is $2 \times 2 = 4$.

Besides, there are two  isolated exceptional triples associated with the root $\delta$. They are formed by a 4-exceptional element $\Omega_1 = {\rm exp}\{\pi i  \omega_\delta/2\}$ with the centralizer $[SU(4) \times SU(4) \times SU(2)]/Z_4$ and  two different Heisenberg pairs \p{Heisen} in the $SU(4)$ factors with the center elements $\pm i \mathbb{1}$. This gives
two isolated vacua.  
 The total count is 
\be
I_W[E_7] \ =\ (r+1) + 4_{\gamma \epsilon} + 4_{\beta \kappa \rho \delta} + 2_\delta  \ = \ 18 \ =\ h^\vee[E_7]\,.
\ee

\item $E_8$. The Dynkin diagram is shown in Fig. \ref{DynE8}.

\begin{figure}
\bc
    \includegraphics[width=0.7\textwidth]{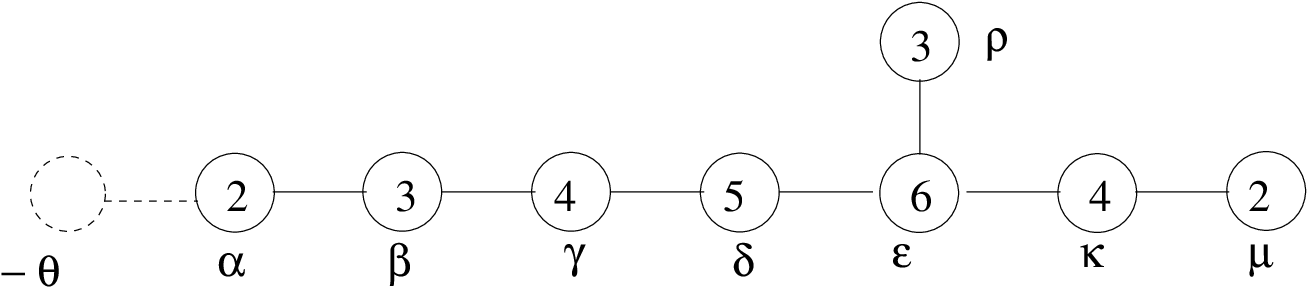}                  
     \ec
\caption{Extended Dynkin diagram for $E_8$.}
\label{DynE8}
\end{figure}

There is a 4-parametric moduli space of 2-exceptional elements spanned on the nodes $\alpha, \gamma, \epsilon, \kappa, \mu$. This gives 5 states.

Besides, the nodes $\gamma, \kappa$ give a 1-parametric moduli space of 4-exceptional elements $\Omega_1 = {\rm exp} \{2\pi i  (s_\gamma \omega_\gamma
+ s_\kappa \omega_\kappa)\}$ with $4(s_\gamma + s_\kappa) = 1$. Their centralizer is $[SU(4) \times SU(4) \times SU(2)]/Z_4 \times U(1)$. Two different Heisenberg pairs in $SU(4)$ with the center elements $\pm i \mathbb{1}$ give two distinct 3-parametric moduli spaces of exceptional triples and
   4 quantum states.

There is a 2-parametric moduli space of 3-exceptional elements:  $\Omega_1 = {\rm exp} \{2\pi i (s_\beta \omega_\beta
+ s_\rho \omega_\rho + s_\epsilon \omega_\epsilon)\}$ with $3(s_\beta + s_\rho + s_\epsilon) = 1$. Their centralizer is $[SU(3)]^3/Z_3\times [U(1)]^2$. Two Heisenberg pairs for $SU(3)$ give two 6-parametric moduli spaces of the exceptional triples. This gives $2\times 3 = 6$ quantum states.

The node $\delta$ gives an isolated 5-exceptional element $\Omega = {\rm exp} \{ 2\pi i \omega_\delta/5 \}$ with the centralizer $[SU(5) \times SU(5)]/Z_5$, four different Heisenberg pairs, and 4 quantum states.

The node $\epsilon$ gives an isolated 6-exceptional element $\Omega = {\rm exp} \{ \pi i s_\delta \omega_\delta/3 \}$ with the centralizer $[SU(6) \times SU(3) \times SU(2)]/Z_6$. The elements $e^{\pi i/3} \mathbb{1}$ and $e^{5\pi i/3} \mathbb{1}$ of $Z_6$ give two isolated Heisenberg pairs, two isolated triples and two vacua.
The total count is 
\be
I_W[E_8] \ = \ (r+1) + 5_{\alpha \gamma\epsilon\kappa\mu} + 6_{\beta\rho\epsilon} + 4_{\gamma\kappa} + 4_\delta + 2_\epsilon =  30 \ 
 = \ h^\vee[E_8]\,.\nn
\, 
\ee
\end{itemize}

Up to now, we only discussed in this section the dynamics of the theories with periodic boundary conditions. However, as soon as the center of the group is nontrivial, one can also impose twisted boundary conditions. For the unitary groups, the center is $Z_N$ and, as we discussed in   Sec.~6.3.2, if one does not admit  Euclidean configurations with a fractional topological charge, the twisted conditions lead to $N$ distinct vacuum states in a universe with a given vacuum angle $\vartheta_{\rm inst} \in [0,2\pi)$. For other classical groups, the center is smaller so that the classical vacua are not isolated, but form moduli spaces, and the analysis is more complicated. But it can be done  \cite{Borel} along the lines of \cite{Kac1,Keurentjes}. The total number of quantum vacua in a given instanton universe is always equal to $h^\vee$.  

\section{Theories with nonchiral matter}
\setcounter{equation}0

Now we will turn to the analysis of the vacuum dynamics of non-Abelian $4D$ supersymmetric gauge theories with matter and concentrate in this section on the models with left-right symmetric fermion content. 

Consider a theory involving, besides the adjoint vector multiplet $\hat V(x^\mu, \theta^\alpha, \bar \theta_{\bar \alpha})$, an equal number of chiral multiplets $S^f_j(x_\mu^L, \theta^\alpha)$ and $T^{fj}(x_\mu^L, \theta^\alpha)$ belonging to the fundamental and antifundamental representations of the gauge group,  $f = 1,\ldots, N_f$. Consider the following superfield Lagrangian (of {\it supersymmetric} $QCD$):
 \begin{empheq}[box=\fcolorbox{black}{white}]{align}
\lb{SQCD}
{\cal L}_{SQCD} \ =\ {\cal L}_{SYM} + \ \frac 14 \int d^4\theta \, (\overline S^f e^{g\hat V} S^f  + T^f e^{-g\hat V}  \overline T^f) \, +  \nn 
\frac {m_{fg}}2 \left(  \int T^{fj} S^g_j  \, d^2\theta + {\rm c.c.} \right).
  \end{empheq}
The complex mass matrix $m_{fg}$ can be arbitrary, but to avoid unnecessary complications, we will  assume $m_{fg} = m\delta_{fg}$. 

If the mass $m$ is large compared to $\Lambda_{\rm SYM}$, one seems to be able to get rid of the matter fields altogether by integrating them out so that the Witten index for such a
theory coincides with the index for the SYM theory \cite{Wit82}. This is so, indeed, for the model \p{SQCD}, but we will see later that the presence of matter may affect the index in the theories involving trilinear couplings between matter multiplets. Such is the theory based on  $G_2$ gauge group and  including three fundamental matter 7-plets with a nonzero  cubic superpotential. 

We will discuss this theory at the end of the section, and to begin with, we will describe in detail
the vacuum dynamics of the ordinary supersymmetric QCD based on the $SU(N)$ gauge group.  As far as the mass is not zero, the Witten index of such a model is the same as in the SYM theory, $I_W = N$, but at low masses, the vacuum dynamics is quite different from the dynamics of pure SYM theories and represents a separate interest.

\subsection{Unitary groups}

First of all, in a theory with fundamental fermions, fractional topological charges \p{Pontr} are not admissible --- neither in the infinite volume nor on a finite torus: the Euclidean path integral including  fundamental  fermion determinants is not defined in the gauge background with a fractional Pontryagin number. Thus, there are only the ordinary instanton universes with $\vartheta \in [0,2\pi)$.

 Take the  simplest theory \p{SQCD} based on the $SU(2)$ gauge group and including a single pair of the chiral matter multiplets.  Assume at first that the matter field are {\it massless}. As we shall see, the dynamics in this case is quite peculiar.
The component Lagrangian of this model  reads
\be
\lb{L-SQCD-massless}
{\cal L} \ = && - \frac 14 (F^a_{\mu\nu})^2 +  i \lambda^a \sigma^\mu {\cal D}_\mu \bar \lambda^a + i\bar \psi \sigma^\mu {\cal D}_\mu \psi 
+ i\xi \sigma^\mu {\cal D}_\mu \bar \xi  \nn
&& - s^* {\cal D}_\mu {\cal D}_\mu s  - t {\cal D}_\mu {\cal D}_\mu t^*  - \frac {g^2}2 (s^* t^a  s - t t^a t^*)^2 \nn
&& + \, ig\sqrt{2}[ s^*(\hat{ \lambda}  \psi ) - (\bar \psi \hat{\bar\lambda}) s -  (\xi \hat {\lambda} )t^*   +
 t( \hat { \bar \lambda} \bar \xi)]
\,.
\ee
 It does not involve a dimensional parameter and is invariant (at the tree level) under scale transformations. The corresponding conserved current is $\Theta_{\mu\nu} x^\nu$, where 
 $\Theta_{\mu\nu}$ is the energy-momentum tensor. 

But the theory is also supersymmetric and the dilatational current  has superpartners. One of the superpartners is the supercurrent and another is an axial current, the so-called $R$ current \cite{Salam} --- the generator of the following transformations:
\be
\lb{R}
\lambda \ \to \ e^{i\alpha} \lambda, \qquad  (s,t) \ \to \ e^{2i \alpha/3} (s, t), \qquad (\psi, \xi) \ \to \ e^{-i\alpha/3} (\psi, \xi)\,, 
\ee
while $A_\mu$ does not transform.
The invariance of \p{L-SQCD-massless} under \p{R} is manifest. The transformations \p{R} correspond to the following transformations of the superfields supplemented by a transformation of the odd superspace coordinate:
\be
\lb{R-super}
W \to e^{i\alpha}  W, \qquad (S,T) \to e^{2i\alpha/3} (S,T), \qquad \theta \to e^{i\alpha} \theta\,.
 \ee
We already met a similar transformation in   Sec. 3 when we discussed  pure SYM theory [see Eq. \p{sup-chiral}]. As we mentioned there, the symmetry 
\p{sup-chiral} is {\it anomalous} being broken by quantum effects. And the same concerns the symmetry \p{R-super}.
 
One may observe, however, that the Lagrangians \p{SQCD}, \p{L-SQCD-massless} are also invariant under two extra  symmetries when the superfield $\hat W$ and the coordinate $\theta$ are left intact, but the matter superfields (and their components) are multiplied by a phase:
 \be
\lb{chiral-ST}
S \ \to\  e^{i\beta} S \qquad   T \ \to\  e^{i\gamma} T\,.
 \ee
The generator of the transformation \p{chiral-ST} with $\beta = - \gamma$ is the {\it vector} $U(1)$ current. Such a transformation is a true symmetry of the theory even when the mass term is switched on. And the generator of the transformation  \p{chiral-ST} with $\beta =  \gamma$ is the {\it axial} $U(1)$ current. This symmetry is broken by the mass term. And not only by the mass term. As was the case for the $R$-symmetry \p{R-super}, the symmetry $(S,T) \to e^{i\beta}(S,T)$ is anomalous --- it is broken by quantum effects.

 A remarkable fact, however, is that there exist a certain combination of two anomalous chiral symmetries that is {\it not} anomalous and represents a {\it true} symmetry of the massless theory   \p{L-SQCD-massless}.  It reads \cite{ADS-84}
 \be
\lb{bez-anom-SU2}
W \ \to \ e^{i\alpha}  W, \qquad (S,T) \ \to \  e^{-i\alpha}(S,T),  \qquad \theta \to e^{i\alpha} \theta\,.
 \ee
The physics here is rather similar to what happens in ordinary QCD. The tree Lagrangian of QCD with one massless quark $q(x)$ admits two conserved currents: the vector current $\overline q \gamma_\mu q$ and the axial current $\overline q \gamma_\mu   \gamma^5  q$. The vector current is not anomalous, but the axial current is. The corresponding symmetry is also broken explicitly if the quark is endowed with a mass.

If one considers QCD with $N_f$ massless quark flavours, the symmetry of the tree Lagrangian is $U_L(N_f) \times U_R(N_f)$ describing independent unitary rotations of the left-handed  and right-handed quarks. The flavor-singlet axial symmetry is anomalous, but other symmetries are not. Spontaneous breaking of $SU_L(N_f) \times SU_R(N_f) \times U_V(1)$ down to $U_V(N_f)$ associated with the formation of quark condensate leads by Goldstone theorem to appearance of $N_f^2 - 1$ massless particles --- pseudoscalar mesons.\footnote{We hasten to add, however, that the physical $u,d$ and $s$ quarks are light, but  not massless, the axial symmetry is not exact, and pseudoscalar mesons carry mass, though this mass is comparatively small.} 

In massless supersymmetric QCD with one flavor (one pair of chiral multiplets), the situation is similar, the only essential difference is the presence of an extra massless Weyl fermion --- the gluino. As a result, at the tree level, we have two conserved axial currents, one of them being anomalous while the other is not. In a massless $SU(N)$ theory with $N_f$ pairs of multiplets, one has $N_f^2 - 1$  flavor-nonsinglet conserved axial currents, as in the ordinary QCD, and one extra flavor-singlet nonanomalous current corresponding to the symmetry
\be
\lb{bez-anom-SUN}
W  \to  e^{i\alpha}  W, \quad (S^f,T^f)  \to   e^{-i\alpha(N-N_f)/N_f }(S^f,T^f),  \quad \theta \to e^{i\alpha} \theta\,.
 \ee
Let us go back to the simplest $SU(2)$ theory  with a single flavor and write a structure
\be
\lb{struct-ADS} 
 \propto \int \frac {d^2\theta}{T^j S_j} \,, 
 \ee
which is {\it invariant} under the transformations \p{bez-anom-SU2} (we are restoring the color index, $j = 1,2$). This structure has a quite precise and important physical meaning: it enters the effective Born-Oppenheimer Lagrangian of massless SQCD in the region where the values of the scalar fields $|s|,|t|$ are much larger than the parameter $\Lambda$ such that the effective coupling constant $g^2(\Lambda)$ is of order 1. Then $g^2(|s|\sim|t|)$ is small.\footnote{As was the case for  pure SYM theory, the theory \p{SQCD} is asymptotically free and the effective coupling drops at large energies and grows at small energies.}

The structure \p{struct-ADS} is allowed by symmetry, but it does not {\it a priori} mean that it is actually generated. In fact, it {\it is}. The mechanism is purely non-perturbative associated with the instanton solutions \p{inst} of the Euclidean Yang-Mills field equations. We will not describe it here referring the reader
to the original papers \cite{ADS-84,NSVZ-supinst} and the reviews \cite{Amati,versus}.  

The {\it Affleck-Dine-Seiberg effective theory} represents a Wess-Zumino model including the chiral supermultiplets $S_j$ and $T^j$. The Lagrangian reads
\begin{empheq}[box=\fcolorbox{black}{white}]{align}
\lb{eff-ADS}
{\cal L}_{\rm ADS} \ = \  \frac 14 \int d^4\theta (\overline S^j S_j + T^j \overline T_j) + \frac{\Lambda^5}2  \left( \int \frac {d^2\theta}{T^j S_j} + {\rm c.c.} \right) \,.
 \end{empheq} 
The factor $\sim \Lambda^5$ appears by dimensional reasons,\footnote{A precise numerical coefficient in the second term depends on the precise definition of $\Lambda$. We have inserted $1/2$,  bearing in mind the convention in Eq.\p{L-WZ}.} The first term in \p{eff-ADS} comes from the first line of the original SQCD Lagrangian \p{SQCD}, where we omitted the supergauge fields, which play the role of fast variables in the region $|s| \sim |t| \gg \Lambda$ and are irrelevant for the low-energy dynamics. The second term includes the superpotential generated by nonperturbative effects:
\be
\lb{W-ADS}
{\cal W}^{\rm np}(S,T) \ = \ \frac {\Lambda^5} {T^j S_j}.
\ee
 The mass term is not included yet.

To find the vacuum states in this theory, one has to solve the equations
 \be
\lb{eq-vac}
\frac {\partial {\cal W}^{\rm np}}{\partial s_j} \ =\ \frac {\partial {\cal W}^{\rm np}}{\partial t^j} \ =\ 0\,.
 \ee
And such  solutions are {\it absent}! The energy density
 \be
 \epsilon(s, t) \ =\ \left| \frac {\partial {\cal W}^{\rm np}}{\partial s_j}  \right|^2 +  \left| \frac {\partial {\cal W}^{\rm np}}{\partial t^j}  \right|^2 \ =\ \Lambda^{10} 
\frac { s^{j*} s_j + t^j t_j^*}{|t^k s_k|^4}
 \ee
is positive definite, but it tends to zero when $s$ or $t$ go to infinity. We face the situation of {\it run-away vacuum}. The ground state in the quantum Hamiltonian is {\it absent} (so that supersymmetry is spontaneously broken), but there is a continuous spectrum including the states with arbitrarily low energies. 
Such is a rather peculiar (as promised) low-energy dynamics of massless $SU(2)$ SQCD with one flavor. 

But what happens if we switch on the mass? In this case, we have to add to the ADS superpotential \p{W-ADS} the tree superpotential $m T^j S_j$. In this case, the
 equations \p{eq-vac} (with ${\cal W}^{\rm np} \to {\cal W}^{\rm tot}$) have  finite solutions with
 \be
\lb{1complex}
(t^j s_j)^2 \ =\ \frac {\Lambda^5}m \ \stackrel{\rm def}= \ v^4\,.
 \ee
 As indicated in Eq.\p{L-SQCD-massless}, there is also a quartic scalar contribution in the potential energy.\footnote{It is often called ``$D$-term", because the quartic scalar contribution in  \p{L-SQCD-massless} takes its origin from the term $\sim D^a D^a$ in \p{SYM-comp}.}
 In vacuum, it should also vanish giving
 \be 
\lb{Dreal}
s^* t^a s - t t^a t^* \ =\ 0\,.
 \ee
We have altogether $2+3 = 5$ real conditions for $4\times 2 = 8$ real components of $s_j$ and $t^j$.  Modulo a 3-parametric gauge freedom,  this equation system has two distinct solutions. They can be chosen in the form
 \be
\lb{scal-vev} 
{\bf I}: \ s_j \ =\ \left( \begin{array}{c} v \\ 0 \end{array} \right), \qquad \qquad t^j \ =\ (v, 0)\,; \nn
{\bf II}: \ s_j \ =\ \left( \begin{array}{c} iv \\ 0 \end{array} \right), \qquad \qquad t^j \ =\ (iv, 0)\,.
  \ee
The existence of two distinct vacua conforms with the result $I_W = 2$ for  pure SYM theory with $SU(2)$ gauge group. 
 
Nonzero vacuum expectation values of the scalar fields break $SU(2)$ gauge symmetry by Higgs mechanism giving mass to gauge superfields. 
When $m$ decreases, the absolute value $v$ of the scalar vevs increases. In the limit $m \to 0$, the vacua run away to infinity.

For a small finite $m$, the characteristic mass of the massive vector superfields is of order
$$ M_V \sim g(v) v \sim v \sim \left(\frac {\Lambda^5}m \right)^{1/4} 
$$
up to calculable logarithmic corrections associated with the scale dependence of the coupling constant. This mass is much larger than the characteristic mass $m$ of 
scalar excitations, which endows the ADS Lagrangian the true BO status and justifies the analysis above. 

As was the case for  pure SYM theory, the vacua are characterized by a nonzero gluino condensate. Once the scalar vevs \p{scal-vev} are known, the values of the condensate can be calculated using the  anomalous {\it Konishi relation} \cite{Konishi} (see also \cite{Clark}):
 \begin{empheq}[box=\fcolorbox{black}{white}]{align}
\lb{Konishi}
  \{\hat {\bar Q}_{\dot \alpha}, \bar \psi^{\dot \alpha j} s_j\}  \ \propto \ -m t^j s_j + \frac {g^2}{32\pi^2} \lambda^a \lambda^a \,.
 \end{empheq}
Here the first term is the ``naive" classical result: the supercharge $\hat {\bar Q}_{\dot \alpha}$ makes the auxiliary field $F^*$ out of $\bar \psi$, and $F^* = -mt$ by an equation of motion. The second term is a quantum anomaly arising from a certain triangle graph. 
It is akin to chiral anomaly in ordinary QCD.

Take the vev of  both sides of \p{Konishi}. The average of the left-hand side over a supersymmetric vacuum  vanishes, and so should the right-hand side. We derive  
\be
\lb{ll=mL5}
g^2\langle \lambda^a \lambda^a \rangle \ = \ \pm 32 \pi^2 \sqrt{m\Lambda^5} \,.
 \ee

This analysis can be generalized to higher $N$. We dwell on the theory involving $N_f = N-1$ pairs of chiral multiplets having the same small mass $m$. In this case, the scalar vevs break gauge symmetry completely. For $N > 2$, the instantons generate the following effective superpotential:
 \be
\lb{suppot-N}
{\cal W}^{\rm np} \ = \ \frac {\Lambda^{2N+1}} {\det (T^{jf} S_j^g)}\,,
\ee 
 where $ f,g = 1,\ldots, N_f = N-1$ are the flavor indices. The integral $\int d^2 \theta \, {\cal W}^{\rm np}$ is invariant under the symmetry \p{bez-anom-SUN}. Add to \p{suppot-N} the tree superpotential $mT^{jf}S_j^f$. The vacuum states appear as the solutions to the equation system
\be
\lb{F-eqns}
\frac \pd {\pd s_j^f} \left[{\cal W}^{\rm np}(s, t) +  m \,t^{jf} s_j^f   \right]  =
  \frac \pd {\pd t^{jf}} \left[{\cal W}^{\rm np}(s, t) +  m\, t^{jf} s_j^f   \right]  = 0.
 \ee
supplemented with the ``$D$ flatness condition" \p{Dreal}.
The equations \p{F-eqns} imply that    
\be
\lb{vac-SUN}
t^{jf} s^g_j \ =\ e^{2\pi ik/N} \left( \frac {\Lambda^{2N+1}}m \right)^{1/N} \delta^{fg}\  \stackrel{\rm def}= \ v^2 e^{2\pi ik/N}  \delta^{fg} \,, 
 \ee
$k = 0,\ldots, N-1$. 
Taking in account also \p{Dreal} (involving now a sum over the flavors), we obtain (modulo a gauge transformation) $N$ distinct solutions:
\be
\lb{scalar-vev-N}
s_1 \ =\ \left( \begin{array}{c} v e^{i\pi k/N}\\ 0\\ \ldots \\ 0  \end{array} \right), \quad \ldots , \quad s_{N-1} \ =\ \left(\begin{array}{c} 0 \\ \ldots \\  v e^{i\pi k/N}\\ 0 \end{array} \right), \nn 
t^1 \ =\ (v e^{i\pi k/N}, 0,\ldots,0), \quad \ldots, \quad  t^{N-1}  \ =\ (0,\ldots, v e^{i\pi k/N}, 0) \,.
\ee
This conforms again with the vacuum counting $I_W = N$ for  pure SYM theory. As was mentioned, the scalar vevs \p{scalar-vev-N} break the gauge $SU(N)$ symmetry completely. The Konishi relation implies the formation of the gluino condensate:
\be
\langle \lambda^a \lambda^a \rangle \ \sim \   e^{2\pi ik/N} \left(\Lambda^{2N+1} m^{N-1}   \right)^{1/N} \,.
 \ee
If $N_f < N-1$, $SU(N)$ gauge symmetry is broken down to $SU(N-N_f)$, not completely. Still, for small masses, the low-energy dynamics is determined by scalar fields, and the analysis exhibits $N$ vacuum states also in this case \cite{ADS-84}.

 Take for example the $SU(3)$ theory with one pair of chiral multiplets. The full effective superpotential is\footnote{Then $\int d^2\theta \, {\cal W}^{\rm np}(S,T)$ is invariant under \p{bez-anom-SUN}.}
 \be
\lb{W-ST}
{\cal W}(S,T) \ =\ \frac {\Lambda^4}{\sqrt{T^jS_j}} + m T^jS_j\,.
 \ee
There are three vacua at
\be
\lb{v-SU3}
v^2 = t^j s_j \ =\ e^{2\pi i k/3} \left( \frac {\Lambda^4}{2m}   \right)^{2/3}\,.
 \ee

The corresponding gluino condensates  are
 \be
\langle \lambda^a \lambda^a \rangle \ \sim \   e^{2\pi ik/3} (m\Lambda^8)^{1/3}\,. 
\ee
In this case, the spectrum involves three different physical scales:
  \begin{itemize}
\item The low scale $\sim m$ of the scalar excitations and their superpartners.

\item  The mass $M_V  \sim v \sim (\Lambda^4/m)^{1/3}$ of the massive vector multiplets. 

\item An intermediate scale $ \Lambda_2$  --- the characteristic dimensional parameter for the remaining SQCD theory with the unbroken $SU(2)$ group. 
\end{itemize}

Both the original $SU(3)$ theory and the $SU(2)$ theory describing the dynamics at energies below $v$ are asymptotically free, and the effective coupling drops with energy. But the speeds by which it drops at $E \la v$ and $E \ga v$ are different. Generically, for SQCD with $N$ colors and $N_f < 3 N$ flavors, one can write at the one-loop level
 \be
g^2(\mu) \ =\ \frac {2\pi}{(3N - N_f) \ln \frac \mu {\Lambda_{N}}}\,.
 \ee
In the case under consideration, we have
 \be
\lb{g-SU2}
 g^2(\mu) \ \propto\ \frac {2\pi}{8\ln \frac \mu {\Lambda_3}}
\ee
for $\mu \ga v$ and 
 \be
\lb{g-SU3}
 g^2(\mu) \ \propto\ \frac {2\pi}{5\ln \frac \mu {\Lambda_2}}
\ee
for $\mu \la v$. 

The parameter $\Lambda_2$  can be found from the condition that the estimates \p{g-SU2} and \p{g-SU3} coincide at $\mu \sim v$. We obtain
\be
\Lambda_2 \sim \Lambda_3^{8/5} v^{-3/5}
 \ee
or, bearing in mind \p{v-SU3}, 
 \be
\Lambda_2 \ \sim \ \Lambda_3^{8/5} \left( \frac{\Lambda_3^4}m   \right)^{-1/5} = m^{1/5} \Lambda_3^{4/5} \,.
 \ee
For a generic $SU(N)$ theory with $N_f$ pairs of fundamental matter multiplets, where the symmetry is broken down to $SU(N-N_f)$, we derive
 \be
\Lambda_{N-N_f} \ \sim \ \Lambda_N  \left( \frac{m} {\Lambda_N}   \right)^{\frac {3N_f(N-N_f)}{2N(3N-4N_f)}} .
\ee 
This estimate\footnote{We disagree at this point with the estimate quoted in Ref. \cite{ADS-84}:
 $$
\Lambda_{N-N_f} \ \sim \ \Lambda_N  \left( \frac{m} {\Lambda_N}   \right)^{N_f/3N}\,.
$$
The latter estimate could be reproduced under the assumption that the matter does not affect the evolution of the coupling below $v$.} 
is valid as long as $3N > 4N_f$. Otherwise, the $SU(N-N_f)$ theory has too many matter fields and is not asymptotically free.

\vspace{1mm}

If $N_f \geq N$, the effective nonperturbative potential is not generated and the physics is the same as in pure SYM theory with $N$ colors. This gives again $I_W = N$.

\subsection{Domain walls}

The existence of distinct classical vacua suggests the existence of domain walls that interpolate between them. And in the limit  $m \ll \Lambda$ when the ADS Lagrangian describes the low-energy dynamics of SQCD, these walls exist, indeed. One can observe furthermore that they  are {\it BPS saturated}.

The abbreviation BPS stands for Bogomolny, Prasad, and Sommerfeld who found \cite{BPS} that, in the limit when the scalar self-interaction is switched  off, the mass of the `t Hooft-Polyakov monopole can be found by solving a simple first-order differential equation. The surface  energy and the profile of the domain walls interpolating between the different vacuum states of the ADS effective potential can be found  by a similar method.

In the problems in interest, the superpotentials depend on several superfields $S_j, T^j$. But to understand what happens, consider first the Wess-Zumino model \p{L-WZ} including only one superfield $\Phi$ with the  superpotential
 \be
{\cal W}(\Phi) \ =\ \mu^2 \Phi - \frac {\alpha \Phi^3}3\,.
 \ee 
The component Lagrangian reads
 \be
\lb{WZ-comp}
{\cal L}_{WZ} \ =\  \partial_\mu \phi^* \partial^\mu \phi + i \psi \sigma^\mu \partial_\mu \bar \psi  - \bar {\cal W}'(\phi^*) 
 {\cal W}'(\phi) \nn
 - \frac 12 {\cal W}''(\phi)\psi^2 - \frac 12 \bar {\cal W}''(\phi^*)\bar\psi^2 \,.
   \ee

If $\alpha > 0$, the model includes two vacua at $\phi_{\rm vac} = \pm \mu /\sqrt{\alpha}$.
Consider a static wall configuration interpolating between two minima:\footnote{Note that this configuration does not depend on two transverse coordinates and the problem  reduces to evaluating the mass of the corresponding {\it soliton} in a $2D$ model \cite{BPS-2D}.}  $\phi(z = -\infty) = -\mu/\sqrt{\alpha}$ and $\phi(z = \infty) = \mu/\sqrt{\alpha}$. Its surface energy density
is
 \be
\epsilon \ =\ \int_{-\infty}^\infty dz \, (|\pd_z \phi|^2 +  |{\cal W}'(\phi)|^2) \,.
 \ee
We can rewrite it as 
\be
&&\epsilon  = 
\int_{-\infty}^\infty dz \, [\pd_z \phi - e^{-i\delta}{\cal W}'(\phi^*)]  [\pd_z \phi^* - e^{i\delta}{\cal W}'(\phi)] \nn
&& + \int_{-\infty}^\infty dz \, [e^{i\delta} \pd_z \phi\,
{\cal W}'(\phi) + e^{-i\delta} \pd_z \phi^*
{\cal W}'(\phi^*)] \nn
&& = \int_{-\infty}^\infty dz |\pd_z \phi- e^{-i\delta}{\cal W}'(\phi^*)|^2 + 2 {\rm Re} \{e^{i\delta}[{\cal W}(\infty) - {\cal W}(-\infty)]\} 
 \ee
with a phase $\delta$ to be shortly determined.
The solution to the classical equations of motion represents a minimum of the energy functional. We arrive at the first-order equation
\be
\lb{eq-BPS}
\pd_z \phi - e^{-i\delta}{\cal W}'(\phi^*) \ =\ 0\,.
 \ee
In the case considered, we may set $\delta = 0$. The BPS wall configuration satisfying the equation $\pd_z \phi = {\cal W}'(\phi^*)$ is real:
\be
\phi(z) \ =\ \frac \mu{\sqrt{\alpha}} \tanh (\mu  \sqrt{\alpha} z) \,.
 \ee
 The surface energy is 
  \be
\lb{}
\epsilon \ =\  2 [{\cal W}(\infty) - {\cal W}(-\infty)] \ =\ \frac {8\mu^3}{3\sqrt{\alpha}}\,.
 \ee
Suppose now that $\alpha$ is negative. Then the vacua occur at purely imaginary values:
 $\phi_{\rm vac} = \pm i \mu/\sqrt{|\alpha|}$. The wall solution is also purely imaginary, it satisfies the equation \p{eq-BPS} with $\delta = \pi/2$ or  $\delta = -\pi/2$, 
depending on how $z$ is chosen. In other words, the parameter $\delta$ depends on the positions of the vacua on the complex plane, between which the wall interpolates \cite{Chibisov}. 

This method can be easily generalized to a WZ model including several superfields. It is not difficult to prove that the wall surface energy is always equal to or exceeds the {\it BPS bound},
 \begin{empheq}[box=\fcolorbox{black}{white}]{align}
\lb{BPS-bound}
\epsilon_{\rm BPS} \ =\  2 |{\cal W}(\infty) - {\cal W}(-\infty)|\,.
 \end{empheq}
For the SU(2) ADS model with the superpotential
$${\cal W} \ =\ \frac {\Lambda^5}{T^jS_j} + m T^j S_j\,, $$ we find \cite{KSS} the following profile of the BPS domain wall interpolating between the vacua \p{scal-vev}:
 \be
t^js_j(z) \ =\ v^2 \frac {(e^{4mz} + i)^2}{e^{8mz} + 1} \,.
 \ee   
 The absolute value of the moduli $t^js_j$ is constant, being equal to $v^2 = \sqrt{\Lambda^5/m}$, but its phase changes from  $\pi$ at $z = -\infty$ to 0 at $z = \infty$.
The surface energy of this wall
is 
 \be
\epsilon \ =\  2  |{\cal W}(\infty) - {\cal W}(-\infty)| \ =\ 4\left(\frac {\Lambda^5}{v^2} + mv^2   \right) = 8 \sqrt{m\Lambda^5}\,.
 \ee
Using \p{ll=mL5}, we can also express it via the gluino condensate in the original theory:
 \be
\lb{wall-energy}
\epsilon \ = \ \frac {g^2} {8\pi^2} \left| \langle\lambda^a \lambda^a\rangle_{z=\infty} - \langle\lambda^a \lambda^a\rangle_{z = - \infty} \right| \,.
 \ee
 
For higher $N$, there are $N$ complex vacua \p{vac-SUN}. There are domain walls interpolating between any pair of them and all these walls are BPS. The profile of the wall relating the adjacent vacua  with $\langle t^f s^g\rangle  = v^2 \delta^{fg}$ and $\langle t^f s^g \rangle = v^2 e^{2\pi i/N}\delta^{fg}$ was found numerically in Ref. \cite{wall-N}.   In contrast to the case $N=2$ where $|ts|$ did not depend on $z$, at higher $N$, the wall exhibits a small bump in the middle (see Fig. \ref{profil}).

 \begin{figure} [ht!]
      \bc
\vspace{-2cm}
    \includegraphics[width=0.7\textwidth]{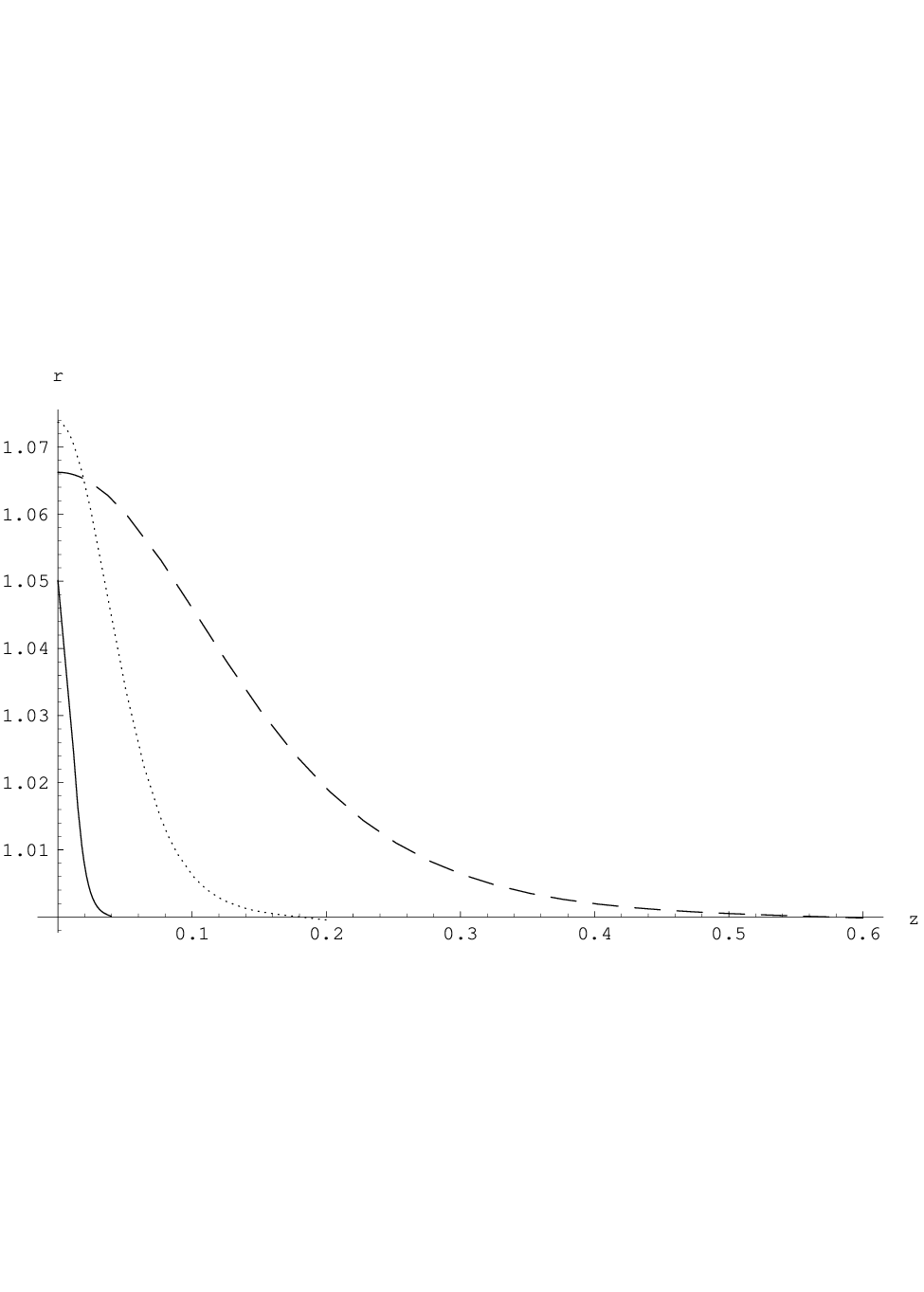}                  
     \ec
\vspace{-3cm}
    \caption{The plots of  $r(z) = |ts|(z)/|ts|(\pm \infty)$ for the BPS walls in Higgs phase for N = 3 (dashed line), N = 5 (dotted line)
and N = 10 (solid line). In each case, only a half of the wall (from its middle to the right) is displayed.} 
 \label{profil}
    \end{figure}  

\vspace{-1mm}

What happens if we increase the mass? If $m \sim \Lambda$, the ADS Lagrangian does not describe the low-energy dynamics of the theory anymore and the latter depends on the scalar superfields and gauge superfields on equal footing. In  pure SYM theory, we tried to extract information about  the vacua and the walls by studying the effective {\it non-Wilsonian} Veneziano-Yankielowicz Lagrangian \p{VY}. The results were surprising: besides the chirally asymmetric vacua, a chirally symmetric vacuum (probably, an artifact of the VY description) was observed, and the only kind of walls which we saw there were the walls relating a chirally asymmetric vacuum to the chirally symmetric one. The walls between different  asymmetric vacua did not exist.

In the theory with matter, we can write a generalization of \p{VY}, the {\it Taylor-Veneziano-Yankielowicz} effective Lagrangian \cite{TVY}. This Lagrangian depends on the composite gauge superfield \p{S=WW} (its cubic root $\Phi$ has canonical dimension of mass) and on the matter moduli 
${\cal M}^{fg} = T^{jf} S^g_j, \, f,g = 1, \ldots, N_f$. If the masses of all the matter fields are equal, the matrix ${\cal M}^{fg}$ can also be chosen in the diagonal form, ${\cal M}^{fg} =  X^2
\delta^{fg}$. 
For $N_f = N-1$, the TVY Lagrangian reads
\begin{empheq}[box=\fcolorbox{black}{white}]{align}
{\cal L}_{\rm TVY} \ =\ \alpha \int d^4\theta \, \overline \Phi \Phi + \beta \int d^4\theta \, \overline X X 
+  \left[\int d^2\theta \left( \gamma \Phi^3 \ln \frac {\Phi^3 X^{2(N-1)}}{\Lambda^{2N+1}} + \frac 12 m X^2\right) + {\rm c.c.}  \right]
 \end{empheq}
with some $\alpha,\beta,\gamma$.

 When the mass is much smaller than $\Lambda$ and the characteristic values of the lowest component of $\Phi$, one can integrate out the gauge fields and arrive at the ADS Lagrangian. When $m \gg \Lambda$, one can integrate out the matter fields and arrive at the VY effective Lagrangian.
The intermediate values of $m$ can be explored numerically. The results are rather amusing \cite{Veselov,wall-N}.
 \begin{itemize}
\item For all values of mass, there exist a ``real" BPS domain wall interpolating between a chirally asymmetric vacuum and the chirally symmetric one,  where  $\phi$ and $\chi$ (the lowest components of $\Phi$ and $X$) vanish and which the TVY Lagrangian enjoys in the same way as the VY Lagrangian discussed in the previous chapter did.
\item For $m \ll \Lambda$, the complex ADS domain wall is reproduced, but on top of that, one observes {\it another} BPS wall, in the middle of which the values of $\phi$ and $\chi$ are rather close to zero: the wall passes close to the chirally symmetric vacuum and dwells for a while in its vicinity.  When $N=2$, it simply splits in two real walls. For $N >2$, it stays at some distance from the chirally symmetric point even in the limit $m \to 0$.
\item When the mass increases, the second BPS wall retreats from the chirally symmetric point and approaches the first wall. At some critical value of the mass $m^* \sim \Lambda$, the two walls merge, a bifurcation occurs, and at still larger values of mass, the BPS solutions disappear.
 \item In a mass interval $m^* < m \leq m^{**}$, the wall solution to the second-order field equations still exists, but it is not BPS anymore.
 \item For $m > m^{**}$, complex domain walls are absent, as they were absent in the VY case.
\end{itemize}
The last fact is especially strange. If we believe in the picture of spontaneous chiral symmetry breaking in SYM theory or in supersymmetric QCD, the walls between different phases {\it must} exist. Maybe they do. But there is no trace of them in the VY (resp. TVY) approach for the theory $N_f = N-1$ with equal masses of all matter fields.

Note, however, that the situation is different  when the number of flavors $N_f$ is less than $N/2$. In this case, the TVY walls between different chiral phases are {\it tenacious} --- they stay there even in the limit $m \to \infty$ \cite{tenacious}\footnote{This  may happen also when $N_f = N-1$ if the masses of the matter fields are different, $m_{fg} \neq m\delta_{fg}$. } But this limit
 is not smooth. For large masses, the tenacious TVY walls acquire ``hard cores" near the cuts of the Kovner-Shifman potential (solid lines
in Fig. \ref{vac-sect}). In these narrow cores, the field $\phi$ does not change much, but the field $\chi$ does, so that the cores carry
  a significal fraction of the total wall energy.

 Well, the TVY dynamics is very rich and interesting. I am afraid, however, that  it does not recount well what happens (would physically happen if SQCD described nature) in supersymmetric QCD...

\subsection{Orthogonal groups}
The constructions above can be generalized to all other simple Lie groups \cite{ADS-ort,Cordes,ShifVain-ort,Morozov}. 

Consider e.g. a theory based on an orthogonal  group\footnote{The cases $N=3,4$ are special.  $Spin(3) \simeq SU(2), \, Spin(4) \simeq SU(2) \times SU(2)$, and   the formula $h^\vee[Spin(N)] = N-2$ is not valid.} ${\cal G} = SO(N \ge 5)$ and involving a  light matter chiral multiplet $S_a$ in the defining vector representation.
 Instantons and other nonpertubative gauge field configurations generate an effective Lagrangian depending on the gauge-invariant composite matter superfield $S_aS_a$. (The vector representation is real and one does not need a pair of mutiplets to cook up a gauge invariant.) To determine its form, we use the same logic as for the unitary groups. For all such theories, there exists a certain combination of  the $R$-current and the  matter axial current that is free from anomaly. In the case considered,  the corresponding exact symmetry of the action is
 \be
\lb{bez-anom-ON}
W \ \to \ e^{i\alpha}  W, \qquad S \ \to \  e^{-i\alpha(N-3)}S,  \qquad \theta \to e^{i\alpha} \theta\,.
 \ee
The effective Lagrangian that is invariant under this symmetry reads
 \be
 {\cal L}^{\rm np} \ =\ \int d^2 \theta \, \frac {\Lambda^{(3N-7)/(N-3)}}{(S_aS_a)^{1/(N-3)}} \ +\  {\rm c.c.}  
 \ee
Add here the tree mass term,
 $${\cal L}_m \ = \ \frac m2 \int d^2\theta \, S_a S_a \,, $$
and solve the equation $\pd {\cal W}/\pd S_a = 0$. We obtain $N-2$ distinct solutions:
 \be
\lb{vac-ON}
\langle s_a s_a \rangle_{\rm vac} \ \sim \ \frac {\Lambda^{(3N-7)/(N-2)}}{m^{(N-3)/(N-2)}} e^{2 \pi i k /(N-2)},  \qquad k = 0,\ldots,N-3 \,.
\ee
 Now, $N-2$ coincides with the dual Coxeter number $h^\vee$ for $SO(N \ge 5)$ and hence the Witten index for our theory coincides with the result \p{ind-ON} for the pure SYM theory based on $SO(N)$ group. 

For the $SO(N)$ theory involving $N_f$ vector multiplets $S^f$, the anomaly-free symmetry is\footnote{Take Eq.\p{bez-anom-SUN} and replace there 
$$N = h^\vee[SU(N)] \ \to \ N-2 = h^\vee[SO(N)]\,. $$
Note that a {\it single} vector multiplet $S_a$ for $SO(N)$ plays here the same role as a {\it pair} of multiplets $S_j$ and $T^j$ for $SU(N)$. It follows from the fact that the triangle diagrams describing the anomalous nonconservation of the matter axial currents give in these two cases the same contributions. 

And this in turn follows from two facts: {\it (i)} the Dynkin index \p{Dyn-ind} of a vector representation in $SO(N)$ is two  times larger than the Dynkin index of a fundamental representation in $SU(N)$; {\it (ii)} the vector representation is real and the anomalous triangle involves an additional symmetry factor $1/2$.} 
 \be
\lb{bez-anom-ON-Nf}
W \ \to \ e^{i\alpha}  W, \qquad S^f \ \to \  e^{-i\alpha(N-2-N_f)/N_f}S^f,  \qquad \theta \to e^{i\alpha} \theta\,.
 \ee
The invariant nonperturbative potential is generated for $N_f < N-2$ and reads
\be
  {\cal L}^{\rm np}(N_f) \ =\ \int d^2 \theta \, \frac {\Lambda^{(3N-N_f-6)/(N-2-N_f)}}{\det({\cal M})^{1/(N-2-N_f)}} \ +\  {\rm c.c.}  
 \ee
with ${\cal M}^{fg} = S_a^f S_a^g$.

The Witten index is equal to $N-2$ also in this case.

 \subsection{$G_2$}

Consider the SQCD theory based on the $G_2$ gauge group. We discuss first its simplest version with the Lagrangian including a single matter fundamental supermultiplet $S_{\alpha = 1,\ldots,7}$. The symmetry considerations dictate the following form of the invariant non-perturbative effective potential:
 \be
{\cal L}^{\rm np} \ = \ -\int d^2\theta \left( \frac {\Lambda^{11}}{S_\alpha S_\alpha} \right)^{1/3} + {\rm c.c.}
\ee
Adding the mass term $\sim m \int  d^2\theta \, S_\alpha S_\alpha$ and solving the equation $\pd {\cal W}^{\rm tot}/\pd S_\alpha = 0$, we find four distinct vacuum states:
 \begin{empheq}[box=\fcolorbox{black}{white}]{align}
\lb{vac-G2}
\langle s_\alpha s_\alpha \rangle_{\rm vac} \ \sim \left(  \frac {\Lambda^{11}}{m^3} \right)^{1/4} e^{i\pi k/2}\,.
\end{empheq}
The Witten index is equal to
\be
I_W \ =\ 4 \ =\ h^\vee(G_2)\,,
\ee
 which coincides with \p{ind-G2}. The same result is obtained in the theory including two or three massive fundamental supermultiplets. In the theories with four or more multiplets, the nonperturbative superpotential is not generated and the index also coincides with  the pure SYM result \p{ind-G2}. 

A similar analysis can be made for higher exceptional groups \cite{Morozov}. In the SQCD theories whose Lagrangian represents a sum of the pure SYM term, the kinetic term and the mass term for the matter fields, the index is always equal to $h^\vee(\cal G)$.

Note, however, that in the $G_2$ theory including {\it three} fundamental supermultiplets $S^f_\alpha$, we are allowed to write in the  Lagrangian, 
besides the mass term, also the  {\it Yukawa} term \cite{G2}:
  \be
  \lb{Yukawa-G2}
  {\cal L}^{\rm Yukawa} \ =\ \frac h 6 \int d^2\theta f_{\alpha\beta\gamma}\varepsilon_{fgh}\, S^f_\alpha S^g_\beta S^h_\gamma 
   \ee
with $f_{\alpha\beta\gamma}$ defined in \p{f}.
The coupling constant $h$ carries no dimension and the theory keeps its renormalizability.

The structure
\be
B \ =\ \frac 16 f_{\alpha\beta\gamma}\varepsilon_{fgh}\, S^f_\alpha S^g_\beta S^h_\gamma
 \ee
also appears in this theory as an extra moduli parameter on which the effective nonperturbative superpotential depends  besides the familiar bilinear modulus ${\cal M}^{fg} = S_\alpha^f S_\alpha^g$. This superpotential reads \cite{G2-ante}
 \be
\lb{suppot-G2}
{\cal W}^{\rm np}({\cal M}, B) \ =\ \frac {\Lambda^9}{\det {\cal M} - B^2}\,.
 \ee
Indeed, symmetry considerations dictate 
\be
 {\cal W}^{\rm np}({\cal M}, B) \ =\ \frac {\Lambda^9}{\det {\cal M}} f\left(\frac {B^2}{\det {\cal M}} \right)\,, 
 \ee 
 and the easiest way to derive that $f(x) = 1/(1-x)$  is to note that, if one
 breaks $G_2$ down to  $SU(3)$ by Higgs mechanism, assuming a  nonzero vacuum expectation value of one of
 the scalar fields, $\langle s^1_\alpha \rangle = v_0 \delta_{\alpha 7}$, the
  nonperturbative superpotential 
 \p{suppot-G2} should coincide with the one in the $SU(3)$ theory with
 two chiral triplets $\tilde S^{f=1,2}_j$ and two antitriplets $T^{(f=1,2) j}$ into which the 7-plets $S^2$ and $S^3$ split:
   \be
  \label{WSU3}
  {\cal W}^{\rm np}_{G_2}  \to v_0^2 {\cal W}^{\rm np}_{SU(3)} \ \propto \ 
   \frac 1{\det (T \tilde S) }
 \ee
[see Eq. \p{suppot-N}].

The full superpotential of the effective theory is 
     \begin{empheq}[box=\fcolorbox{black}{white}]{align}
    \label{WfullG2}
    {\cal W}  \ =\  \frac {1}
    {2({\rm det}  {\cal M} - B^2)}  +   \frac m2 {\rm Tr} \ {\cal M}  + h B\,.
    \end{empheq}
  (From now on, everything will be measured in the units of $\Lambda$.
  The factors $1/2$ are introduced  for convenience.)

Our goal is to find the classical vacua. Were the Yukawa term absent, we could use the ansatz like
 \be
\lb{Ans-lam0}
 \langle s^1_\alpha \rangle_{\rm vac} = v_0\delta_{\alpha 1}, \qquad  \langle s^2_\alpha \rangle_{\rm vac} = v_0\delta_{\alpha 2}, \qquad \langle s^3_\alpha \rangle_{\rm vac}  = v_0\delta_{\alpha 3}
 \ee
(any other relevant configurations could be obtained from this by a gauge transformation). Then the equation
$\pd {\cal W} /\pd v_0 = 0 $  would reduce to
\be
  \lb{mv08}
m(v_0^2)^4 \ = \ 1
  \ee 
  with 4 distinct solutions for the moduli $u = v_0^2$. 

If $h \neq 0$, the ansatz \p{Ans-lam0} should be modified. One can choose  it as \cite{G2}:
\be
\lb{Ans-lam}
\langle s^1_\alpha \rangle \ =\ \left(\begin{array}{c} 0\\ 0 \\ 0\\ 0\\0\\0\\v_0 \end{array} \right), \quad 
\langle s^2_\alpha \rangle \ =\  \frac 1{\sqrt{2}} \left(\begin{array}{c} v_1\\ 0 \\ 0\\ 0\\0\\v_2\\0 \end{array} \right), \quad 
\langle s^3_\alpha \rangle \ =\ \frac 1{\sqrt{2}} \left(\begin{array}{c} 0\\ v_1 \\ v_2\\ 0\\0\\0\\ 0 \end{array} \right) 
 \ee
with the constraint $v_0^2 = (v_1^2 + v_2^2)/2 \equiv u$.  In this case, the matrix {\large m}$^{fg} = \langle s^f_\alpha s^g_\alpha \rangle$ is still equal to
$ u\, \delta^{fg}$, whereas  
\be
\lb{b}
b \ =\ \frac 16 f_{\alpha\beta\gamma}\varepsilon_{fgh}\, \langle s^f_\alpha s^g_\beta s^h_\gamma \rangle \ =  \ \frac 12 v_0(v_1^2 - v_2^2)\,.
 \ee
 We have to solve the equations
 \be
\lb{eqn-ub}
\frac {\pd {\cal W}(u,b)}{\pd u} \ &=&\ - \frac {3u^2}{2(u^3 - b^2)^2} + \frac {3m}2 \ =\ 0\,, \nn
\frac {\pd {\cal W}(u,b)}{\pd b} \ &=&\  \frac b{(u^3 - b^2)^2} + h \ =\ 0\,.
 \ee
It follows that $b = -h u^2/m$ while $u$ satisfies the equation
 \be
\lb{eqn-u}
mu^4 \left(1 - \frac {h^2 u}{m^2} \right)^2 \ =\ 1\,.
 \ee
It has {\it six} roots bringing about six distinct vacua!  

Let us see what happens in the limit of small $|h|$, the smallness
   being characterized by a dimensionless parameter 
    \be
    \label{kappa}
   \kappa = \left| \frac{h^2}{m^{9/4}} \right| 
    \equiv \left|h^2 (\Lambda/m)^{9/4} \right| \ll 1
    \ee
    (do not forget that $h$ and $m$ are generically complex).
    Then four of the roots are very close to the known solutions of Eq.
    (\ref{mv08}), but on top of this, we see two extra roots at {\it large}
values of $u$: $ u = m^2/h^2 \pm h^2/m^{5/2}+\ldots$. In the limit
$h \to 0$, these new vacua run away to infinity of the moduli space
and decouple. 

The latter statement can be attributed a precise meaning. Different vacua
are separated by the domain walls.  Their surface energy (which is  the measure of the height of the barrier
separating different vacua) satisfies a strict lower BPS bound,
  \be
  \label{bound}
  \epsilon \ \geq \ 2|{\cal W}_1 - {\cal W}_2 |\,,
  \ee
 where ${\cal W}_{1,2}$ are the values of the superpotential at the vacua 1,2
 between which the wall interpolates [see Eq. \p{BPS-bound}].

When $\kappa = 0$ and only 4 vacua
are left, the domain walls between  them are BPS--saturated and have the
energy density $\epsilon = 2\sqrt{2} |m|^{3/4} $ for the walls connecting the
``adjacent'' vacua (with, say, $u = m^{-1/4}$ and $u = i m^{-1/4}$)  and
 $\epsilon = 4|m|^{3/4}$ for the walls connecting the
``opposite'' vacua with $u = \pm m^{-1/4}$ or with  $u = \pm i m^{-1/4}$.
In fact,
 the effective Higgs Lagrangian is exactly the same here 
as for the $SU(4)$ theory with 3 chiral quartets and 3 antiquartets,
whose domain walls were studied in  Ref.\cite{wall-N}.

When $\kappa$ is nonzero but small so that a couple of new vacua at large
values of $u$ appear, the values of the superpotential \p{WfullG2}
at these vacua are also large $\sim  m^3/h^2$. The bound 
(\ref{bound})
dictates that the energy density of a wall connecting  an ``old'' and a 
``new'' vacua is of order $|m^3/h^2| = m^{3/4}/\kappa$ and tends to infinity in the 
limit $\kappa \to 0$.

We can understand now why a heuristic argument of Ref. \cite{Wit82}, saying  that the matter sector decouples  for large masses and does not affect  the index counting,
does not work in this case. In the limit $m \to \infty$, it decouples. But if the mass is very large but finite, extra vacua at very large values of $u$ appear. They are separated from four conventional vacua at small $u$ by a very high barrier and are irrelevant, as far as the conventional low-energy dynamics is concerned. But a traveller who managed to  overpass this barrier would find himself in a completely different fascinating world. The presence of this world should be taken into account in evaluating the Witten index.

When we increase $\kappa$, the new vacua move in from infinity and, at
$\kappa \sim 1$, the values of $u_{\rm vac}$ for the vacua of both types
 are roughly the same. It is interesting that, at $\kappa = \frac 2{3\sqrt{3}}
  \approx .385$, two of the vacua (an ``old'' and
  a ``new'' one) become
 degenerate. At  $\kappa \gg 1$,
 six vacua find themselves at the vertices of a perfect hexagon in the complex
 $u$ -- plane.\footnote{The size of this hexagon is of order $|\sqrt{m}/h^{2/3}|$. To assure that $|u_{\rm vac}|$ are still large (in physical units, 
$|u_{\rm vac}| \gg \Lambda^2$) so that we are still in the Higgs phase and the light and heavy
 degrees of freedom are well separated, we should keep $|h/m^{3/4}| \ll 1$.
 But if $|m|$ is small enough, this condition can be fulfilled at arbitrary
 large $\kappa$.} The complex roots of Eq.(\ref{eqn-u}) in the units of $m^{-1/4}$
 for 3 illustrative values of $\kappa$ are displayed in Fig. \ref{6-roots}.

 \unitlength=.5mm

 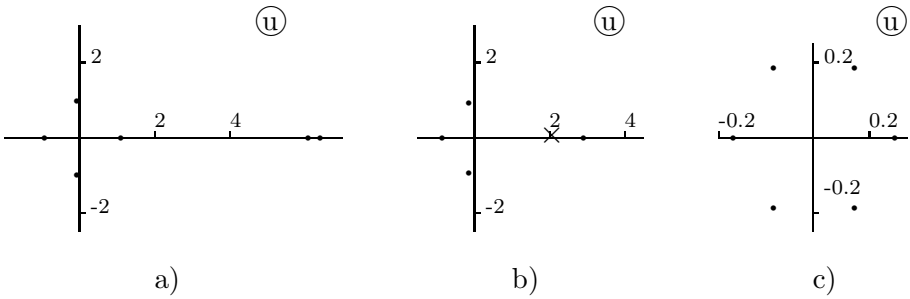
\begin{figure}

 \begin{picture}(220,100)
 \put(-20,50){\line(1,0){90}}
 \put(0,25){\line(0,1){55}}
\put(-9.3,50){\circle*{1.5}}
 \put(-.7,40.2){\circle*{1.5}}
 \put(-.7,59.8){\circle*{1.5}}
 \put(11,50){\circle*{1.5}} 
 \put(60.8,50){\circle*{1.5}}
 \put(64,50){\circle*{1.5}}
\put(20,50){\line(0,1){1.5}}
 \put(20,53){{\scriptsize 2}} 
 \put(40,50){\line(0,1){1.5}}
 \put(40,53){{\scriptsize 4}} 
 \put(0,70){\line(1,0){1.5}}
 \put(3,70){{\scriptsize 2}} 
 \put(0,30){\line(1,0){1.5}}
 \put(3,30){{\scriptsize -2}}
\put(20,10){{\small a)}}
 \put(48.6,78.6){{\small u}} 
 \put(51,81){\circle{8}}

 
 \put(90,50){\line(1,0){60}}
 \put(105,25){\line(0,1){55}} 

 \put(96.4,50){\circle*{1.5}}
 \put(103.5,40.7){\circle*{1.5}}
 \put(103.5,59.3){\circle*{1.5}}
 \put(122.3,48.5){$\times$} 
 \put(134,50){\circle*{1.5}}
 
 \put(125,50){\line(0,1){1.5}}
 \put(125,53){{\scriptsize 2}} 
 \put(145,50){\line(0,1){1.5}}
 \put(145,53){{\scriptsize 4}} 
 \put(105,70){\line(1,0){1.5}}
 \put(108,70){{\scriptsize 2}} 
\put(105,30){\line(1,0){1.5}}
 \put(108,30){{\scriptsize -2}}
 
 \put(115,10){{\small b)}}
 \put(138.6,78.6){{\small u}} 
 \put(141,81){\circle{8}}

 \put(170,50){\line(1,0){50}}
 \put(195,25){\line(0,1){50}}
 
\put(173.8,50){\circle*{1.5}}
 \put(184.6,31.4){\circle*{1.5}}
 \put(184.6,68.6){\circle*{1.5}}
 \put(206.1,31.4){\circle*{1.5}} 
 \put(206.1,68.6){\circle*{1.5}}
 \put(216.8,50){\circle*{1.5}}
 
\put(210,50){\line(0,1){1.5}}
 \put(210,53){{\scriptsize 0.2}} 
 \put(170,50){\line(0,1){1.5}}
 \put(170,53){{\scriptsize -0.2}} 
 \put(195,70){\line(1,0){1.5}}
 \put(198,70){{\scriptsize 0.2}} 
 \put(195,30){\line(1,0){1.5}}
 \put(198,35){{\scriptsize -0.2}}
 
\put(195,10){{\small c)}}
 \put(213.6,78.6){{\small u}}  
 \put(216,81){\circle{8}}
 \end{picture}

 \caption{ Solutions of the equation $u^4(1 - \kappa u)^2 = 1$
 for {\it a)} $\kappa = .16$, {\it b)} $\kappa = .385$, 
 and {\it c)} $\kappa = 100$. The cross marks out the double degenerate
root at $\kappa = .385$.}

\label{6-roots}
 \end{figure} 

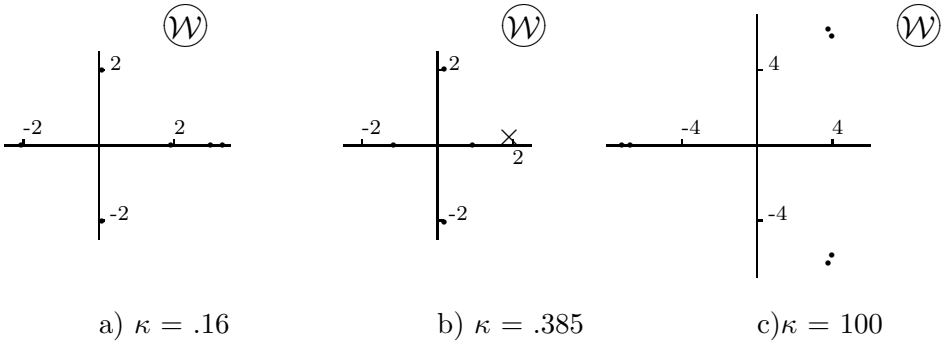
\begin{figure}

 \begin{picture}(220,100)

 \put(-10,50){\line(1,0){60}}
 \put(15,25){\line(0,1){50}}
 
\put(-5.7,50){\circle*{1.5}}
 \put(15.8,70.1){\circle*{1.5}}
 \put(15.8,29.9){\circle*{1.5}}
 \put(34.1,50){\circle*{1.5}} 
 \put(47.8,50){\circle*{1.5}}
 \put(44.7,50){\circle*{1.5}}

\put(35,50){\line(0,1){1.5}}
 \put(35,53){{\scriptsize 2}} 
 \put(-5,50){\line(0,1){1.5}}
 \put(-5,53){{\scriptsize -2}} 
 \put(15,70){\line(1,0){1.5}}
 \put(18,70){{\scriptsize 2}} 
 \put(15,30){\line(1,0){1.5}}
 \put(18,30){{\scriptsize -2}}
 
\put(15,0){{\small a) $\kappa$ = .16}}
 \put(33.4,78.6){$\cal W$} 
 \put(38,82){\circle{12}}


 \put(80,50){\line(1,0){50}}
 \put(105,25){\line(0,1){50}} 

\put(93.3,50){\circle*{1.5}}
 \put(106.8,70.3){\circle*{1.5}}
 \put(106.8,29.7){\circle*{1.5}}
 \put(114.4,50){\circle*{1.5}} 
 \put(120.8,50){$\times$}

 \put(125,50){\line(0,1){1.5}}
 \put(125,45){{\scriptsize 2}} 
 \put(85,50){\line(0,1){1.5}}
 \put(85,53){{\scriptsize -2}} 
 \put(105,70){\line(1,0){1.5}}
 \put(108,70){{\scriptsize 2}} 
\put(105,30){\line(1,0){1.5}}
 \put(108,30){{\scriptsize -2}}
 
 \put(105,0){{\small b) $\kappa$ = .385}}
 \put(123.4,78.6){$\cal W$} 
 \put(128,82){\circle{12}}

 \put(150,50){\line(1,0){70}}
 \put(190,15){\line(0,1){70}}

 \put(154.1,50){\circle*{1.5}}
 \put(208.9,18.9){\circle*{1.5}}
 \put(208.9,81.1){\circle*{1.5}}
 \put(209.9,20.8){\circle*{1.5}} 
 \put(209.9,79.2){\circle*{1.5}}
 \put(156.3,50){\circle*{1.5}}

\put(210,50){\line(0,1){1.5}}
 \put(210,53){{\scriptsize 4}} 
 \put(170,50){\line(0,1){1.5}}
 \put(170,53){{\scriptsize -4}} 
 \put(190,70){\line(1,0){1.5}}
 \put(193,70){{\scriptsize 4}} 
 \put(190,30){\line(1,0){1.5}}
 \put(193,30){{\scriptsize -4}}
 
\put(190,0){{\small c)$\kappa$ = 100}}
 \put(228.4,78.6){$\cal W$}  
 \put(233,82){\circle{12}}
 \end{picture}

 \caption{Values of the superpotential ${\cal W}_{\rm vac}$}
\label{suppot-6vac}

 \end{figure}

 The values of the superpotential ${\cal W}(u_{\rm vac}, b_{\rm vac})$
 for the same values of $\kappa$ in the units of $m^{3/4}$ 
 are shown in Fig. \ref{suppot-6vac}. We see that, for large
 $\kappa$, 6 vacua are clustered in 3 pairs (each pair corresponding to
 the opposite vertices of the hexagon in Fig. \ref{6-roots}c). The values of the
  superpotential for  two vacua of the same pair are close and hence the
  energy barrier between them is small. 

 Indeed, for very small masses and
  fixed $h$, one can neglect the mass term in the superpotential,
  in which case $|u_{\rm vac}|^3 \ll |b_{\rm vac}|^2$  and the vacuum values of $b$ are  determined from the equation $hb^3 = -1$.
 In this limit, we see only three vacua,  
 while their  splitting 
  inside a pair is the effect of higher order in $1/\kappa$.

\vspace{2mm}

Note finally that this phenomenon, the emergence of extra vacuum states when Yukawa couplings are switched on, is not specific for the $G_2$ theory. It is common for all models where such Yukawa couplings are admissible. Probably the simplest example of such a model is the theory based on the gauge group $SU(2)$ and involving a pair of chiral multiplets $S_j$ and $T^j = \varepsilon^{jk} T_k$ in the fundamental doublet representation [in $SU(2)$, the fundamental representation is pseudoreal, being equivalent to the antifundamental one] and an adjoint multiplet $\Phi_j^k$ with $\Phi_j^j = 0$. The theory including the Yukawa term,
\be
{\cal L} \ =\ h \int d^2\theta\, S^j \Phi_j^k T_k \ + \ {\rm c.c.},
\ee
has {\it three} rather than two vacuum states \cite{Gorsky}, with the third state going to infinity in the limit $h \to 0$.

\section{Chiral theories}
\setcounter{equation}0

In this section, we will discuss supersymmetric gauge theories with chiral matter content where the number of left-handed and right-handed fermions is not equal. 
One of such  (nonsupersymmetric) gauge theories is well known. It is the theory of electroweak interactions, which describes nature. But we are not studying in this review the real world, but rather a multitude of imaginary supersymmetric worlds...

\subsection{Chiral SQED}

The simplest version of a chiral gauge supersymmetric theory is chiral SQED.  The Lagrangian \p{L-SQED-comp} of the ordinary SQED includes a left-handed and a right-handed Weyl fermion fields or else a couple of left-handed fields with opposite electric charges. 
A chiral QED includes  sets of left-handed and right-handed fields, which are inherently different. It is more convenient to represent them by a set of left-handed fields with different electric charges $q_f$. 

There is one important restriction, however: the sum of the {\it cubes} of the charges  should vanish. The matter is that each charged Weyl fermion generates a triple photon interaction $M_{AAA}$ by the Feynman graph depicted in Fig. \ref{triangle}. This interaction breaks gauge invariance of the theory and its renormalizability.   Clearly, the contribution of each fermion to $M_{AAA}$ is proportional to $q_f^3$. The sum of all such contributions 
is proportional to 
$\sum_f q_f^3$, and it  should vanish. Otherwise, the theory would be {\it anomalous} (gauge symmetry of classical action would be broken by quantum effects) and not kosher.\footnote{We are talking now about an {\it internal} anomaly when the vertices of the triangle are coupled to the dynamical vector fields. This should be distinguished from the {\it external} anomalies discussed in Secs. 3,5, when one of the vertices corresponded to external axial current.  Such an external anomaly may break a certain symmetry of the theory, but it does not ruin it, keeping its gauge invariance.}

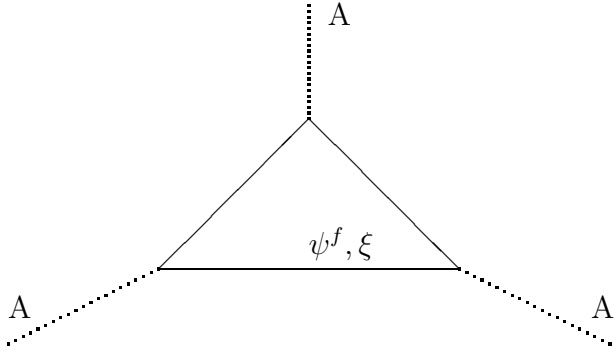
\begin{figure}

\begin{picture}(220,100)
 \put(50,30){\line(1,0){80}}
\put(50,30){\line(1,1){40}}
\put(90,70){\line(1,-1){40}}

\thicklines
\qbezier[20](10,10)(30,20)(50,30)
\qbezier[20](130,30)(150,20)(170,10)
\qbezier[20](90,70)(90,85)(90,100)
\put(10,17){A} \put(165,17){A}  \put(95,95){A}
\put(90,34){$\psi^f, \xi$}
 \end{picture}

 \caption{Anomalous triangles in chiral SQED}
\label{triangle}

 \end{figure}

The simplest version of chiral SQED includes 8 chiral superfields $S^f$ carrying a positive charge $e$ and a superfield $T$ carrying a negative charge $-2e$. The Lagrangian of the model reads
\begin{empheq}[box=\fcolorbox{black}{white}]{align}
\lb{L-SQED-chiral}
{\cal L}\ = \frac 14 \int \!\! d^4\theta (\overline S^f e^{eV} S^f + \overline T e^{-2eV} T) + 
\left( \frac 18 \int  \!\! d^2\theta \,W^\alpha W_\alpha   + {\rm c.c.} \right)
 \end{empheq}
We keep the coupling constant $e^2$ small.

The corresponding component Lagrangian includes among other things a term [cf. the last term in Eq.~\p{L-SQED-comp}]
$ -  e^2 (s^{*f} s^f - 2 t^* t)^2/2\,, $
where $s^f$ and $t$ are the scalar components of the multiplets $S^f$ and $T$.  This gives a classical vacuum valley of the type discussed in the previous section. Supersymmetry protects it  --- the corrections to the potential in all the orders of perturbation theory vanish. As a result, the spectrum of the theory does not have a gap: it is continuous starting right from zero.

In massless nonchiral theories, we had the same situation, but this degeneracy could be lifted there if the matter fields were given a mass. In  chiral theory, we cannot write a mass term, but we can add to the Lagrangian \p{L-SQED-chiral} a Yukawa term like 
\be
\lb{Yukawa-chiral}
{\cal L}_{\rm Yukawa} \
 = \ \frac h2 \int \!\!d^2\theta \, (S^f S^f T) \ + \ {\rm c.c.} 
 \ee
This gives a potential 
 \be
\lb{quart-pot}
V \ = \ h^2( 4|s^f t|^2 + |s^f s^f|^2 )\,,
 \ee
which plays the same role as the mass term in nonchiral theories: it locks the valley and brings about a gap in the spectrum.  Unfortunately, we cannot send in this case the interaction constant $h$ to infinity, as we did with the mass parameter for the nonchiral SQED. In fact, we need to keep both $h$ and $e^2$ small to avoid running into the Landau pole (because of which the chiral SQED, as well as the ordinary SQED and the ordinary QED are not, strictly speaking, well-defined quantum theories).  Also, we do not have here a nonperturbative instanton superpotential, which pushed for small masses the system far away along the valley in the case of nonchiral SQCD. 

All this makes the analysis of the vacuum dynamics of chiral theories essentially more difficult than for nonchiral ones. Still, we will try to perform it. As we will see, one can choose the value of the Yukawa coupling in such a way that this dynamics is mainly determined by the gauge sector and the matter degrees of freedom can 
(with a proper care) be integrated over.

Following \cite{chiral-SQED}, we proceed  in the same way as we did for the ordinary SQED in   Sec. 2. We put the system in a finite box of size $L$, impose  periodic boundary conditions on all the fields and expand the latter in Fourier series:
\be
\lb{Fourier}
A_j(\vecg{x}) &=& \sum_{\vecind{n}} A_j^{(\vecind{n})} e^{2\pi i \vecind{n}\cdot \vecind{x}/L}, \qquad 
\lambda_\alpha(\vecg{x}) = \sum_{\vecind{n}} \lambda_\alpha^{(\vecind{n})} e^{2\pi i \vecind{n}\cdot \vecind{x}/L}, \nn
  s^f(\vecg{x}) &=& 
\sum_{\vecind{n}} s^{f(\vecind{n})} e^{2\pi i \vecind{n}\cdot \vecind{x}/L}, \qquad 
\psi_\alpha^f(\vecg{x}) = \sum_{\vecind{n}} \psi_\alpha^{f(\vecind{n})} e^{2\pi i \vecind{n}\cdot \vecind{x}/L}, \nn
t(\vecg{x}) &=& \sum_{\vecind{n}} t^{(\vecind{n})} e^{2\pi i \vecind{n}\cdot \vecind{x}/L}, \qquad \quad 
\xi_\alpha(\vecg{x}) = \sum_{\vecind{n}} \xi_\alpha^{(\vecind{n})} e^{2\pi i \vecind{n}\cdot \vecind{x}/L}\,.
 \ee
The slow variables $A_j^{(\vecind{0})} \equiv c_j$ and their superpartners  $\lambda_\alpha^{(\vecind{0})} \equiv \eta_\alpha$ play a special role. Thinking in Hamiltonian terms, we 
stay in the universe where the wave functions are invariant under the large gauge transformations [cf. Eq.~\p{large-gauge-QED}],
\be
&&A_j \to A_j + \frac{2\pi n_j}{eL}, \nn
 &&(s^f,\psi^f)  \to (s^f,\psi^f) e^{-2\pi i \vecind{n} \cdot \vecind{x}/L}, \quad (t, \xi) \to (t,\xi) e^{4\pi i \vecind{n} \cdot \vecind{x}/L}\,.
 \ee
The zero mode $c_j$ of the gauge potential lies in a dual torus, $c_j \in [0, \frac {2\pi}{eL})$, and a characteristic energy of this degree of freedom is $\sim e^2/L$. Most other degrees of freedom have a  characteristic energy  $\sim 1/L \gg e^2/L$ and are fast, to be integrated out. We will find then the effective supercharges by averaging the full supercharges over the  ground eigenstate of the fast Hamiltonian:\footnote{In contrast to what we had for ordinary SQED, we cannot now, even at the leading order, just to cross out the terms including fast variables in the full Hamiltonian. Well, we can do so for the terms including higher harmonics of $A_j$ and $\lambda_\alpha$, but not for the matter fields.} 
 \be
\hat Q_\alpha^{\rm eff} \ =\ \langle \hat Q_\alpha^{\rm full} \rangle_{\rm fast\ vacuum}, \qquad \hat{\bar Q}^{\alpha\, {\rm eff}} \ =\ \langle \hat{\bar Q}^{\alpha\, {\rm full}} \rangle_{\rm fast\ vacuum}\,.
 \ee
The effective Hamiltonian will then be restored as the anticommutator $\hat H^{\rm eff} = \{\hat Q_\alpha^{\rm eff},   \hat{\bar Q}^{\alpha\, {\rm eff}} \}/2$.

However, in contrast to  ordinary SQED, the variables $c_j$ and their superpartners $\eta_\alpha$ are not the {\it only} slow variables that determine the low-energy dynamics of  chiral SQED. There are also the variables describing  the motion along the valley 
\be
\lb{valley-SQED}
\sum_{f=1}^8 s^{*f} s^f - 2 t^* t \ =\ 0\,.
 \ee
When the Yukawa coupling $h$ vanishes, the spectrum associated with this motion is continuous. 

If $h$ is nonzero but small, the spectrum is discrete, but the characteristic  excitation energies and level spacings are also small. The scalar valley dynamics is described in this case by a WZ model with the quartic potential \p{quart-pot}. By virial theorem, the characteristic kinetic and potential energies should be of the same order, which means
$$ \frac 1 {L^3|s|^2} \ \sim \   \frac 1 {L^3|t|^2}  \ \sim \ h^2 L^3 |s|^2 |t|^2 \ \sim \ h^2 L^3  |s|^4 \,,$$
giving $|s|^2 \sim |t|^2 \sim h^{-2/3}/L^2$ and 
 \be
E_{\rm char}^{\rm WZ} \ \sim \ \frac {h^{2/3}}L\,.
 \ee
 If $h \ll e^3$, this is much less than the characteristic energies in the sector of the theory associated with the slow gauge dynamics. We will be interested in the case where the low-energy dynamics is mainly determined by the gauge sector. To this end, we must keep $h \gg e^3$. 
Still, as we will see soon, $h$ cannot  be too large. The calculation below is valid in the region 
 \be
\lb{range-h}
e^3 \ \ll \ h \ \ll \ e\,.
 \ee 

All the details of this calculation are spelled out in Ref.~\cite{chiral-SQED}, but let us give here some hints explaining how it works. 
What we should do is to take the fast Hamiltonian in the leading quadratic approximation, find its spectrum, treating the slow variables $c_j, \eta_\alpha$ as parameters, and determine the effective supercharges by calculating the averages of the full supercharges over the vacuum state of $\hat H^{\rm fast}$.

Take  one of the multiplets of unit charge, set $L=1$, concentrate on the values of $\vecg{c}$ close to the origin, neglect the Yukawa terms, and keep, to begin with, only the zero Fourier modes of the matter fields.
We obtain a SQM model \cite{Palumbo} with the supercharges\footnote{We slightly modified the conventions of Ref. \cite{chiral-SQED}. Note that $\bar X_\alpha = - (X^\alpha)^\dagger$, which follows  from $\bar X^\alpha =  (X_\alpha)^\dagger$  and from \p{low-raise}.}
\be
\lb{Q-chiral-SQM}
\hat Q_\alpha &=&  \left[\hat P_k (\sigma_k)_\alpha{}^\gamma + i e s^*\!s \,\delta_\alpha^\gamma  \right] \eta_\gamma + 
\sqrt{2}\left[ -i \hat \pi^\dagger \delta_\alpha^\gamma + e s c_k (\sigma_k)_\alpha{}^\gamma\right] \bar \psi_\gamma, \nn
\hat {\bar Q}^\beta &=&  {\bar \eta^\delta} \left[\hat P_k (\sigma_k)_\delta{}^\beta - i e s^*\!s \,\delta_\delta^\beta  \right]  -  
\sqrt{2} \left[ i \hat {\pi} \delta_\delta^\beta + e s^* c_k (\sigma_k)_\delta{}^\beta\right] \psi^\delta 
   \ee
with $\hat \pi = -i \pd/\pd s, \ \hat \pi^\dagger = -i \pd/\pd s^*$. The anticommutator of the supercharges \p{Q-chiral-SQM}
has the form
 \be
\lb{alg-gauge}
\{\hat Q_\alpha, \hat {\bar Q}^\beta \} \ =\ 2 \delta_\alpha{}^\beta \hat H - 2e c_k (\sigma_k)_\alpha{}^\beta \hat G\,,
 \ee
where
 \be
\lb{Ham-full}
\hat H \ &=&\  - \frac {\pd^2}{\pd s \pd s^*} + e^2 \vecg{c}^2 s^*\!s  + ec_k \hat{\bar \psi} \sigma_k   \psi  \nn
&-& \frac 12 \triangle_{\vecind{c}} + \frac {e^2}2  (s^* s)^2 + ie\sqrt{2} [(\hat{\bar \psi} \hat{\bar \eta}) s - (\psi \eta) s^*] 
\ee
is the Hamiltonian and the presence of the second term in the R.H.S. of Eq. \p{alg-gauge} with
\be
\lb{G-ab}
\hat G \ =\ s^* \frac \pd {\pd s^*}  - s \frac \pd {\pd s} + \psi_\alpha \hat {\bar \psi}^\alpha -1 
 \ee
 tells us that we are dealing with a {\it gauge} SQM system. $\hat G$ is the generator of the Abelian gauge transformations. It commutes with the Hamiltonian. The physical states in the spectrum are the solutions of the Schr\"odinger equation $\hat H \Psi = E \Psi$ that satisfy the subsidiary condition $\hat G \Psi = 0$.

The Hamiltonian \p{Ham-full}  matches the 
Lagrangian \p{L-SQED-comp} [a part of it not involving the mass terms and the fields $t(x)$ and $\xi(x)$] deprived of the higher Fourier modes contribution. 
 If we add the contribution of seven other matter fields $(s^f, \psi^f)$ and of $(t,\xi)$, take into account also higher Fourier harmonics,  and evaluate the anticommutator $\{\hat Q_\alpha,  \hat {\bar Q}^\alpha \}$, the full Hamiltonian of SQED in a finite box is reproduced.

The second terms in Eqs.~\p{Q-chiral-SQM} are {\it fast} supercharges and the first line in \p{Ham-full} is the fast Hamiltonian describing  a variant of supersymmetric oscillator. \index{oscillator!supersymmetric} The ground state of $\hat H^{\rm fast}$, which the fast supercharges annihilate, has the unit fermion charge:
\be
\lb{fast-ground}
\Psi_0(s, s^*; \psi^\alpha) \ =\ w_\alpha(\vecg{c}) \psi^\alpha \exp\{-e|\vecg{c}|s s^*\}
\ee
with 
\be
\lb{w}
w_\alpha \ =\ \left( \begin{array}{c} - \sin \frac \theta 2 e^{-i\phi} \\ \cos \frac \theta 2  \end{array} \right).
\ee
Here $\theta$ and $\phi$ are the polar and azimuthal angles of the vector $\vecg{c}$. The spinor \p{w} satisfies 
$$\vecg{n} \vecg{\sigma} w  \ =\ -w$$ 
and coincides with the wave function of a spin $\frac 12$ particle with the spin directed along the negative $z$ axis.
The state \p{fast-ground} is gauge-invariant: $\hat G \Psi = 0$. 

Our next task is to find the effective supercharges by averaging the slow supercharges [the first terms in \p{Q-chiral-SQM}] over the fast vacuum \p{fast-ground}. It is easy to derive 
\be
\lb{s-sstar}
\langle s^*\!s \rangle_{\rm fast\ vacuum} \ =\ \frac 1{2e|\vecg{c}|} \,.
\ee
But one has also take into account the factor $w_\alpha \psi^\alpha$ in $\Psi_0(s, s^*; \psi^\alpha)$ with a nontrivial dependence on $\vecg{c}$. As a result,  $\langle \hat P_k \rangle_{\rm fast\ vacuum}$ does not vanish, giving a
``vector potential" ${\cal A}_k(\vecg{c})$ which should be added to the operator $\hat P_k$ in the effective \index{supercharges} supercharge. We derive

\be
\lb{Q-chiral-SQED}
   \hat Q_\alpha^{\rm eff} \ =\  \left[(\sigma_k)_\alpha{}^\beta (\hat P_k + {\cal A}_k) +i  \delta_\alpha^\beta {\cal D}    \right] \eta_\beta, \nn
\hat {\bar Q}^{\alpha \, {\rm eff}} \ =\   \hat{\bar \eta}^\beta \left[(\sigma_k)_\beta{}^\alpha (\hat P_k + {\cal A}_k) - i \delta_\beta^\alpha {\cal D}    \right], 
 \ee
\begin{empheq}[box=\fcolorbox{black}{white}]{align}
\lb{H-chiral}
\hat H^{\rm eff} \ =\ \frac 12 (\hat P_k + {\cal A}_k )^2 + \frac 12 {\cal D}^2 + \vecg{\cal H}\cdot \hat{\bar \eta} \vecg{\sigma} \eta\,,
  \end{empheq}
where  $ \hat P_k = -i\pd/\pd c_k$, ${\cal D} = 1/(2|\vecg{c}|)$ and
\be
\lb{cond-KAH}
\vecg{\cal H} \ =\ \vecg{\nabla} \times \vecg{\cal A} \ =\ \vecg{\nabla} {\cal D} \,. 
\ee
The vector function $\vecg{{\cal A}}(\vecg{c})$ depends on field variables rather than on spatial coordinates and has nothing to do with the  vector 
potentials $\vecg{A}(\vecg{x})$ !

As was also the case for the ordinary SQED [see Eq. \p{kappa-BO}], the effective Hamiltonian \p{H-chiral} has the benign Wilsonian nature as soon as $|\vecg{c}|$ is not {\it too} small and 
\be
\kappa \ =\ \frac 1{e|\vecg{c}|^3} \ \ll \ 1\,.
\ee
But chiral SQED includes also the Yukawa term \p{Yukawa-chiral}, and we want to make sure that the corresponding corrections to the Hamiltonian $\sim 
h^2 (s s^*)^2$ are irrelevant if $\kappa$ is small. We have to compare the characteristic value of these corrections with the gap in the spectrum of the fast Hamiltonian, which is of order $e|\vecg{c}|$. Bearing in mind \p{s-sstar}, we derive
$$ \frac {h^2}{e^2|\vecg{c}|^2} \ \ll \ e|\vecg{c}|$$
or $h^2 \kappa \ll e^2$. We are on the safe side if 
$$ h \ \ll \ e\,,$$
and that is how the upper limit for the allowed range of $h$ in \p{range-h} was derived.

Let us now take into account   all other Fourier harmonics of all the matter  fields $s^f(x), \psi^f(x), t(x)$ and $\xi(x)$ and restore the dependence on $L$. We obtain  
\be
\lb{D-SQED}
{\cal D} \ =\ 4\sum_{\vecind{n}} \frac 1 {\left|\vecg{c} - \frac {2\pi \vecind{n}}{eL} \right|}  - \frac 12 \sum_{\vecind{n}}   \frac 1{\left|\vecg{c} - \frac {\pi \vecind{n}}{eL} \right|  } .       
 \ee

The Hamiltonian \p{H-chiral} with a generic function ${\cal D}(\vecg{c})$ coincides with the Hamiltonian written in the paper \cite{Ritten}. If ${\cal D} = -1/(2 |\vecg{c}|)$, the magnetic field
  \be
\vecg{\cal H} \ =\  \frac {\vecg{c}}{2|\vecg{c}|^3}
 \ee
is the field of a magnetic monopole with the minimal magnetic charge\footnote{Note, however, the presence of the extra scalar potential $1/(8\vecg{c}^2)$, which makes the problem supersymmetric.} $g = 1/2$. This monopole singularity in the configuration space is known in other branches of physics by the name {\it Pancharatnam-Berry phase} \cite{Berry}.

In our case, we are dealing with a  {\it cubic lattice} of monopoles. In fact, there are two superimposed cubic lattices: a lattice of monopoles of charge $1/2$ with edge length
$\pi/eL$ and the lattice of monopoles of charge $-4$ located at the nodes of a twice less dense lattice with edge length $2\pi/eL$. The net magnetic charge density in field space  is zero --- this is a corollary of the condition 
 \be
 \sum_f q_f^3 \ =\ 0\,, 
 \ee
which our chiral model enjoys.

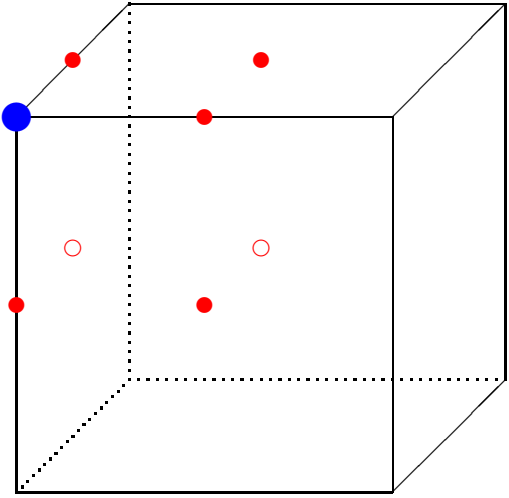
\begin{figure}

\begin{picture}(150,150)
 \put(0,0){\line(1,0){100}}
\put(0,0){\line(0,1){100}}
\put(0,100){\line(1,0){100}}
\put(100,0){\line(0,1){100}}
\put(0,100){\line(1,1){30}}
\put(100,100){\line(1,1){30}}
\put(30,130){\line(1,0){100}}
\put(100,0){\line(1,1){30}}
\put(130,30){\line(0,1){100}}

\put(0,100){\blue \circle*{8}}
\put(0,50){\red \circle*{4}}
\put(50,100){\red \circle*{4}}
\put(50,50){\red \circle*{4}}
\put(15,115){\red \circle*{4}}
\put(65,65){\red \circle{4}}
\put(65,115){\red \circle*{4}}
\put(15,65){\red \circle{4}}

\thicklines
\qbezier[20](0,0)(15,15)(30,30)
\qbezier[40](30,30)(80,30)(130,30)
\qbezier[40](30,30)(30,80)(30,130)

 \end{picture}

 \caption{A unit cubic cell of the monopole crystal in field space. The edge of the cube is $2\pi/(eL)$. The blue blob marks the site with the monopole of charge $-7/2$. 
The red blobs mark the sites with the monopoles of charge $1/2$.}
\label{crystal}

 \end{figure} 

\vspace{1mm}

Now let us calculate the Witten index in this system. There are two ways to do so.

\vspace{1mm}

{\bf 1.} One can use the Cecotti-Girardello method  and evaluate the integral \p{IW-int-final} with the classical counterpart of the Hamiltonian \p{H-chiral}. 
Integration over  fermion variables reduces to the evaluation of the determinant
\be
\det \left[\beta (\sigma_j)_\alpha{}^\beta {\cal H}_j   \right] \ =\ - \beta^2 \vecg{\cal H}^2  \ =\ -\beta^2 (\vecg{\nabla} {\cal D})^2\,.
\ee
When we also integrate over momenta, we obtain
\be
\lb{intind-chiral-QED}
I_W \ =\ \lim_{\beta \to 0} \left[ - \left( \frac \beta {8\pi^3} \right)^{1/2} \int d \vecg{c} \, (\vecg{\nabla} {\cal D})^2 e^{-\beta {\cal D}^2/2}    \right],
\ee
where the integral is evaluated over the unit cell of our ``crystal" depicted in Fig. \ref{crystal}. For small $\beta$, the integral is saturated by the regions adjacent to the singularities of ${\cal D}$ --- the sites where the monopoles sit. The contribution of each such region is equal to $-|g|$, where $g$ is the charge of the monopole. In our case, the unit cell contains a monopole of charge $-4+1/2 = -7/2$ and seven monopoles of charge $1/2$. The final result is
 \begin{empheq}[box=\fcolorbox{black}{white}]{align}
\lb{ind-chiral-SQED}
I_W \ =\ - \frac 72 - \frac 72 \ =\ -7\,.
 \end{empheq}
This calculation strongly suggests the existence of seven fermion vacuum states in this model, but it does not {\it prove} it. The matter is that the Cecotti-Girardello method that reduces the functional integral for the index to an ordinary one has its applicability limits. First of all, it does not work for  systems with continuous spectrum where the notion of index is ill-defined.  In our case, the motion is finite and the spectrum is discrete,  but we a facing another \lb{IW-sing}
problem: the Hamiltonian \p{H-chiral} has singularities $\sim 1/|\vecg{c} - \vecg{c}_{\rm node}|^2$  at the nodes of the lattice. The integral for the index is saturated by the regions of $\vecg{c}$ at the vicinity of these singularities when   $|\vecg{c} - \vecg{c}_{\rm node}| \sim \sqrt{\beta}$. But if one looks at the original functional integral \p{path-Grassmann}, one can observe that the characteristic values of the higher Fourier harmonics $\vecg{c}^{(n)}$ are also of order $\sqrt{\beta}$, and one cannot
be sure that the integral \p{intind-chiral-QED} not involving higher harmonics gives a correct value for the index.  

Indeed, we will see in the next subsection that it {\it does} not [the value of the integral \p{IW-int-final}  is not integer] for non-Abelian chiral theories,  For SQED it {\it does}, however. We know this  
  because we can confirm the result \p{ind-chiral-SQED} by a direct analysis of the Schr\"odinger equation.

\vspace{1mm}

 {\bf 2.} To do so, we deform the Hamiltonian \p{H-chiral} by adding a constant $C$ to $D$. The condition \p{cond-KAH} is still fulfilled, which
assures sypersymmetry. The index does not change under such deformation. 

 Let $C$ be large and positive, $C \gg 1/(eL)$.  The wave functions of the low-lying states are then localized in the region of small 
${\cal D}' = {\cal D}+C$, i.e. near the ``red" sites with positively charged monopoles. If $C$ is large, the sites do not ``talk" to each other, and we can pick out one of them, say, the site in the center of the cube with $\vecg{r} = \vecg{c} - (\pi/eL)(1,1,1) = \vecg{0}$.  Then the problem reduces to the solution of the Schr\"odinger equation with the Hamiltonian
\be
\lb{H-chiral-C}
\hat H  \ =\ \frac 12 (\hat {\vecg{P}} + \vecg{\cal A} )^2 + \frac 12  \left(C - \frac 1{2r} \right) ^2 + \frac {\vecg{r}}{2 r^3} \cdot\hat{\bar \eta} \vecg{\sigma} \eta\,.
  \ee
 Zero-energy states are possible only in the $F=1$ sector, where the problem reduces to the motion of a spin $\frac 12$ particle with gyromagnetic ration 2 in the field of the monopole and an additional spherically symmetric potential. The angular variables can then be separated (as is spelled out in Appendix B), and the lowest harmonic has total angular momentum $j = |g| - 1/2 = 0$. The equation for the corresponding radial function $\chi(r)$ reads
\be
\lb{Schr-rad}
- \frac 1{2r} \frac {\pd^2}{\pd r^2} (r\chi ) + \frac 12 \left(C^2 - \frac Cr - \frac 1{4r^2} \right) \chi \ =\ E\chi \,.
\ee
The  zero-energy solution of \p{Schr-rad} is readily found:
 \be
\chi(r)  \ \sim  \frac 1 {\sqrt{r}} e^{-C r} \,.
 \ee
In a particular gauge, where the ``vector potentials" are chosen in the form
\be
\lb{vecpot}
&&{\cal A}_x \ =\ - \frac {y}{2r (r + z)} \ =\  -\frac {\sin \theta \sin \phi}{2r(1 + \cos \theta)}, \nn
&&{\cal A}_y \ =\  \frac {x}{2r (r + z)} \ =\  \frac {\sin \theta \cos \phi}{2r(1 + \cos \theta)}, \nn
&&{\cal A}_z \ =\ 0\,,
 \ee
the full vacuum wave function has the form
 \be
\Psi(r,\theta,\phi;\, \eta_\alpha) \ =\ \frac 1 {\sqrt{r}} e^{-C r} \left( \eta_1 \cos \frac \theta 2   + \eta_2 \sin \frac \theta 2 e^{-i\phi} \right)\,.
 \ee
The singularity of this function at $r=0$ is benign --- the integral $\int d \vecg{r} |\Psi|^2 $ converges  quite well.\footnote{
Note that a Hamiltonian that contains an attractive singular potential $V(r) = - \gamma/r^2$ with $\gamma > 1/8$ exhibits a {\it falling on the center} phenomenon such that the ground state is absent and, which is worse, unitarity is lost \cite{falling}. Such is the Hamiltonian describing the interaction of the physical electron with the monopole \p{Malkus}. But in our case the repuslive singular potential $D^2/2$ is also present so that $\gamma$ takes the border value $\gamma = 1/8$. Then there is no collapse and the quantum problem is benign.}
According to general theorems, $\Psi$ is annihilated by the action of the supercharges \p{Q-chiral-SQED}. A dedicated reader may check it in explicit calculation.

The unit cell in Fig. \ref{crystal} contains seven sites with ``red" monopoles of charge $g = 1/2$, so that the spectrum contains seven fermionic zero-energy states. 

If the constant $C$ is large and negative, the wave function settles on the blue site with $g = -7/2$. Then the ground state has the total angular momentum $j = 3$, which produces a sevenfold degeneracy. The equation for the radial wave function now reads
\be
\lb{Schr-rad-7}
- \frac 1{2r} \frac {\pd^2}{\pd r^2} (r\chi ) + \frac 12 \left(C^2 + \frac {7C}r + \frac {35}{4r^2} \right) \chi \ =\ E\chi 
\ee
with a zero-energy solution
\be
r\chi(r) \ =\ r^{7/2} e^{Cr}\,.
\ee
The result $I_W = -7$ is reproduced also in this way.

We calculated the index of the effective Hamiltonian \p{H-chiral}. This Hamiltonian has a Wilsonian nature if the Yukawa constant lies in the range 
\p{range-h} and if the conditions 
\be
\kappa_{\vecind{n}} \ =\ \frac 1{e|L\vecg{c} - \pi \vecg{n}/e|^3} \ \ll \ 1\,.
\ee
are all fulfilled so that we are not too close to the corners of the dual torus where the monopoles sit and where the BP corrections are large.
Can we be sure that our result $I_W = -7$ is valid also for the original field theory 
\p{L-SQED-chiral} ?

Well, I can {\it swear} that it is, but I would not {\it bet} my head on that. The reason by which I still hesitate to do so is the limited range of the dynamic variables where the effective Hamiltonian \p{H-chiral} is Wilsonian. On the other hand, the original Witten's calculation of the index in SYM theories was also based on the effective Hamiltonian, which was only valid  outside the corners of the dual torus. Hypothetically, these corners could bring about surprises for SYM theories, and we cannot exclude that they do so for chiral SQED. Of course, the Hamiltonian \p{Heff-SYM} is much simpler than \p{H-chiral} and it has no singularities.
 In addition, the result $I_W = h^\vee$ was independently confirmed by other methods, as was explained in   Sec.~5.

For the chiral SQED, the Hamiltonian is more complicated (the necessity to introduce a Yukawa coupling being an extra complication), it has singularities at the corners, and we do not have an independent confirmation of the result \p{ind-chiral-SQED}.

That is where we stand today.   

 \subsection{A chiral $SU(3)$  model}

The simplest such model includes besides the vector multiplet  seven matter chiral multiplets in the fundamental group representation $\bf 3$ and a multiplet in the representation $\bar {\bf 6}$. 
In terms of superfields, the Lagrangian of the model is
\begin{empheq}[box=\fcolorbox{black}{white}]{align}
{\cal L} = \left( \frac 14 {\rm Tr} \int d^2\theta\, \hat W^\alpha \hat W_\alpha  + {\rm c.c.} \right)
 + \frac  14 \int d^4\theta \left(\sum_{f=1}^7 \overline S^{fj} e^{g\hat V} S^f_j + 
\Phi^{jk} e^{-g\hat V  } \overline \Phi_{jk} \right) 
 \end{empheq}
with $j,k = 1,2,3$ and $\Phi^{jk} = \Phi^{kj}$.

This theory is anomaly-free. Indeed, the contribution of the matter fermions in the non-Abelian anomalous triangle in Fig. \ref{triangle} is equal in this case to 
\be
\langle A^a A^b A^c \rangle \ \propto \ 7 \,{\rm Tr} \{t^{(a} t^b t^{c)} \} - {\rm Tr} \{ T^{(a} T^b T^{c)} \}\,,
 \ee
where  $t^a$ and $T^a$ are the generators in the triplet and sextet representations, correspondingly. The latter have the form
\be
(T^a)_{ij}^{kl} \ =\ \frac 12 \left[ (t^a)_i^k \delta_j^l + (t^a)_i^l \delta_j^k + (t^a)_j^k \delta_i^l + (t^a)_j^l \delta_i^k \right]\,,
\ee
where the factor $1/2$ stems from the requirement that $T^a$ satisfy the same commutation relations, $[T^a, T^b] = i f^{abc} T^c$, as $t^a$.

For any $SU(N)$ group, one can derive
\be
{\rm Tr} \{ T^{(a} T^b T^{c)} \} \ = \ (N+4) \, {\rm Tr} \{t^{(a} t^b t^{c)} \}\,.
\ee
In our case, $N+4 = 7$ and the net contribution to the anomaly vanishes.

The component Lagrangian of the model includes a quartic scalar potential, $V(s, \phi) \ = \ \frac {g^2}2 \left(s^{*f}t^a s^f - \phi T^a \phi^* \right)^2$, 
which vanishes on the valley 
 \be
\lb{D2-SU3}
 s^{*f}t^a s^f - \phi T^a \phi^* \ = \ 0\,.
 \ee
The motion along this valley brings about a continuous spectrum  in the same way as  for chiral SQED. We can also lock this valley by including in the Lagrangian the Yukawa term,\footnote{When $h=0$, the theory is asymptotically free: the effective coupling constant $g(\mu)$ runs to zero in the limit $\mu \to \infty$. But in the theory with a nonzero Yukawa coupling, the latter runs into a pole in this limit, so that this theory has the same status as QED and SQED: one can only treat it in the region of energies where both $g(\mu)$ and $h(\mu)$ are small.}
\be
{\cal L}_{\rm Yukawa}^{SU(3)}  \ = \ \frac h2 \int \!\!d^2\theta \, (S^f_j S^f_k \Phi^{jk}) \ + \ {\rm c.c.} \,.
 \ee
The corresponding scalar potential has the terms $\sim |s\phi|^2$ and $\sim |ss|^2$. 

We  repeat what we did for chiral SQED: put the theory in a  small spatial box of length $L$, expand all the fields in the Fourier series, assume the effective coupling constant $g^2(L)$ to be small and assume that the Yukawa coupling belongs in the range $g^3 \ll h \ll g$. Under these conditions, the low-energy spectrum of the theory is determined by the effective Hamiltonian in the gauge sector --- it depends on zero Fourier modes of $\vecg{A}^a(x)$ belonging to the Cartan subalgebra of $SU(3)$  and their fermion superpartners, while the Yukawa terms and the whole dynamics associated with the valley \p{D2-SU3} can be disregarded. 
In the following, we will set $g = L = 1$. 

 The effective Hamiltonian can be found along the same lines as for the Abelian theory. We have 6 bosonic slow variables that can be presented as
 \be
\lb{Cartan3}
\hat {\vecg{c}} \ =\ \frac 12 \,{\rm diag} (\vecg{a}, \vecg{b} - \vecg{a}, - \vecg{b})
\ee
with\footnote{It is the rhombus in Fig. \ref{Weyl}. For the reasons to become clear at the end of this subsection, we do not care now about the Weyl alcove.} $a_j, b_j \in [0, 4\pi)$.

Take a triplet $s_j$ and concentrate on its zero Fourier mode. In the leading approximation,  the fast Hamiltonian describing its interaction with the background gauge field $\hat {\vecg{c}}$ reads
\be 
\lb{fast-fund}
\hat H^{\rm fast}_{\rm tripl} &=& - \frac {\partial^2}{\partial  s^*_j \partial s_j }  + s^* (\hat {\vecg{c}})^2 s \nn
&=&  - \frac {\partial^2}{\partial  s^*_j \partial s_j } + \frac 14 \left[ \vecg{a}^2 s_1^* s_1 \  +\ (\vecg{a} - \vecg{b})^2 s_2^* s_2 \ + \ \vecg{b}^2 s_3^* s_3\right]
\ee
The fast ground state wave function (its bosonic part) is
\be 
\!\!\!\!\!\!\!\! \Psi(s_j, s^*_j) = \exp\left\{- \frac{|\vecg{a}|}2 |s_1|^2 \right \} \exp\left\{- \frac {|\vecg{a} - \vecg{b}|}2 |s_2|^2  \right\} 
\exp\left\{- \frac {| \vecg{b}|}2 |s_3|^2 \right\} . 
 \ee
It follows:
 \be
 \langle   |s_1|^2  \rangle = \frac 1{|\vecg{a}|}, \qquad  \langle   |s_2|^2  \rangle = \frac 1{|\vecg{a} - \vecg{b}|}, \qquad
 \langle   |s_3|^2 \rangle = \frac 1{|\vecg{b}|}\,.
 \ee
The corresponding contribution to the slow supercharge [cf. the first terms in Eq.~\p{Q-chiral-SQM}] is
\be
\lb{Q-chiral-SU3}
 \hat Q_\alpha \ =\  \eta_\gamma^a \left[  (\sigma_k)_\alpha{}^\gamma \hat P^a_k  + i  \delta_\alpha^\gamma  s^* t^a s\right]\,.
 \ee
Averaging the last term in \p{Q-chiral-SU3} over the fast bosonic vacuum, we extract the following contribution to the effective supercharge:
 \be 
&&\frac {i\eta^3_\alpha}{2} [ \langle   |s_1|^2  \rangle  -   \langle   |s_2|^2  \rangle ] + \frac {i\eta^8_\alpha}{2\sqrt{3}} [ \langle   |s_1|^2  \rangle  +   \langle   |s_2|^2  \rangle   - 2 \langle   |s_3|^2  \rangle ] \ = \\
&&\!\!\!\!\!\!\!\!\!\!\!\!\!\!\!\!\!\frac {i\eta^{(a)}_\alpha}{2} \left[\frac 1{|\vecg{a}|} - \frac 1 {|\vecg{a} - \vecg{b}|}  \right]  + \frac {i\eta^{(b)}_\alpha}{2} \left[\frac 1{|\vecg{a} - \vecg{b}|} - \frac 1 {|\vecg{b}|} 
\right]  \equiv   i \left[\eta^{(a)} {\cal D}_{(a)} + \eta^{(b)} {\cal D}_{(b)}\right]\,, \nonumber
 \ee
where\footnote{That means that 
\be
\hat \eta_\alpha  = \eta_\alpha^a t^a \ =\ \frac 12 {\rm diag} (\eta_\alpha^{(a)}, \eta_\alpha^{(b)} - \eta_\alpha^{(a)}, - \eta_\alpha^{(b)}).
\ee}
\be 
\lb{slow-triplet}
\eta^{(a)}_\alpha \ =\ \eta^3_\alpha + \frac 1{\sqrt{3}} \eta^8_\alpha, \qquad
\eta^{(b)}_\alpha \ =\  \frac 2{\sqrt{3}} \eta^8_\alpha \,.
 \ee 
Here ${\cal D}_{(a)}$ and ${\cal D}_{(b)}$ are the contributions due to the zero Fourier mode of the triplet to the induced $D$-terms. If we add the contributions of  the other triplet Fourier modes with the averages $\langle |s_1^{(\vecind{n})}|^2 \rangle = 1/|\vecg{a} - 4\pi \vecg{n}|$ etc., and multiply the result by 7 (we have 7 triplets), 
we derive in analogy with the first term in Eq. \p{D-SQED} : 
\be
\lb{Dab-triplet}
&&{\cal D}^{(a)} = \frac 72 \sum_{\vecind{n}}\left[\frac 1{|\vecg{a}_{2\vecind{n}}|} - \frac 1{|(\vecg{a} - \vecg{b})_{2\vecind{n}}|}   \right] , \nn
&&{\cal D}^{(b)}  = \frac 72 \sum_{\vecind{n}}\left[ \frac 1{|(\vecg{a} - \vecg{b})_{2\vecind{n}}|}  - \frac 1{|\vecg{b}_{2\vecind{n}}|}\right] \,,
\ee
where $\vecg{a}_{\vecind{m}} = \vecg{a} - 2\pi \vecg{m}, \ \vecg{b}_{\vecind{m}} = \vecg{b} - 2\pi \vecg{m}$.

\vspace{1mm}

Consider now the contribution of the antisextet. 
The part of the fast Hamiltonian including the zero modes of $\phi_{ij}(\vecg{x})$ is
\be
\lb{fast-sex}
\hat H^{\rm fast}_{\rm sext} \ =\  \hat {\bar\pi}^{ij} \hat \pi_{ij} + \phi^{ij} \vecg{\cal X}^2 \phi^*_{ij} ,
 \ee
where 
$$ \vecg{\cal X} \ =\ \vecg{A}^3 T^3 + \vecg{A}^8 T^8 \ =\ \vecg{a} X + (\vecg{b} - \vecg{a}) Y - \vecg{b} Z$$
with
\be
X_{ij}^{kl} \ =\ \frac 14 [ \delta_{i1} \delta^{k1} \delta_j^l + \delta_{i1} \delta^{l1} \delta_j^k +
 \delta_{j1} \delta^{k1} \delta_i^l +  \delta_{j1} \delta^{l1} \delta_i^k ], \nn  
Y_{ij}^{kl} \ =\ \frac 14 [ \delta_{i2} \delta^{k2} \delta_j^l + \delta_{i2} \delta^{l2} \delta_j^k +
 \delta_{j2} \delta^{k2} \delta_i^l +  \delta_{j2} \delta^{l2} \delta_i^k ], \nn
Z_{ij}^{kl} \ =\ \frac 14 [ \delta_{i3} \delta^{k3} \delta_j^l + \delta_{i3} \delta^{l3} \delta_j^k +
 \delta_{j3} \delta^{k3} \delta_i^l +  \delta_{j3} \delta^{l3} \delta_i^k ].    
\ee
One finds the following scalar averages:
\be
\lb{aver-6}
 \langle |\phi^{11}|^2 \rangle \ =\ \frac 1 {2|\vecg{a}|}, \qquad  \langle  |\phi^{22}|^2 \rangle \ =\ \frac 1 {2|\vecg{a} - \vecg{b}|}, \qquad \langle | \phi^{33}|^2 \rangle \ =\ \frac 1 {2|\vecg{b}|}, \nn
  \langle |\phi^{23}|^2 \rangle \ =\ \frac 1 {|\vecg{a}|}, \qquad   \langle | \phi^{12}|^2 \rangle \ =\ \frac 1 {|\vecg{b}|}, \qquad
  \langle |\phi^{13}|^2 \rangle \ =\ \frac 1 {|\vecg{a} - \vecg{b}|}.\nn
 \ee
 The slow supercharge acquires the contribution
 \be
\Delta^{\rm antisextet}  \hat Q_\alpha\ =\ -i {\eta_\alpha^a} \phi T^a \phi^*\,.
\ee 
Averaging it over the fast vacuum and taking Eq.~\p{aver-6} into account, we derive the contribution
\be
\frac {i\eta_\alpha^{(a)}}{2}\left[\frac 1{|\vecg{a}|} - \frac 1{|\vecg{a} - \vecg{b}|}   \right]
+ \frac {i\eta_\alpha^{(b)}}{2}\left[ \frac 1{|\vecg{a} - \vecg{b}|}  - \frac 1{|\vecg{b}|}    \right]
\ee
to the effective supercharge, which happens to coincide with the contribution of a triplet. 
Consider now the contribution of the higher modes $\phi^{ij}_{(\vecind{n})}$ in the fast Hamiltonian. An analysis similar to what we did for the triplets brings about the averages
\be
\lb{aver-6-n}
 \langle |\phi^{11}_{(\vecind{n})}|^2 \rangle \ =\ \frac 1 {2|\vecg{a}_{\vecind{n}}|}, \qquad  \langle  |\phi^{22}_{(\vecind{n})}|^2 \rangle \ =\ \frac 1 {2|(\vecg{a} - \vecg{b})_{\vecind{n}}|}, \qquad \langle | \phi^{33}_{(\vecind{n})}|^2 \rangle \ =\ \frac 1 {2|\vecg{b}_{\vecind{n}}|}, \nn
  \langle |\phi^{23}_{(\vecind{n})}|^2 \rangle \ =\ \frac 1 {|\vecg{a}_{2\vecind{n}}|}, \qquad   \langle | \phi^{12}_{(\vecind{n})}|^2 \rangle \ =\ \frac 1 {|\vecg{b}_{2\vecind{n}}|}, \qquad
  \langle |\phi^{13}_{(\vecind{n})}|^2 \rangle \ =\ \frac 1 {|(\vecg{a} - \vecg{b})_{2\vecind{n}}|}.\nn
 \ee
Note the different periodicity patterns for the diagonal components $\phi^{11}, \phi^{22}, \phi^{33}$ vs. the mixed components $\phi^{12}, \phi^{13}, \phi^{23}$. The former, in contrast to the latter, carry double charges with respect to $T^3$ and/or $T^8 \sqrt{3} - T^3$ and play here the same role as the field $t$ in the Abelian theory.

Plugging \p{aver-6-n} into the fast vacuum averages of the slow supercharge and adding the triplet contribution \p{Dab-triplet}, we derive:
\begin{empheq}[box=\fcolorbox{black}{white}]{align}
\lb{Dab}
\!\!\!\!\!{\cal D}_{(a)} &=& \frac 92 \sum_{\vecind{n}}\left[\frac 1{|\vecg{a}_{2\vecind{n}}|} - \frac 1{|(\vecg{a} - \vecg{b})_{2\vecind{n}}|}   \right] \ - \  \frac 12     \sum_{\vecind{n}} \left[\frac 1{|\vecg{a}_{\vecind{n}}|} - \frac 1{|(\vecg{a} - \vecg{b})_{\vecind{n}}|} \right] , \nn
\!\!\!\!\!{\cal D}_{(b)} &=& \frac 92 \sum_{\vecind{n}}\left[ \frac 1{|(\vecg{a} - \vecg{b})_{2\vecind{n}}|}  - \frac 1{|\vecg{b}_{2\vecind{n}}|}\right] \ - \ \frac 12 \sum_{\vecind{n}}\left[ \frac 1{|(\vecg{a} - \vecg{b})_{\vecind{n}}|}  - \frac 1{|\vecg{b}_{\vecind{n}}|}\right].
\end{empheq}
 The second terms are the contributions of the sextet components $\phi^{11},  \phi^{22}, \phi^{33}$. The corresponding lattice is twice as dense as the lattice for the triplets and for the components $\phi^{12},  \phi^{13}, \phi^{23}$. The reason is the same as for chiral SQED. For a large gauge transformation when, say, $a_x$ is shifted by $4\pi$, a sextet matter fields $\phi^{11}, \phi^{22}, \phi^{33}$  (but not $\phi^{12}, \phi^{13}, \phi^{23}$)  have to ``make two turns" being multiplied by $e^{4i\pi x}$ rather than by $e^{2i\pi x}$ to compensate this shift, in the same way as the fields of double charge did in the Abelian model.
 
As in the Abelian model,  the effective supercharges acquire besides the induced ${\cal D}$ terms the ``vector potential" contributions coming from the averages $\langle \pd/\pd \vecg{a} \rangle$ and  $\langle \pd/\pd \vecg{b} \rangle$ over the fermion factor in the fast ground state wave function. As a result, we derive the following expressions for the effective supercharges:
 \be
\lb{Q-chiral}
   \hat Q_\alpha^{\rm eff} \ &=& \  {\eta^{(a)}_\beta}\left[(\sigma_k)_\alpha{}^\beta (\hat P_k^{(a)} + {\cal A}^{(a)}_k) 
+i  \delta_\alpha^\beta {\cal D}_{(a)}    \right] \nn
&+& \   {\eta^{(b)}_\beta} \left[(\sigma_k)_\alpha{}^\beta (\hat P_k^{(b)} + {\cal A}^{(b)}_k) +i  \delta_\alpha^\beta {\cal D}_{(b)}    \right],  \nn
\hat {\bar Q}^{\alpha \, {\rm eff}} \ &=&\    {\hat{\bar \eta}^{(a)\beta}}\left[(\sigma_k)_\beta{}^\alpha (\hat P_k^{(a)} 
 + {\cal A}^{(a)}_k) + i \delta_\beta^\alpha {\cal D}_{(a)}    \right] \nn
&+&\ {\hat{\bar \eta}^{(b)\beta}} \left[(\sigma_k)_\beta{}^\alpha (\hat P_k^{(b)} + {\cal A}^{(b)}_k) + i \delta_\beta^\alpha {\cal D}_{(b)}    \right],
 \ee
where $\hat P_k^{(a)} = -i \pd/\pd a^k$, $\hat P_k^{(b)} = -i \pd/\pd b^k$ and $\vecg{\cal A}^{(a,b)}$ are related to  ${\cal D}_{(a,b)}$  so that
\be
\lb{cond-nab}
{\cal H}^{(a)}_j &\equiv& \varepsilon_{jkl}\, \pd_k^a {\cal A}^{(a)}_l \ =\ \pd_j^a {\cal D}_{(a)}\,, \nn
{\cal H}^{(b)}_j &\equiv& \varepsilon_{jkl} \,\pd_k^b {\cal A}^{(b)}_l \ =\ \pd_j^b {\cal D}_{(b)}\,, \nn
\varepsilon_{jkl} {\cal H}^{(ab)}_j  \ &\equiv& \ \pd_k^a {\cal A}^{(b)}_l  - \pd_l^b {\cal A}^{(a)}_k  \ =\   \varepsilon_{jkl}  \pd_j^a {\cal D}_{(b)}  \ =\   \varepsilon_{jkl}  \pd_j^b {\cal D}_{(a)}  \,.
 \ee
Bearing in mind that
\be
\lb{anticom-eta}
 \{\hat{\bar \eta}^{(a)\beta}, \eta^{(a)}_\alpha\} = \{\hat{\bar \eta}^{(b)\beta}, \eta^{(b)}_\alpha\} = \frac 43 \delta_\alpha^\beta,
 \quad \{\hat{\bar \eta}^{(a)\beta}, \eta^{(b)}_\alpha\} = \{\hat{\bar \eta}^{(b)\beta}, \eta^{(a)}_\alpha\} = \frac 23 \delta_\alpha^\beta\,, \nn
 \ee
we derive the effective Hamiltonian:\footnote{The effective supercharges \p{Q-chiral} and the Hamiltonian \p{H-chiral-nab} were derived in Ref. \cite{Blok}. Unfortunately, the expressions for the induced ${\cal D}$ terms quoted in that paper involved a mistake. As a result, the value of the phase space integral for the Witten index was also not correct. But we will confirm in the calculations below the main conclusion of that work: the phase space integral \p{IW-int-final} evaluated for the Hamiltonian \p{H-chiral-nab} has a fractional value, so that the CG method to evaluate the index does not work in this case.}
\begin{empheq}[box=\fcolorbox{black}{white}]{align}
\lb{H-chiral-nab}
\hat H^{\rm eff}\  = \ \frac 23 \left[(\hat P_k^{(a)} + {\cal A}^{(a)}_k )^2 + (\hat P_k^{(b)} + {\cal A}^{(b)}_k )^2 + (\hat P_k^{(a)} + 
{\cal A}^{(a)}_k ) (\hat P_k^{(b)} + {\cal A}^{(b)}_k) \right.  \nn
+ \left. {\cal D}_{(a)}^2 + {\cal D}_{(b)}^2 
   + {\cal D}_{(a)} {\cal D}_{(b)}\right]\nn 
 + \, \vecg{\cal H}^{(a)}\, \hat{\bar \eta}^{(a)} \vecg{\sigma} \eta^{(a)} +  \vecg{\cal H}^{(b)}\, \hat{\bar \eta}^{(b)} \vecg{\sigma} \eta^{(b)} +  \vecg{\cal H}^{(ab)}( \hat{\bar \eta}^{(a)} \vecg{\sigma} \eta^{(b)} 
+  \hat{\bar \eta}^b \vecg{\sigma} \eta^{(a)} )\,.
  \end{empheq}

\subsubsection{Phase space integral}
The Hamiltonian \p{H-chiral-nab} is essentially more complicated than \p{H-chiral}, and
it is difficult in this case to find the vacuum states explicitly, as we did for SQED. So we use the Cecotti-Girardello method and evaluate the integral \p{IW-int-final} for the classical version of the Hamiltonian \p{H-chiral-nab}. This calculation is  more complicated than in the Abelian case, but is still feasible.

Note first of all that the integral \p{IW-int-final} was written in the assumption that the bosonic and fermion variables form a canonical basis with the ordinary Poisson brackets
 \be
\{p_j, q_k\}_{\rm P.B.} \ =\ \delta_{jk}, \qquad \{\psi_a, \bar \psi_b \}_{\rm P.B.} \ =\ i\delta_{ab} \,.
 \ee
 In our case, it is so for bosonic variables, but not so for the fermion variables whose Poisson brackets are nontrivial, as follows from \p{anticom-eta}. The Jacobian of the transformation from the canonical variables $\eta_\alpha^3,   \eta_\alpha^8, \bar \eta^{\alpha 3}, \bar \eta^{\alpha 8}$ to 
$\eta_\alpha^{(a)},   \eta_\alpha^{(b)}, \bar \eta^{\alpha (a)}, \bar \eta^{\alpha (b)}$ gives the factor $(2/\sqrt{3})^4 = 16/9$. Integrating over the fermions, we derive
 \be
\lb{IW-int-nab}
&&I_W  =  \ \lim_{\beta \to 0} \frac {16}9 \beta^4  \int \frac {d\vecg{a} \,d\vecg{P^{(a)}}d\vecg{b}\, d\vecg{P^{(b)}}  }{(2\pi)^6}  \left| \begin{array}{cc} \vecg{\cal H}^{(a)} 
\vecg{\sigma}, &   \vecg{\cal H}^{(ab)} 
\vecg{\sigma} \\
\vecg{\cal H}^{(ab)} 
\vecg{\sigma}, & \vecg{\cal H}^{(b)} 
\vecg{\sigma} \end{array}  \right| \nn
&&\cdot \exp\left\{- \frac 23 \beta \left[ (\vecg{P}^{(a)} + \vecg{\cal A}^{(a)}_k )^2 + (\vecg{P}^{(b)} + \vecg{\cal A}^{(b)}_k )^2 + (\vecg{P}^{(a)} + 
\vecg{\cal A}^{(a)}_k ) ( \vecg{P}^{(b)} + \vecg{\cal A}^{(b)}_k)  \right. \right. \nn
&&\left. \left. + {\cal D}_{(a)}^2 + {\cal D}_{(b)}^2 
   + {\cal D}_{(a)} {\cal D}_{(b)} \right]      \right\}.
\ee
Calculating the determinant and integrating over momenta, we obtain
\begin{empheq}[box=\fcolorbox{black}{white}]{align}
\lb{intind-DH}
&&I_W  =  \lim_{\beta \to 0} \frac \beta{4\pi^3 \sqrt{3}} \int d\vecg{a} d\vecg{b} \left[ \|\vecg{H}^{(ab)}\|^4 + 2  \|\vecg{H}^{(ab)}\|^2 (\vecg{H}^{(a)} \cdot 
\vecg{H}^{(b)})  \right. \nn
&&\left. + \ \|\vecg{H}^{(a)}\|^2 \|\vecg{H}^{(b)}\|^2 - 4(\vecg{H}^{(a)} \cdot 
\vecg{H}^{(ab)}) (\vecg{H}^{(b)} \cdot 
\vecg{H}^{(ab)}) \right]\nn
&& \times \exp\left\{ - \frac 23 \beta  \left[D_{(a)}^2 + D_{(b)}^2 
   + D_{(a)} D_{(b)}\right]\right\}. 
\end{empheq}
The integral runs over the fundamental cell $ a_j, b_j \in [0, 4\pi)$ of our ``crystal" (recall that the gauge invariant slow wave functions are periodic in the space of weights with the period $4\pi$ for each component of $\vecg{a}$ and $\vecg{b}$).
If $\beta$ is small, the integral is saturated by the regions near the singularities of ${\cal D}_{(a)}$ {\it and} ${\cal D}_{(b)}$. There are 64  such regions in the fundamental cell: 

{\it (i)} its corner  $\vecg{a} \sim \vecg{b} \sim 0$; 

{\it (ii)} 42 symmetric regions  of the type $\vecg{a} \sim \vecg{a}_0 = 2\pi(1,0,0),  \vecg{b} \sim \vecg{b}_0 = 2\pi(0,1,0)$ (so that all three vectors $\vecg{a}_0, \vecg{b}_0, \vecg{a}_0 - \vecg{b}_0$ are nonzero); and

{\it (iii)}   21 asymmetric regions of the type $\vecg{a} \sim 2\pi(1,0,0), \vecg{b} \sim \vecg{0}$ where one of these vectors is equal to zero.   

Consider first the region around the corner. The angular integration gives \be 
d\vecg{a} d\vecg{b} \ =\ 8\pi^2 \rho_a d \rho_a \rho_b d \rho_b \rho_{ab} d \rho_{ab}
\ee  
with $\rho_a = |\vecg{a}|, \, \rho_b = |\vecg{b}|, \rho_{ab} = |\vecg{a} - \vecg{b}|$. The induced ${\cal D}$ terms are 
\be
\lb{D-cent}
{\cal D}_{(a)} \ =\ 4\left( \frac 1{\rho_a} - \frac 1{\rho_{ab}} \right), \qquad {\cal D}_{(b)} \ =\ 4\left( \frac 1{\rho_{ab}} - \frac 1{\rho_b} \right).
\ee

Substituting this into \p{intind-DH} and  bearing in mind \p{cond-nab},  we obtain the following expression for the corner contribution:\footnote{The integral does not depend on $\beta$ and we have set $ \beta = 1/16$.} 

\be
\lb{I-cent-int}
&&I_W({\rm corner})  =  \  \frac {32} {\pi \sqrt{3}} \iiint_0^\infty \frac { d \rho_a  d \rho_b  d \rho_{ab}}{\rho_a^3 \rho_b^3 \rho_{ab}^3 }
\left[ \rho_a^4 + \rho_b^4 + \rho_{ab}^4 + \right. \nn
&&\left. \rho_a \rho_b (\rho_{ab}^2 - \rho_a^2 - \rho_b^2) +  \rho_a \rho_{ab} (\rho_{b}^2 - \rho_a^2 - \rho_{ab}^2) +
\rho_b \rho_{ab} (\rho_{a}^2 - \rho_b^2 - \rho_{ab}^2) \right] \nn
&&\cdot \exp \left[    - \frac 23 \left(\frac 1 {\rho_a^2} + \frac 1 {\rho_b^2} + \frac 1 {\rho_{ab}^2}   - \frac 1{\rho_a \rho_b} -  \frac 1{\rho_a \rho_{ab}}-   \frac 1{\rho_b \rho_{ab}} \right)\right]
\ee  
with the constraint $|\rho_a - \rho_b| \leq \rho_{ab} \leq \rho_a + \rho_b$.

It is convenient to make a variable change 
\be
\rho_a = \mu \rho, \quad \rho_b = \nu \rho, \quad \rho_{ab} = (1-\mu - \nu)\rho\,,
 \ee 
after which the integral over $\rho$ can be easily done, and we derive
 \be
I_W({\rm corner}) \  =  \ \frac {8\sqrt{3}}{\pi} \int_0^{1/2} d\mu \int_{1/2 -\mu}^{1/2} d\nu \,\frac {A(\mu,\nu)}{\mu\nu\kappa B(\mu,\nu)}\,, 
\ee
where $\kappa = 1 - \mu - \nu$ and
\be
\lb{for-Math}
\!\!\!\!\!\!\!A(\mu,\nu) &=& \mu^4 + \nu^4 + \kappa^4  + \mu\nu(\kappa^2 - \mu^2 - \nu^2) + \mu\kappa(\nu^2 - \mu^2 - \kappa^2)
+ \nu\kappa (\mu^2 - \nu^2 - \kappa^2)\,,  \nn
\!\!\!\!\!\!\!B(\mu, \nu) &=& \mu^2 \nu^2 + \mu^2 \kappa^2 + \nu^2 \kappa^2 - \mu \nu \kappa\,.
 \ee
Mathematica calculated for us the integral \p{for-Math} numerically, and this calculation produced an  integer answer: 
\be
\lb{I-cent}
I_W({\rm corner}) \ =\ 16\,.
 \ee

So far so good. Consider now the contributions of 42 symmetric regions. The corresponding ${\cal D}$ terms have the same form as in Eq.~\p{D-cent} with $\rho_a = |\vecg{a} - \vecg{a}_0|$ etc., but with the factor $-
1/2$ rather than 4. Also the ``magnetic fields" $\vecg{\cal H}^{(a)}, \vecg{\cal H}^{(b)}$ and $\vecg{\cal H}^{(ab)}$ entering the preexponential in \p{intind-DH}  are 8 times smaller than in the corner. As a result, the contribution of a symmetric region in the index is equal to $16/64 = 1/4$. And the contribution of all such regions is
\be
\lb{I-sym}
I_W^{\rm sym} \ =\ \frac{42}{4} \ =\ 10\frac 12 \,.
 \ee
It is {\it fractional}.

Consider now the 21 asymmetric regions. For example, a region around the singularity at $\vecg{a}_0 =  2\pi(1,0,0), \vecg{b}_0 = \vecg{0}$. The $\cal D$ terms are
\be
\lb{D-sym}
{\cal D}_{(a)} \ =\ \frac 12 \left( \frac 1{\rho_{ab}} - \frac 1{\rho_{a}} \right), \qquad {\cal D}_{(b)} \ =\  -\frac 1{2\rho_{ab}} - \frac 4{\rho_b} 
\ee
with $\rho_a = |\vecg{a} - \vecg{a}_0|, \rho_b = |\vecg{b}|$, and $\rho_{ab} = |\vecg{a} - \vecg{b} - \vecg{a}_0|$.
The contribution of this region to the index is\footnote{We have now set $\beta = 4$.}
\be
\lb{I-asym-int}
&&I_W^{\rm asym}  =  \  \frac 1 {2\pi \sqrt{3}} \iiint \frac { d \rho_a  d \rho_b  d \rho_{ab}}{\rho_a^3 \rho_b^3 \rho_{ab}^3 }
\left[ 64(\rho_a^4 +  \rho_{ab}^4) + \rho_b^4   \right. \nn
&&\left. 8\rho_a \rho_b (\rho_a^2 + \rho_b^2 - \rho_{ab}^2) +  64 \rho_a \rho_{ab} (\rho_{b}^2 -\rho_a^2 - \rho_{ab}^2  ) +
8\rho_b \rho_{ab} (\rho_b^2 + \rho_{ab}^2 - \rho_{a}^2) \right] \nn
&&\cdot \exp \left[    - \frac 23 \left(\frac 1 {\rho_a^2} + \frac 1 {\rho_{ab}^2} + \frac {64}{\rho_{b}^2}   - \frac 1{\rho_a \rho_{ab}} +  \frac 8 {\rho_a \rho_{b}} +  \frac 8{\rho_b \rho_{ab}} \right)\right].
\ee 
A numerical calculation gives the value $3$ for this integral.\footnote{The fact that a numerical calculation of these complicated integrals gives nicely looking simple fractions or even integers looks as a small mathematical miracle, of which I have no explanation. Moreover, consider a set of integrals 
\be
\lb{I-q}
&&I_{q}  =  \  \frac 1 {
2\pi \sqrt{3}} \iiint \frac { d \rho_a  d \rho_b  d \rho_{ab}}{\rho_a^3 \rho_b^3 \rho_{ab}^3 }
\left[ q^2(\rho_a^4 +  \rho_{ab}^4) + \rho_b^4   \right. \nn
&&\left. q\rho_a \rho_b (\rho_{ab}^2 - \rho_a^2 - \rho_b^2) +  q^2 \rho_a \rho_{ab} (\rho_{b}^2 - \rho_a^2 - \rho_{ab}^2) +
q\rho_b \rho_{ab} (\rho_{a}^2 - \rho_b^2 - \rho_{ab}^2) \right] \nn
&&\cdot \exp \left[    - \frac 23 \left(\frac 1 {\rho_a^2} + \frac 1 {\rho_{ab}^2} + \frac {q^2}{\rho_{b}^2}   - \frac 1{\rho_a \rho_{ab}} -  \frac q {\rho_a \rho_{b}}-   \frac q{\rho_b \rho_{ab}} \right)\right]
\ee 
with integer $q$. The ``experimental" values of these integrals are 
$$ I_{q \leq 0}  =\frac {1-q}3,  \qquad I(1) = \frac 14, \qquad I_{q >1} = \frac {2(q-1)}3 \,.
$$
I wonder if one can  derive these amusing results analytically.}  
The contribution of all 21 asymmetric regions is 
 \be
\lb{I-asym}
I_W^{\rm asym} \ =\ 63\,.
 \ee
Adding all the contributions, we obtain
\begin{empheq}[box=\fcolorbox{black}{white}]{align}
I_W \ =\ 89\frac 12\,.
\end{empheq}
Actually, the notation $I_W$ is not correct here, because the Witten index cannot be fractional in a system with discrete spectrum (the range of bosonic dynamic variables $\vecg{a}$ and $\vecg{b}$ in the Hamiltonian \p{H-chiral-nab} is finite, and the spectrum should be discrete). We have not calculated the index, but just quoted the value of the integral \p{IW-int-nab}. The fractional value of the integral is not a paradox by the reasons clarified on p. \pageref{IW-sing}: the effective Hamiltonian \p{H-chiral-nab} has singularities so that the Cecotti-Girardello method by which the index was evaluated above may not work. Unfortunately, we cannot say why it works for chiral SQED, but does not work for chiral SQCD.\footnote{Were the result of the calculation of the phase space integral integer, we would have to recall that not all the eigenfunctions of the  Hamiltonian \p{H-chiral-nab} contribute in the index of the original theory, but only the functions invariant under Weyl permutations: $\vecg{a} \leftrightarrow \vecg{b} - \vecg{a}$, etc. But the fractional result makes the discussion of this subtle question pointless. }   

\subsection{Minimal $SU(5)$ model}

The simplest chiral anomaly-free $SU(5)$ supersymmetric model includes a quintet $S_j$ and an antidecuplet $\Phi^{[jk]}$. 
The generators in the antisymmetric representation of an $SU(N)$ group read 
\be
\lb{Ta-SU5}
(T^a)_{ij}^{kl} \ =\ \frac 12 \left[ (t^a)_i^k \delta_j^l - (t^a)_i^l \delta_j^k - (t^a)_j^k \delta_i^l + (t^a)_j^l \delta_i^k \right]\,.
\ee
One can derive that 
\be
{\rm Tr} \{ T^{(a} T^b T^{c)} \} \ = \ (N-4) \, {\rm Tr} \{t^{(a} t^b t^{c)} \}\,,
\ee
and, for $N=5$, the contributions of the quintet and antidecuplet in the anomalous triangle cancel, indeed.

The Lagrangian reads
\begin{empheq}[box=\fcolorbox{black}{white}]{align}
\lb{chiral-SU5}
{\cal L} = \left( \frac 14 {\rm Tr} \int d^2\theta\, \hat W^\alpha \hat W_\alpha  + {\rm c.c.} \right) 
 + \frac  14 \int d^4\theta \left( \overline S e^{g\hat V} S + 
\Phi e^{-g\hat V  } \overline \Phi \right) .
 \end{empheq}

An important distinction of this model compared to the models considered in the previous sections is the {\it absence} of the scalar valley.
In other words, the equation system
 \be
\lb{Da-SU5}
\phi^{ij} (T^a)_{ij}{}^{kl} \phi^*_{kl} - s^{*j} (t^a)_j{}^k s_k \ =\ 0
 \ee
with $T^a$ written in Eq. \p{Ta-SU5} has no solutions. 

One can be convinced in that quite explicitly.\footnote{See Appendix A of the review \cite{Amati}.} 
Introduce the matrices $A_k{}^j = s_k s^{*j}$ and $B_k{}^j = \phi^*_{kl} \phi^{lj}$. Then \p{Da-SU5} boils down to Tr$\{(A + 2B) t^a \} = 0$, i.e. $A+2B = C \cdot \mathbb{1}$.
By a gauge rotation, $\phi^{ij}$ can be brought to the following canonical form:
  \be
\phi \ =\ \left(  \begin{array}{ccccc} 0&x&0&0&0 \\ -x&0&0&0&0 \\ 0&0&0&y&0 \\ 0&0&-y&0&0 \\ 0&0&0&0&0 \end{array} \right).
 \ee
This gives
$$ B = \phi^* \phi \ =\ -{\rm diag}\,(|x|^2, |x|^2, |y|^2, |y|^2, 0 ) \,. $$
Hence 
\be s_j s^{*k} \ =\ {\rm diag}\,(C+2|x|^2, C+ 2|x|^2, C+ 2|y|^2, C+ 2|y|^2, C) \,, \ee
and this relation treated as an equation for $s_j$ has no solutions.

Thus, we do not need here to introduce a Yukawa coupling to lock a nonexistent scalar valley (besides, we are not {\it able} to do so: in this case, one cannot write a trilinear gauge-invariant term in the superfield Lagrangian). The theory has only the gauge coupling and is asymptotically free. To study its vacuum structure, it suffices to study the gauge branch of classical vacua.

Take a constant Abelian  background 
\be
\lb{Cartan5}
\hat {\vecg{c}} \ =\ \frac 12 \,{\rm diag} (\vecg{a}, \vecg{b} - \vecg{a},  \vecg{c} - \vecg{b},  \vecg{d} - \vecg{c},- \vecg{d})\,.
\ee
The fast Hamiltonian describing its interaction with the zero Fourier modes of the scalar quintet components is
\be 
\lb{fast-fund-5}
&&H^{\rm fast}_{\rm quint} \ = 
 - \frac {\partial^2}{\partial  s^{*j} \partial s_j } 
+ \frac 14 \left[ \vecg{a}^2 |s_1|^2 \right. \nn 
&&\left. + (\vecg{b} - \vecg{a})^2 |s_2|^2 + (\vecg{c} - \vecg{b})^2 |s_3|^2 
+ (\vecg{d} - \vecg{c})^2 |s_4|^2  + \vecg{d}^2 |s_5|^2\right].
\ee
The scalar averages are 
\be
 &&\langle    |s_1|^2 \rangle = \frac 1{|\vecg{a}|}, \qquad  \langle  |s_2|^2   \rangle = \frac 1{|\vecg{b} - \vecg{a}|}, \qquad
\langle  |s_3|^2   \rangle = \frac 1{|\vecg{c} - \vecg{b}|}, \nn
 &&\quad \langle  |s_4|^2   \rangle = \frac 1{|\vecg{d} - \vecg{c}|}, \qquad
 \langle   |s_5|^2 \rangle = \frac 1{|\vecg{d}|}\,.
 \ee
The fast Hamiltonian describing the interaction of the gauge background with the zero Fourier modes of the scalar decuplet components is
\be 
\lb{fast-fund-10}
&&H^{\rm fast}_{\rm dec} \ = 
 2\left[ |\pi_{12}|^2 + \ldots + |\pi_{45}|^2 \right] 
+ \frac 12 \left[ \vecg{b}^2 |\phi^{12}|^2 + ( \vecg{a} +  \vecg{c} - \vecg{b})^2 |\phi^{13}|^2  \right.\nn
&& + (\vecg{a} +  \vecg{d} - \vecg{c})^2 |\phi^{14}|^2 
 + ( \vecg{a} - \vecg{d})^2 |\phi^{15}|^2   + (\vecg{a} -  \vecg{c})^2 |\phi^{23}|^2 \nn
&& + (\vecg{b} -  \vecg{a} + \vecg{d}- \vecg{c})^2 |\phi^{24}|^2 
 + (\vecg{b} -  \vecg{a} - \vecg{d})^2 |\phi^{25}|^2 
+ (\vecg{d} -\vecg{b})^2 |\phi^{34}|^2  \nn
 &&  \left. + (\vecg{c} -  \vecg{b} - \vecg{d})^2 |\phi^{35}|^2 +\vecg{c}^2 |\phi^{45}|^2  \right].
\ee

The scalar averages are 
\be
&&\langle |\phi^{12}|^2 \rangle = \frac 1 {|\vecg{b}|}, \quad \langle |\phi^{23}|^2 \rangle = \frac 1 {|\vecg{a} - \vecg{c}|}, \quad
 \langle |\phi^{15}|^2\rangle  = \frac 1 {|\vecg{a}  - \vecg{d} |}, \quad \langle |\phi^{45}|^2 \rangle = \frac 1 {|\vecg{c}|},\nn 
 &&   \langle |\phi^{13}|^2 \rangle = \frac 1 {|\vecg{a} + \vecg{c} - \vecg{b} |}, \ 
\langle |\phi^{14}|^2 \rangle = \frac 1 {|\vecg{a} + \vecg{d} - \vecg{c} |}, \
 \langle |\phi^{24}|^2\rangle  = \frac 1 {|\vecg{b} - \vecg{a} + \vecg{d}- \vecg{c} |}, \nn 
&&\langle |\phi^{25}|^2 \rangle = \frac 1 {|\vecg{b} - \vecg{a} - \vecg{d} |},   \ 
 \langle |\phi^{34}|^2 \rangle = \frac 1 {| \vecg{b}-  \vecg{d} |}, \
\langle |\phi^{35}|^2 \rangle = \frac 1 {|\vecg{c} - \vecg{b} - \vecg{d} |}. \nn     
\ee
\smallskip
Taking into account also higher Fourier harmonics of the matter fields and proceeding in the same way as we did  for the $SU(3)$ model, 
we can find the effective supercharges and  Hamiltonian \cite{Blok}.
\be
\lb{Q-chiral-5}
   \hat Q_\alpha^{\rm eff} \ &=& \ {\eta^{(a)}_\beta} \left[(\sigma_k)_\alpha{}^\beta (\hat P_k^{(a)} + {\cal A}^{(a)}_k) 
+i  \delta_\alpha^\beta {\cal D}_{(a)}    \right] \nn
&+& \  {\eta^{(b)}_\beta} \left[(\sigma_k)_\alpha{}^\beta (\hat P_k^{(b)} + {\cal A}^{(b)}_k) +i  \delta_\alpha^\beta {\cal D}_{(b)}    \right],  \nn
&+&  {\eta^{(c)}_\beta} \left[(\sigma_k)_\alpha{}^\beta (\hat P_k^{(c)} + {\cal A}^{(c)}_k) 
+i  \delta_\alpha^\beta {\cal D}_{(c)}    \right] \nn
&+& \  {\eta^{(d)}_\beta} \left[(\sigma_k)_\alpha{}^\beta (\hat P_k^{(d)} + {\cal A}^{(d)}_k) +i  \delta_\alpha^\beta {\cal D}_{(d)}    \right]
 \ee
and similarly for $ \hat {\bar Q}^{\alpha\, {\rm eff}}$ expressed via $\bar \eta^{(a)\alpha}$ etc. Here
\be
\eta_\alpha^{(a)} \ &=&\ \eta_\alpha^3 + \frac {\eta_\alpha^8}{\sqrt{3}} + \frac {\eta_\alpha^{15}}{\sqrt{6}} + \frac {\eta_\alpha^{24}}{\sqrt{10}}\,, \nn
\eta_\alpha^{(b)} \ &=&\  2\left(\frac {\eta_\alpha^8}{\sqrt{3}} + \frac {\eta_\alpha^{15}}{\sqrt{6}} + \frac {\eta_\alpha^{24}}{\sqrt{10}} \right), \nn
\eta_\alpha^{(c)} \ &=&\  3\left(\frac {\eta_\alpha^{15}}{\sqrt{6}} + \frac {\eta_\alpha^{24}}{\sqrt{10}} \right), \nn
\eta_\alpha^{(d)} \ &=&\ \frac {4\eta_\alpha^{24}}{\sqrt{10}}\,. 
 \ee

The fermion variables have the following anticommutators:
\be
\lb{anticom-eta-5}
&& \{\hat{\bar \eta}^{(a)\beta}, \eta^{(a)}_\alpha\} = \{\hat {\bar \eta}^{(d)\beta}, \eta^{(d)}_\alpha\} = \frac 85 \delta_\alpha^\beta,
 \quad \{\hat{\bar \eta}^{(b)\beta}, \eta^{(b)}_\alpha\} = \{\hat{\bar \eta}^{(c)\beta}, \eta^{(c)}_\alpha\} = \frac {12}5 \delta_\alpha^\beta\,, \nn
 &&\{\hat{\bar \eta}^{(a)\beta}, \eta^{(b)}_\alpha\} = \{\hat{\bar \eta}^{(b)\beta}, \eta^{(a)}_\alpha\} = \{\hat{\bar \eta}^{(c)\beta}, \eta^{(d)}_\alpha\} = \{\hat{\bar \eta}^{(d)\beta}, \eta^{(c)}_\alpha\} = \frac 65\delta_\alpha^\beta, \nn
&& \{\hat{\bar \eta}^{(a)\beta}, \eta^{(c)}_\alpha\} = \{\hat{\bar \eta}^{(c)\beta}, \eta^{(a)}_\alpha\} = \{\hat{\bar \eta}^{(b)\beta}, \eta^{(d)}_\alpha\} = \{\hat{\bar \eta}^{(d)\beta}, \eta^{(b)}_\alpha\} = \frac 45\delta_\alpha^\beta, \nn
 &&\{\hat{\bar \eta}^{(a)\beta}, \eta^{(d)}_\alpha\} = \{\hat{\bar \eta}^{(d)\beta}, \eta^{(a)}_\alpha\} = \frac 25 \delta_\alpha^\beta\, \quad
\{\hat{\bar \eta}^{(b)\beta}, \eta^{(c)}_\alpha\} = \{\hat{\bar \eta}^{(c)\beta}, \eta^{(b)}_\alpha\} = \frac 85 \delta_\alpha^\beta\,.\nn
 \ee

The supercharges $ \hat Q_\alpha^{\rm eff}$ and $ \hat {\bar Q}^{\alpha\, {\rm eff}}$ look similar to \p{Q-chiral}, but they now include extra dynamic variables ($\vecg{c}, \vecg{d}$ and their superpartners)  and represent sums of four rather than two terms.
 The induced ${\cal D}$ terms are 
\be
{\cal D}_{(a)} = \frac 12 \left(\langle |s_1|^2\rangle  -  \langle |s_2|^2\rangle  - \langle |\phi_{13}|^2\rangle  - \langle |\phi_{14}|^2\rangle 
-  \langle |\phi_{15}|^2\rangle +   \langle |\phi_{23}|^2\rangle +  \langle |\phi_{24}|^2 +  \langle |\phi_{25}|^2\rangle  
\right), \nn
{\cal D}_{(b)} = \frac 12 \left(\langle |s_2|^2\rangle  -  \langle |s_3|^2\rangle  - \langle |\phi_{21}|^2\rangle  - \langle |\phi_{24}|^2\rangle 
-  \langle |\phi_{25}|^2\rangle +   \langle |\phi_{31}|^2\rangle +  \langle |\phi_{34}|^2 +  \langle |\phi_{35}|^2\rangle  
\right), \nn
{\cal D}_{(c)} = \frac 12 \left(\langle |s_3|^2\rangle  -  \langle |s_4|^2\rangle  - \langle |\phi_{31}|^2\rangle  - \langle |\phi_{32}|^2\rangle 
-  \langle |\phi_{35}|^2\rangle +   \langle |\phi_{41}|^2\rangle +  \langle |\phi_{42}|^2 +  \langle |\phi_{45}|^2\rangle  
\right), \nn
{\cal D}_{(d)} = \frac 12 \left(\langle |s_4|^2\rangle  -  \langle |s_5|^2\rangle  - \langle |\phi_{41}|^2\rangle  - \langle |\phi_{42}|^2\rangle 
-  \langle |\phi_{43}|^2\rangle +   \langle |\phi_{51}|^2\rangle +  \langle |\phi_{52}|^2 +  \langle |\phi_{53}|^2\rangle  
\right)\!, 
\ee
which gives  (when we take into account the contributions of all the Fourier harmonics) 
\be
&&\!\!\!\!{\cal D}_{(a)}  = \frac 12 \sum_{\vecind{n}}\left[  \frac 1{|\vecg{a}_{2\vecind{n}}|} + \frac 1 {|(\vecg{a} - \vecg{c})_{2\vecind{n}|} } + \frac 1 
{|(\vecg{b}- \vecg{a} + \vecg{d}- \vecg{c}) _{2\vecind{n}} |} +  \frac 1 {|(\vecg{b} - \vecg{a} - \vecg{d})_{2\vecind{n}} | }  \right. \nn
&&\left.  - \frac 1{|(\vecg{a} - \vecg{b})_{2\vecind{n}}|} - \frac 1 {|(\vecg{a} + \vecg{c} - \vecg{b})_{2\vecind{n}}| } - \frac 1 {|(\vecg{a} + \vecg{d} - \vecg{c})_{2\vecind{n}} |} - \frac 1 {|(\vecg{a} - \vecg{d})_{2\vecind{n}} | } \right],\nn
&&\!\!\!\!{\cal D}_{(b)}  = \frac 12 \sum_{2\vecind{n}}\left[  \frac 1{|(\vecg{a}- \vecg{b})_{2\vecind{n}}|} + \frac 1 {|(\vecg{b} - \vecg{d})_{2\vecind{n}|} } + \frac 1 
{|(\vecg{a}+ \vecg{c} - \vecg{b}) _{2\vecind{n}} |} +  \frac 1 {|(\vecg{b} - \vecg{c} + \vecg{d})_{2\vecind{n}} | }  \right. \nn
&&\left.  - \frac 1{|\vecg{b}_{2\vecind{n}}|} - \frac 1 {|(\vecg{b} - \vecg{c})_{2\vecind{n}}| } - \frac 1 {|(\vecg{b} - \vecg{a} - \vecg{d})_{2\vecind{n}} |} - \frac 1 {|(\vecg{b} - \vecg{a} + \vecg{d} - \vecg{c})_{2\vecind{n}} | } \right],\nn
&&\!\!\!\!{\cal D}_{(c)}  = \frac 12 \sum_{\vecind{n}}\left[  \frac 1{|\vecg{c}_{2\vecind{n}}|} + \frac 1 {|(\vecg{b} - \vecg{c})_{2\vecind{n}|} } + \frac 1 
{|(\vecg{a}- \vecg{c} + \vecg{d}) _{2\vecind{n}} |} +  \frac 1 {|(\vecg{b} - \vecg{a} + \vecg{d} - \vecg{c})_{2\vecind{n}} | }  \right. \nn
&&\left.  - \frac 1{|(\vecg{d} - \vecg{c})_{2\vecind{n}}|} - \frac 1 {|(\vecg{a} - \vecg{c})_{2\vecind{n}}| } - \frac 1 {|(\vecg{a} + \vecg{c} - \vecg{b})_{2\vecind{n}} |} - \frac 1 {|(\vecg{c} - \vecg{b}- \vecg{d})_{2\vecind{n}} | } \right],\nn
&&\!\!\!\!{\cal D}_{(d)}  = \frac 12 \sum_{\vecind{n}}\left[  \frac 1{|(\vecg{d}- \vecg{c})_{2\vecind{n}}|} + \frac 1 {|(\vecg{d} - \vecg{a})_{2\vecind{n}|} } + \frac 1 
{|(\vecg{c}- \vecg{b} - \vecg{d}) _{2\vecind{n}} |} +  \frac 1 {|(\vecg{b} - \vecg{a} - \vecg{d})_{2\vecind{n}} | }  \right. \nn
&&\left.  - \frac 1{|\vecg{d} _{\vecind{2n}}|} - \frac 1 {|(\vecg{b}  - \vecg{d})_{2\vecind{n}}| } - \frac 1 {|(\vecg{a} + \vecg{d} - \vecg{c})_{2\vecind{n}} |} - \frac 1 {|(\vecg{b} - \vecg{a} +\vecg{d} - \vecg{c})_{2\vecind{n}} | } \right].\nn
\ee
In this case, the elementary cell of our 12-dimensional crystal has only one node.  The induced vector potentials are related to ${\cal D}_{(a,b,c,d)}$ by the relations \p{cond-nab} to which one should add similar relations involving ${\cal D}_{(c,d)}$ and $\vecg{\cal A}^{(c,d)}$:
\be
{\cal H}^{(c)}_j &\equiv& \varepsilon_{jkl}\, \pd_k^c {\cal A}^{(c)}_l \ =\ \pd_j^c {\cal D}_{(c)}\,, \nn
\varepsilon_{jkl} {\cal H}^{(ad)}_j  \ &\equiv& \ \pd_k^a {\cal A}^{(d)}_l  - \pd_l^d {\cal A}^{(a)}_k  \ =\   \varepsilon_{jkl}  \pd_j^a {\cal D}_{(d)}  \ =\   \varepsilon_{jkl}  \pd_j^d {\cal D}_{(a)}  \,.
 \ee
etc. 
Calculating the anticommutator $\{\hat {\bar Q}^{\alpha\, {\rm eff}}, \hat Q_\alpha^{\rm eff} \}$, one can derive an expression for the effective Hamiltonian. It represents a rather obvious generalization of   \p{H-chiral-nab}. For example, bearing in mind \p{anticom-eta-5}, one can write a contribution
 \be
&&\hat H^{\rm eff}_{\cal D} \ =\ \frac 15 \left[4\left(({\cal D}_{(a)}^2 +  {\cal D}_{(d)}^2\right) +  6\left({\cal D}_{(b)}^2 +  {\cal D}_{(c)}^2\right) + 6\left({\cal D}_{(a)} {\cal D}_{(b)} + {\cal D}_{(c)} {\cal D}_{(d)}\right)\right. \nn
&& \left.+  \ 4\left({\cal D}_{(a)} {\cal D}_{(c)} + {\cal D}_{(b)} {\cal D}_{(d)}\right) + 2{\cal D}_{(a)} {\cal D}_{(d)} + 8{\cal D}_{(b)} {\cal D}_{(c)}\right]\,.
 \ee
Another contribution represents an analogous quadratic form of  $\hat P_k^{(a,b,c,d)} + {\cal A}^{(a,b,c,d)}_k$. And there are 10 bifermion terms including $\vecg{\cal H}^{(a)},
 \vecg{\cal H}^{(bc)}$, etc.

One can also write and probably,  applying certain efforts and using a clever numerical method,\footnote{The problem boils down to calculating a rather complicated 10-dimensional integral, which is a nontrivial task even for Mathematica.}
calculate numerically the phase space integral \p{IW-int-final} for this model in an attempt to evaluate the index of the effective Hamiltonian. Bearing in mind the absence of the scalar vacuum valleys and the presence of only one relevant region in the elementary cell (two conceptial simplifications), it might coincide  with the index of the original field theory \p{chiral-SU5}. It cannot be excluded, however, that this complicated calculation would give a meaningless fractional result, as a similar calculation did  for $SU(3)$...

\subsection{Spontaneous supersymmetry breaking}

The main physical question to which the Witten index analysis helps to answer is whether supersymmetry is broken spontaneously in a given theory or not. In the gauge theories discussed so far, supersymmetry was not broken.
But there are also models involving chiral matter multiplets where supersymmetry {\it is}  broken. We are sure of that because the breaking occurs in the {\it weak coupling regime} where everything is under control!

The simplest such model (the so-called {\it 3-2 model}) is based on the gauge group
$SU(3) \times SU(2)$ \cite{ADS-ort} (see also the reviews \cite{Amati,versus,Dine}). It includes the following chiral matter multiplets: a multiplet $Q_{j=1,2,3; f=1,2}$ in the representation ({\bf 3}, {\bf 2}), 
two multiplets $U^j$ and $D^j$ in the representation ($\bar{\bf 3}$, {\bf 1}), and a multiplet $L_f$ in the representation ({\bf 1}, {\bf 2}). 

This matter content is similar to the fermion matter content of one generation of  the Standard Model. The fields $Q$ (their fermion components) may be associated with left quarks, the fields $U$ and $D$ with right antiquarks $\bar u_R$ and $\bar d_R$, and the fields $L$ with left electron and neutrino. There is no $U(1)$ group and no ``right electron". The model is anomaly-free in the same way as the Standard model is\footnote{Of course, we do not mean that this model has some phenomenological significance. Probably not. But for a physicist familiar with the SM, this remark may have a mnemonic value.} 

The component Lagrangian of the model includes the quartic scalar interaction terms [cf. Eqs. \p{L-SQED-comp}, \p{L-SQCD-massless}]:
\be
\lb{V-D2}
V \ \sim \ ({\cal D}^{a=1,\ldots,8})^2 + (\tilde {\cal D}^{A=1,2,3})^2,
\ee
where 
\be
\lb{D-aA} 
{\cal D}^a \ &=&\  (q_{j f})^*   (t^a)_j{}^k q_{k f} - u^j (t^a)_j{}^k u^*_k - d^j (t^a)_j{}^k  d^*_k, \nn
\tilde {\cal D}^A \ &=& \ ( q_{j f})^*   (t^A)_f{}^g q_{j g}  + (l_f)^* (t^A)_f{}^g l_g\,.
\ee
Here $t^a$ are the generators of $SU(3)$ and $t^A$ are the generators of $SU(2)$.
 
The important observation is that this model admits a  vacuum valley --- there are classical  field configurations for which the ${\cal D}$ terms \p{D-aA} and the potential \p{V-D2}  vanish.
Up to gauge rotations, a general solution to the equations ${\cal D}^a = \tilde {\cal D}^A = 0$ is
 \begin{empheq}[box=\fcolorbox{black}{white}]{align}
\lb{dolina-32}
q_{j f} \ =\ \left(\begin{array}{cc} \tau_1 & 0  \\ 0 & \tau_2 \\ 0&0  \end{array} \right), \qquad l_f \ =\ \left(  0 , \sqrt{|\tau_1|^2 - |\tau_2|^2}  \right), \nn
u^j \ =\ ( \tau_1,   0,  0), \qquad \qquad d^j \ =\ ( 0, \tau_2,  0),
\end{empheq}
 where $\tau_{1,2}$ are complex parameters ($|\tau_1| \geq |\tau_2|$).

We have already dealt with the vacuum valleys. 
 For the simplest nonchiral $SU(2)$ model, the valley equations were written in Eq. \p{Dreal}. But the solution to these equations,
 \be
 s_j \ =\ \left( \begin{array}{c} \tau \\ 0 \end{array} \right), \qquad \qquad t^j \ =\ (\tau, 0)\, 
 \ee
was not stable in the massless case due to the presence of the nonperturbative instanton-generated superpotential \p{W-ADS}.

The same happens in the 3-2 model. The instantons generate the following superpotential:
 \be
\lb{W-inst-32}
{\cal W}^{\rm inst}_{3-2} \ =\ \frac {\Lambda^7}{Q_{j f} Q_{k g} U^j D^k \varepsilon^{fg} }.
 \ee
As was the case in the nonchiral SQCD, the valley solution \p{dolina-32} is not stable under the scale transformations $\tau \to \lambda \tau$, and we are facing the phenomenon of run-away vacuum. In nonchiral SQCD, the vacuum was stabilized by giving a mass to the matter fields. In the chiral 3-2 model, we cannot do so, but we can instead include a Yukawa term
\be
{\cal L}_{\rm Yukawa} \ =\ h\, \varepsilon^{fg} Q_{j f} U^j L_g \,,
 \ee
which plays the same role --- it locks the valley and fixes the values of $\tau_{1,2}$ at
 \be
\lb{phi-char}
\tau_{1,2} \ =\ c_{1,2} \frac \Lambda {h^{1/7}},
 \ee
where $c_{1,2}$ are some particular irrelevant for us numerical constants.

And here comes an essential difference between the chiral and nonchiral models. In the nonchiral massive SQCD, supersymmetry is not broken and the  energy of the vacuum states is zero. But in the 3-2 model, the vacuum energy density ${\cal E}$ is not zero. To estimate it, it  suffices to substitute the characteristic value \p{phi-char} of the scalar fields in the quartic WZ potential:
 \begin{empheq}[box=\fcolorbox{black}{white}]{align}
{\cal E}  \sim \left( \frac {\pd {\cal W}_{\rm Yukawa}}{\pd \phi}   \right)^2  \sim h^2 \phi_{\rm char}^4  = c_3 h^{10/7} \Lambda^4\,.
 \end{empheq}
As the breaking occurs in the weak coupling regime, the numerical factor $c_3$ can also be evaluated \cite{ADS-ort}.

The fact that supersymmetry is broken in this case can be explained as follows. Both in  nonchiral SQCD and in the 3-2 model, instanton effects bring about a nonzero gluino condensate $\langle \lambda \lambda \rangle$. 
In  nonchiral SQCD, the gluino condensate was related to the scalar vev, $\langle \lambda \lambda \rangle \propto m \tau^2$, in such a way that the vacuum average of the RHS of the Konishi identity \p{Konishi} vanished. This implied that the average of the anticommutator in the LHS of \p{Konishi} vanished too, which meant that supersymmetry was not broken. 

But there is no mass  in the 3-2 model, and nonvanishing $\langle \lambda \lambda \rangle$ means that the anticommutator on the left-hand side of the corresponding Konishi identity (which is valid for all gauge supersymmetric theories including matter multiplets) has a nonzero vev. Hence $\hat Q_\alpha |{\rm vac} \rangle \neq 0$ and supersymmetry is spontaneously broken. The coupling constants in this theory are small and everything (including the vacuum energy) can be perturbatively calculated.

Spontaneous supersymmetry breaking may occur also in other models if 
\begin{itemize}
\item The model admits a vacuum valley.
\item Instanton-generated superpotential pushes the system along the valley.
\item The model is chiral and the mass term cannot be written. To lock the valley, one has to introduce the Yukawa terms.
\end{itemize}

In Secs. 6.1, 6.2, 6.3, we considered three different chiral models, but neither of them satisfied all the conditions above.
\begin{itemize}
\item The chiral SQED involves a valley, the equation \p{valley-SQED}
has a plenty of solutions, but there are no instantons and no instanton-generated superpotential.
\item The minimal chiral $SU(3)$ model with 7 triplets and an antisextet has a valley, has instantons, but the latter fail to generate the required superpotential  \cite{ADS-ort}.
\item The minimal chiral $SU(5)$ model has no valleys. 
\end{itemize}

 The next in complexity model, where supersymmetry is spontaneously broken in the weak coupling regime, is the $SU(5)$ chiral model with {\it two} quintets $S^{f=1,2}$ and two antidecuplets $\Phi^f$. In this case, the expression for ${\cal D}^a$ includes twice as much variables as in the minimal $SU(5)$ model and solutions to the equations ${\cal D}^a = 0$ can be found though their explicit form is not simple \cite{Veld}. The instanton superpotential {\it is} generated in this valley. It reads
\be
\lb{W-inst-SU5}
{\cal W}^{\rm inst}_{SU(5)} \ =\ \frac {\Lambda^{11}}{B^{fg}B_{fg}},
 \ee
where $B_{fg} = \varepsilon_{ff'}  \varepsilon_{gg'} B^{f'g'} $ and $B^{fg} = \varepsilon_{ijklp}  (\Phi^f)^{ij} (\Phi^h)^{kl}
(\Phi_h)^{pq} S^g_q$. 
To lock the valley, one introduces the Yukawa term $\propto (\Phi^1)^{ij} S^f_i S^g_j \varepsilon_{fg}$.

Many more models with spontaneous SUSY breaking have been studied in the literature. We address the reader to the original paper \cite{ADS-ort}, to the reviews \cite{Amati,versus} and to many papers written afterwards for further details.

I am indebted to Ken Konishi and Misha Vysotsky for valuable comments.

\section*{Appendix A. Nuts and bolts of group theory.}

\renewcommand{\theequation}{A.\arabic{equation}}
\setcounter{equation}0

We quote here some  very basic facts and definitions of the theory of simple Lie groups and Lie algebras.

For a  Lie algebra of rank $r$, we can introduce the {\it Cartan-Weyl basis}, which includes $r$ elements of the Cartan subalgebra $\mathfrak{C}$  and a set of 
{\it root vectors} $E_{\alpha_j}$. Then for any $h \in \mathfrak{C}$, 
\be 
 [h, E_{\alpha_j}] \ =\ \alpha_j(h) E_{\alpha_j} \,.
 \ee
The functions $\alpha_j(h)$ are linear forms on the Cartan subalgebra, which are called the  {\it roots}.  If we choose in $\mathfrak{C}$ an orthonormal basis
the roots are represented as $r$-dimensional vectors.

 The root system has the following properties:

  $\bullet$  For any root $\alpha$, $-\alpha$ is also a root. The corresponding root vectors are Hermitially conjugated to each other: $E_{-\alpha} = E^\dagger_\alpha$. 

 $\bullet$  Let $\alpha, \beta$ be two arbitrary  roots. Then, if $\alpha+\beta$ is also a root, 
 \be
\lb{propto}
[E_\alpha, E_\beta] \ \propto\  E_{\alpha + \beta}\,.
\ee
If $\alpha+\beta$ is not a root, the commutator vanishes. 

  $\bullet$  In the set of all roots, one can distinguish a subset of $r$ {\it simple} roots. 
All other roots represent linear combinations of simple roots with integer coefficients, which either are all positive (and in this case, we are dealing with the {\it positive} roots) 
or all negative for {\it negative} roots.  

$\bullet$  The notation $\alpha$ is convenient to reserve for positive roots and denote  the negative roots by $-\alpha$.
For the positive roots, the sum of the integer coefficients of their expansion is called the {\it level} of the root. A root with the maximal level is called the {\it highest} root.
The  coefficients of the expansion of the highest root are called {\it Dynkin labels} of the algebra.

$\bullet$ The commutator $[E_\alpha, E_{-\alpha}]$ belongs to $\mathfrak{C}$.  A {\it coroot} corresponding to the root $\alpha$ and denoted as  $\alpha^\vee$ is proportional to this commutator, $\alpha^\vee = c [E_\alpha, E_{-\alpha}]$, with the coefficient $c$ determined from the condition
\be
e^{2\pi i \alpha^\vee} \ =\ \mathbb{1}, \qquad {\rm but} \qquad  e^{ i\phi \alpha^\vee} \ \neq\ \mathbb{1} \quad {\rm if} \ \phi < 2\pi\,.
 \ee
In the convenient {\it Chevalley normalization} of the root vectors, $c = 1$.

$\bullet$ Let $E_\alpha$ be a simple root vector. A {\it fundamental coweight} $\omega_\alpha$ is an element of the Cartan subalgebra that commutes with all  simple root vectors besides $E_{\pm \alpha}$, while 
$ [\omega_\alpha, E_\alpha] = E_\alpha$.

$\bullet$ For a generic  Lie algebra, the roots (the $r$-dimensional vectors that represent them) may have different length. If the length of the longer roots is normalized to 1, the  lengths of the shorter roots may be $1/\sqrt{2}$ or  $1/\sqrt{3}$ (for $g_2$). Two different long simple roots may be orthogonal or form the angle $2\pi/3$ between them. The same concerns two short simple roots. The angle between a long simple root and a short simple root of length  $1/\sqrt{2}$ may be $\pi/2$ or $3\pi/4$. The algebra $g_2$ involves a long simple root and a short simple root of length  $1/\sqrt{3}$. They form the angle $5\pi/6$.

$\bullet$ A convenient way to represent a generic simple algebra $\mathfrak{g}$ [or a generic simple group ${\cal G} = \exp(\mathfrak{g})$] is to draw a {\it Dynkin diagram}. Many such diagrams were drawn in   Sec.~4.   The large circles there are the long simple roots and smaller circles represent short simple roots. If the roots are orthogonal, they are not related by a line. A single line means the angle $2\pi/3$, a double line the angle $3\pi/4$ and a triple line  the angle $5\pi/6$. The dashed circles stand for $-\theta$, where 
$\theta$ is  the highest root (that is why the diagrams drawn in Fig. \ref{Dyn7} etc. are {\it extended}; in the ordinary Dynkin diagrams, the highest root is not shown).  The numbers in the circles are the Dynkin labels.

To illustrate this, consider the algebra $spin(7)$. It has rank 3 and  3 simple roots: $\alpha, \beta$ and $\gamma$, two long and one short.  Its Dynkin diagram  was drawn in  Fig. \ref{Dyn7}. There are 6 other positive roots: four long roots $\alpha+\beta, \beta + 2\gamma, \alpha+\beta+2\gamma$ and the highest root $\theta = \alpha + 2\beta + 2\gamma$, and two short roots $\beta+\gamma$ and $\alpha+\beta+\gamma$. The simple coroots were written in Eq. \p{coroots}. The basis of $spin(7)$ includes 9 positive and 9 negative root vectors and 3 independent elements of $\mathfrak{C}$. The dimension of $spin(7)$ is $9 \cdot 2 + 3 = 21$.

 We give here three more definitions.  

{\it (i)} The {\it maximal torus} is an Abelian subgroup of $\cal G$ generated by its Cartan subalgebra. 

{\it (ii)}  The {\it centralizer} of an element $\Omega \in {\cal G}$ or of a set of elements $\{\Omega_i \in {\cal G}\}$ is the subgroup of $\cal G$ that commutes with $\Omega$ (resp. with all the elements $\Omega_i$ in the set).

{\it (iii)} We define the {\it Dynkin index} $I(R)$ of a representation $R$ in group $\cal G$  as\footnote{Sometimes, it is defined with an extra factor 2 \cite{Slansky}.}
\begin{empheq}[box=\fcolorbox{black}{white}]{align}
\lb{Dyn-ind}
{\rm Tr} \{ T^a T^b \} \ =\   I(R) \delta^{ab} \,,
\end{empheq}
where  $T^a$ are the generators in the representation $R$ normalized in a standard way so that $[T^a, T^b] = if^{abc} T^c$, with $ f_{abc}$ being the structure constants of the group.

In any simple group, the Dynkin index of the adjoint representation  coincides with  the adjoint Casimir eigenvalue. Another name for the latter important invariant is {\it dual Coxeter number}.
 $I=1/2$ for the fundamental representation in $SU(N)$, and $I=1$ for the vector representation in $SO(N \ge 5)$.

\subsection*{Octonions}
\begin{figure}
\bc
    \includegraphics[width=.5\textwidth]{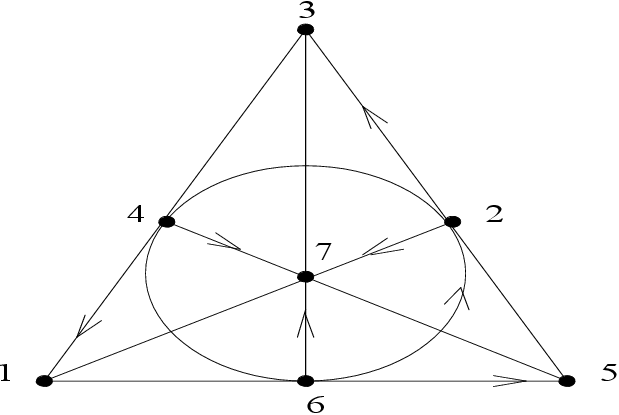}                  
     \ec
\caption{Fano plane}
\label{Fano}
\end{figure}
In   Sec.~5.3, we were discussing a sypersymmetric gauge theory based on the group $G_2$. This group can be defined as 
the group of automorphisms of {\it octonion} algebra. 
An octonion is an object $a_0 + a_\alpha e_\alpha$, where $a_0$ and $a_{1,\ldots,7}$ are real numbers and $e_\alpha$ satisfy the algebra
\be 
\lb{alg-oct}
 e_\alpha e_\beta \ =\ - \delta_{\alpha\beta} + f_{\alpha\beta\gamma} e_\gamma \,.
 \ee
One of the possible choices for the antisymmetric tensor $ f_{\alpha\beta\gamma}$ is
\be
  \label{f}
  f_{165} = f_{341} = f_{523} = f_{271} = f_{673} = f_{475} = f_{246} = 1
  \ee
  and all other nonzero components are restored by antisymmetry. The values
  (\ref{f}) can be mnemonized by looking at the {\it Fano plane} depicted in Fig. \ref{Fano}.
  $f_{\alpha\beta\gamma}$  is nonzero only for the indices lying on the
 same line, with the arrows indicating the order of indices when
  $f_{\alpha\beta\gamma}$ is positive.

\section*{Appendix B. The derivation of Eq. \p{Schr-rad}}
\renewcommand{\theequation}{B.\arabic{equation}}
\setcounter{equation}0

The Schr\"odinger equation for a scalar charged particle in a monopole field was solved by Igor Tamm\footnote{He was visiting Cambridge at that time and learned about Dirac's work directly from Dirac.} back in 1931 \cite{Tamm}. Let the electric charge of the particle be $e=1$. The angular momentum operator that commutes with the Hamiltonian 
\be
\hat H_0 \ =\ \frac 12 \left( \hat {\vecg{P}} + \vecg{\cal A}\right)^2\,,
 \ee
where $ \vecg{\cal A}$ is the vector potential of a monopole\footnote{Dirac and Tamm used an expression for $\vecg{\cal A}$ that involved a singularity on the {\it Dirac string} [the singularity at $\theta = \pi$ in Eq. \p{vecpot}], but later it was realized that the projection  $\vecg{\cal A} \cdot \vecg{n}$ represents a {\it fiber bundle} on $S^2$, which can be described by nonsingular expressions on two overlapping charts \cite{Wu-Yang}.} with the minimal magnetic charge $g= 1/2$,
has the form 
 \be
\hat {\vecg{L}_0} \ = \ \vecg{r} \times \left( \hat {\vecg{P}} + \vecg{\cal A}\right)  + \frac {\vecg{r}}{2r}\,.
 \ee
The second term is the angular momentum of electromagnetic field. The eigenvalues of $\hat {\vecg{L}_0}^2$ are $l(l+1)$ with r half-integer $l$.
The eigenvalues of  $\hat{L}_3^{(0)}$ are $m = -l,-l+1,\ldots, l$. The corresponding eigenfunction are the so-called 
{\it monopole harmonics} \cite{Wu-Yang}. 

We have 
\be
\hat H_0 \ =\ - \frac 1{2r^2} \frac \pd{\pd r} \left( r^2  \frac \pd{\pd r} \right) +  \frac 1{2r^2} \left[\vecg{r} \times (\hat{\vecg{P}} + \vecg{\cal A})  \right]^2 
=  - \frac 1{2r^2} \frac \pd{\pd r} \left( r^2  \frac \pd{\pd r} \right)  +  \frac {\hat {\vecg{L}_0}^2 - \frac 14}{2r^2} \,.
\ee
For the states with definite $l$, this gives
$$  \hat H_0 \ =\  - \frac 1{2r^2} \frac \pd{\pd r} \left( r^2  \frac \pd{\pd r} \right) 
\ + \  \frac {l(l+1) - \frac 14}{2r^2} \,.$$
The full Hamiltonian \p{H-chiral} commutes with
 \be
\lb{full-L}
\hat{\vecg{L}} \ =\ \hat{\vecg{L}}_0 + \frac 12 \hat{\bar \eta} \vecg{\sigma} \eta.
 \ee
In fact, \p{full-L} may be interpreted as a Hamiltonian describing the interaction of a spin $1/2$ particle with gyromagnetic ration 2 (twice as much as for electron) with a monopole \cite{Malkus}\footnote{One can also consider the interaction of a $W$ boson, which has spin 1 \cite{W-mon}. A bound state of a monopole and a   $W$ boson is a {\it dyon} \cite{dyon}.} and the second term in \p{full-L} is the spin part of the angular momentum. The matrix $\vecg{\sigma} \cdot \vecg{n}$ has the eigenvalues $\lambda = \pm 1$. Choosing $\lambda = -1$ (this corresponds to the total angular momentum $j = 0$), we are led to the radial Schr\"odinger equation
 \be
- \frac 1{2r} \frac {\pd^2}{\pd r^2} (r\psi) - \frac 1{4r^2} \psi \ =\ E\psi \,.
\ee
Adding the term $D^2/2$ in the potential, we reproduce \p{Schr-rad}.


\begin{thebibliography}{99}

\bibitem{Wit82} E.~Witten, \href{https://doi.org/10.1016/0550-3213(82)90071-2}{\it Constraints on supersymmetry breaking}, 
Nucl. Phys. B {\bf 202} (1982)  253.

\bibitem{Cecotti} S.~Cecotti and   L.~Girardello, {\it Functional measure, topology and dynamical supersymmetry breaking}, Phys. Lett. {\bf B110} (1982) 39;\\
 L.~Girardello, C.~Imbimbo and 
S.~Mukhi, {\it On constant configurations and the evaluation of the Witten index},  
Phys. Lett. {\bf B132} (1982) 69.

\bibitem{Cartan}  H.~Cartan, {\it La transgression dans un groupe de Lie et dans un fibr\'e principal}, Colloque de topologie (espaces fibr\'es), Centre belge de recherches math\'ematiques, Masson et Cie, Paris, 1951, p.57.

\bibitem{At-Sin} M.F.~Atiyah and I.M.~Singer, {\it The index of elliptic operators}, Annals Math. {\bf 87} (1968) 484, 546; {\bf 93} (1971) 119, 139.

\bibitem{Wit-Morse} E.~Witten, 
{\it Supersymmetry and Morse theory}, J. Diff. Geom. {\bf 17} (1982) 661.

\bibitem{glasses} A.V.~Smilga, {\it Differential geometry through supersymmetric glasses}, World Scientific, 2020.

\bibitem{WB} J.~Wess and J.~Bagger,  {\it Supersymmetry and Supergravity}, Princeton University Press, 1992.

\bibitem{Golfand} Yu.A.~Golfand and E.P.~Likhtman, {\it Extension of the algebra of Poincare group generators by bispinor generators}, JETP Lett. {\bf 13} (1971) 323.

\bibitem{BO} M.~Born and J.R.~Oppenheimer, {\it Zur Quantentheorie der Molekeln}, Annalen der Physik {\bf 389} (1927) 457.

\bibitem{BOcorners} A.V.~Smilga, {\it Perturbative corrections to effective zero-mode Hamiltonian in supersymmetric QED}, Nucl. Phys. {\bf B291} (1987) 241;
{\it Born-Oppenheimer corrections to the effective zero-mode Hamiltonian in SYM theory},  JHEP {\bf 04} (2002) 054, {\tt arXiv:hep-th/0201048}.

\bibitem{Zohar} Z.~Komargodski, K.~Ohmori, K.~Roumpedakis and S.~Seifnashri, {\it Symmetries and strings
of adjoint $QCD_2$}, J. High Energ. Phys. {\bf 03} (2021) 103, {\tt arXiv:2008.07567 [hep-th]}.

\bibitem{FI} P.~Fayet and J.~Illiopoulos, {\it Spontaneously broken supergauge symmetries and goldstone spinors}, Phys. Lett. {\bf  B51}  (1974) 461. 

\bibitem{KRS} A.~Keurentjes, A.A.~Rosly and A.V.~Smilga, {\it Isolated vacua in supersymmetric Yang-Mills theories}, 
Phys. Rev. {\bf D58} (1998) 081701, {\tt arXiv:hep-th/9805183}.

\bibitem{Baal} P.~van Baal, {\it Some results for $SU(N)$ gauge fields on the hypertorus}, Comm. Math. Phys. {\bf 85} (1982) 529.

\bibitem{Jackiw-Rebbi} R.~Jackiw and C.~Rebbi, {\it Vacuum periodicity in a Yang-Mills quantum theory},
Phys. Rev. Lett. {\bf 37} (1976) 172.

\bibitem{BPST} A.A.~Belavin, A.M.~Polyakov, A.S.~Schwartz and Yu.S.~Tyupkin, {\it Pseudoparticle solutions of the Yang-Mills equations}, Phys. Lett. {\bf B59}, 85 (1975).

\bibitem{Hooft} G.~`t Hooft, {\it Computation of the quantum effects due to a four-dimensional pseudoparticle}, Phys. Rev. {\bf D14} (1976) 3432.

\bibitem{twisted-Hooft} G.~`t Hooft, {\it A property of electric and magnetic flux in non-Abelian gauge theories}, Nucl. Phys. {\bf B153} (1979) 141;
{\it Some twisted selfdual solutions for the Yang-Mills equations on a hypertorus}, Commun. Math. Phys. {\bf 81} (1981) 267.

\bibitem{Selivanov} K.G.~Selivanov and A.V.~Smilga, {\it Classical Yang-Mills vacua on $T^3$: explicit constructions}, Phys. Rev. {\bf D63} (2001) 125020, {\tt arXiv:hep-th/0010243}.

\bibitem{Mumford} D.~Mumford, {\it Tata lectures on Theta}, Birkhauser, Boston, 1983.

\bibitem{merons} C.G.~Callan, R.~Dashen and D.J.~Gross, {\it Mechanism for quark confinement}, Phys. Lett. {\bf 66B} (1977) 375.

  \bibitem{gluino-cond} V.~Novikov, M.~Shifman, A.~Vainshtein and V.~Zakharov, {\it Instanton effects in supersymmetric theories}, Nucl. Phys. {\bf B229}  (1983) 407; \\
D.~Amati, G.C.~Rossi and G.~Veneziano, {\it Instanton effects in supersymmetric gauge theories}, Nucl. Phys. {\bf B249} (1985) 1.
  

\bibitem{Pisarski} T.~Bhattacharya, A.~Gocksh, C.~Korthals Altes and R.D.~Pisarski, {\it $Z(N)$ interface tension in a hot $SU(N)$ gauge theory},  Nucl. Phys. {\bf B383} (1992) 497, {\tt arXiv:hep-ph/9205231}.

\bibitem{no-bubbles} A.V.~Smilga, {\it Are Z(N) bubbles really there?}, Ann. Phys. {\bf 234} (1994)1;\\
{\it Physics of thermal $QCD$}, Phys. Repts. {\bf 291} (1997) 1, {\tt arXiv:hep-ph/9612347}, Chap.~3.

\bibitem{Kiskis} J.E.~Kiskis, {\it Phase of the Wilson line}, Phys. Rev. {\bf D49} (1995) 3781, {\tt arXiv:hep-lat/9407001}.

\bibitem{Ven-Yan} G.~Veneziano and S.~Yankielowicz, {\it An effective Lagrangian for the pure ${\cal N} = 1$ supersymmetric Yang-Mills theory}, Phys. Lett. {\bf 113B} (1982) 231.

\bibitem{WZ} J.~Wess and B.~Zumino, {\it A Lagrangian model invariant under supergauge transformations}, Phys. Lett. {\bf B49} (1974) 52. 

\bibitem{Kov-Shif} A.~Kovner and M.~Shifman, {\it Chirally symmetric phase of supersymmetric gluodynamics}, Phys. Rev. {\bf D56} (1997) 2396, {\tt arXiv:hep-th/9702174}.

\bibitem{KSS}  A.~Kovner, M.~Shifman and A.~Smilga,  {\it Domain walls in supersymmetric Yang-Mills
theories}, Phys. Rev. {\bf D56} (1997) 7978, {\tt arXiv:hep-th/9706089}.

\bibitem{Witten-O7} E.~Witten, {\it Toroidal compactification without vector structure}, JHEP {\bf 02} (1998) 006, {\tt arXiv:hep-th/9712028}.

\bibitem{Kac1} V.G.~Kac and A.V.~Smilga, {\it Vacuum structure in supersymmetric Yang–Mills theories with any gauge
group}, in: [Many faces of the superworld: Yuri Golfand memorial volume], World Scientific, 2000, p.185, {\tt arXiv:hep-th/9902029}.

\bibitem{Keurentjes} A.~Keurentjes, {\it Nontrivial flat connections on the 3 torus I: G(2) and the orthogonal groups},  JHEP {\bf 05} (1999) 001,
{\tt arXiv:hep-th/9901154}; {\it Nontrivial flat connections on the three torus. 2. The exceptional groups F4 and E6, E7, E8},  JHEP {\bf 05} (1999) 014.

\bibitem{Borel} A.~Borel, R.~Friedman, and J.~Morgan, {\it Almost commuting elements in compact Lie groups}, Mem. Amer. Math. Soc. {\bf 157} (2002), No. 747, {\tt arXiv:math.GR/9907007}.

\bibitem{Salam} A.~Salam and J.~Strathdee,  {\it Supersymmetry and fermion number conservation}, Nucl. Phys. {\bf B87} (1975) 85;

P.~Fayet, {\it Supergauge invariant extension of the Higgs mechanism and a model for the electron and its neutrino},  Nucl. Phys. {\bf B90} (1975) 104.

\bibitem{ADS-84} I.~Affleck, M.~Dine, and N.~Seiberg, {\it Dynamical supersymmetry breaking in supersymmetric QCD}, Nucl. Phys. {\bf B241} (1984) 493.

\bibitem{NSVZ-supinst} V.A.~Novikov, M.A.~Shifman, V.I.~Zakharov, and A.I.~Vainshtein, {\it Supersymmetric instanton calculus (gauge theories with matter)}, Nucl. Phys. {\bf B260} (1985) 157.

\bibitem{Amati} D.~Amati, K.~Konishi, Y.~Meurice, G.C.~Rossi, and G.~Veneziano, {\it Non-perturbative aspects in supersymmetric gauge theories}, Phys. Repts. {\bf 162} (1988) 169.

\bibitem{versus} M.~Shifman and A.~Vainshtein, {\it Instantons versus supersymmetry: fifteen years later}, in: [M. Shifman, ITEP lectures on particle physics and field theory, p.485], World Scientific, 1999, {\tt arXiv:hep-th/9902018}.

\bibitem{Konishi}  K.~Konishi, {\it Anomalous supersymmetry transformation of some composite operators in SQCD}, Phys. Lett. {\bf B135} (1984) 439. 

\bibitem{Clark} T.E.~Clark, O.~Piguet, and K.~Sibold, {\it Absence of radiative corrections to the axial current anomaly in supersymmetric QED}, Nucl. Phys. {\bf B159} (1979) 1.

\bibitem{BPS} E.B.~Bogomolny, {\it Stability of classical solutions}, Sov. J. Nucl. Phys. {\bf 24} (1976) 441; 

M.K.~Prasad and C.M.~Sommerfeld, {\it An exact classical solution for the `t Hooft monopole and Julia-Zee dyon}, Phys. Rev. Lett. {\bf 35} (1975) 760.

\bibitem{BPS-2D} E.~Witten and D.~Olive, {\it Supersymmetry algebras that include topological charges}, Phys. Lett. {\bf B78} (1978) 97;

S.~Cecotti and C.~Vafa, {\it On classification of N=2 supersymmetric theories}, Commun. Math. Phys. {\bf 158} (1993) 569.

\bibitem{Chibisov} B.  Chibisov and M.  Shifman, {\it BPS-saturated walls in supersymmetric theories}, Phys. Rev. D {\bf 56} (1997) 7990,  {\tt arXiv:hep-th/9706141}. 



\bibitem{wall-N} A.V.~Smilga, {\it BPS domain walls in supersymmetric QCD: higher unitary groups}, Phys. Rev. {\bf D58} (1998) 065005, {\tt arXiv:hep-th/9711032}.



\bibitem{TVY} T.~Taylor, G.~Veneziano, and S.~Yankielowicz, {\it Supersymmetric QCD and its massless limit: an effective Lagrangian analysis}, Nucl. Phys. {\bf B218} (1983) 493.

\bibitem{Veselov} A.~Smilga and A.~Veselov, {\it Complex BPS domain walls and phase transition in mass in supersymmetric QCD}, Phys. Rev. Lett. {\bf 79} (1997) 4529, {\tt arXiv:hep-th/9706217};

 {\it Domain walls zoo in supersymmetric QCD}, Nucl. Phys. {\bf B515} (1998) 163, {\tt arXiv:hep-th/9810123}.

\bibitem{tenacious} B. de Carlos and J.M. Moreno, {\it Domain walls in supersymmetric QCD: From weak to strong coupling}, Phys. Rev. Lett. {\bf 83} (1999) 2120, {\tt arxiv:hep-th/9905165 [hep-th]};

D. Binosi and T. van Veldhuis, {\it Domain walls in supersymmetric QCD: The taming of the zoo}, 
 Phys. Rev. {\bf D63} (2001) 085016, {\tt arxiv:hep-th/0011113 [hep-th]};

A. Smilga, {\it Tenacious domain walls in supersymmetric QCD},     Phys. Rev. {\bf D64} (2001) 125008,
{\tt arxiv:hep-th/0104195 [hep-th]}.



\bibitem{ADS-ort} I.~Affleck, M.~Dine, and N.~Seiberg, {\it Dynamical supersymmetry breaking in four dimensions and its phenomenological implications}, Nucl. Phys. {\bf B256} (1985) 557.

 \bibitem{Cordes} S.F.~Cordes and M.~Dine, {\it Chiral symmetry breaking in supersymmetric $O(N)$ gauge theories}, Nucl. Phys. {\bf B273} (1986) 581.

 \bibitem{ShifVain-ort} M.A.~Shifman and A.I.~Vainshtein, {\it On gluino condensation in supersymmetric gauge theories with $SU(N)$ and $O(N)$ groups}, Nucl. Phys. {\bf B296} (1988) 445.

 \bibitem{Morozov} A.Yu.~Morozov, M.A.~Ol'shanetsky, and M.A.~Shifman,  {\it Gluino condensate in supersymmetric gluodynamics (II)}, Nucl. Phys. {\bf B304} (1988) 291.

\bibitem{G2} A.V.~Smilga, {\it 6+1 vacua in supersymmetric QCD with $G_2$ gauge group}, Phys. Rev. {\bf D58} (1998) 105014, {\tt arXiv:hep-th/9801078}.

\bibitem{G2-ante} I.~Pesando, {\it Exact results for the supersymmetric $G_2$ gauge theories}, Mod. Phys. Lett. {\bf A10} (1995) 1871, {\tt arXiv:hep-th/9506139};

S.B.~Giddings and J.B.~Pierre, {\it Some exact results in supersymmetric theories based on exceptional groups}, Phys.~Rev. {\bf D52} (1995) 6065. 

\bibitem{Gorsky} K.~Intriligator and N.~Seiberg, {\it Phases of ${\cal N} = 1$  supersymmetric gauge theories in four dimensions},  Nucl. Phys. {\bf B431} (1994) 551, {\tt arXiv:hep-th/9408155}.

\bibitem{chiral-SQED} A.V.~Smilga, {\it Structure of vacuum in chiral supersymmetric quantum electrodynamics},  
Sov. Phys. JETP {\bf 64} (1986) 8.

\bibitem{Palumbo} A.  Nakamura and F.  Palumbo, {\it Ordering ambiguities in supersymmetric gauge theories}, Phys. Lett. B {\bf 147} (1984) 96.

\bibitem{Ritten}  M.~de Crombrugghe and V.~Rittenberg, {\it Supersymmetric quantum mechanics},  Ann. Phys. {\bf 151} (1983) 99.

\bibitem{Berry} S.~Pancharatnam, {\it Generalized theory of interference, and its applications. Part I. Coherent pencils}, Proc. Ind. Acad. Sci. {\bf A44} (1956) 247;

M.~V.~Berry (1984), {\it Quantal phase factors accompanying adiabatic changes}, Proc. Royal Soc. {\bf A392} (1984) 45.

\bibitem{falling} K.M.~Case, {\it Singular potentials}, Phys. Rev. {\bf 80} (1950) 797;

K.~Meetz, {\it Singular potentials in non-relativistic quantum mechanics}, Nuov. Cim. {\bf 34} (1964) 690;

A.M.~Perelomov and V.S.~Popov, {\it Collapse onto scattering center in quantum mechanics}, Teor. Mat, Fiz. {\bf 4} (1970) 48.

\bibitem{Malkus} 
 P.P. Banderet, {\it Zum Theorie singularer Magnetpole}, Helv. Phys. Acta {\bf 19} (1946) 503;
W.V.R. Malkus, {\it The interaction of the Dirac magnetic monopole with matter}, Phys. Rev. {\bf 83} (1951) 899.

\bibitem{Blok} B.Yu.~Blok and A.V.~Smilga, {\it Effective zero-mode Hamiltonian in supersymmetric chiral non-Abelian gauge theories}, Nucl. Phys. {\bf B287} (1987) 589.

\bibitem{Dine} M.~Dine and J.D.~Mason, {\it Supersymmetry and its dynamical breaking}, Rept. Progr. Phys. {\bf 74} (2011) 056201, {\tt arXiv:1012.2836 [hep-th]}.

\bibitem{Veld} T.~ter Veldhuis, {\it Unexpected symmetries in classical moduli spaces}, Phys. Rev. {\bf D58} (1998) 015010, {\tt arXiv:hep-th/9811132}.

\bibitem{Slansky} R.~Slansky, {\it Group theory for unified model buidling}, Phys. Rept. {\bf 79} (1981) 1.

\bibitem{Tamm} I.~Tamm, {\it Die verallgemeinerten Kugelfunktionen und die Wellenfunktionen eines Elektrons
im Felde eines Magnetpoles.}, Z. Phys. {\bf 71} (1931) 141.

\bibitem{Wu-Yang} T.T.~Wu and C.N.~Yang, {\it Dirac monopole without strings: monopole harmonics}, Nucl. Phys. {\bf B107} (1976) 365. 



Y. Kazama, C.N. Yang and A.S. Goldhaber, {\it Scattering of a Dirac particle with charge $Ze$
by a fixed magnetic monopole}, Phys. Rev. {\bf D15} (1977) 2287; 



\bibitem{W-mon} A.V.~Smilga, {\it $W$ boson scattering  on monopoles},  Sov. J. Nucl. Phys., {\bf 47} (1988) 692.

\bibitem{dyon} B.~Julia and A.~Zee, {\it Poles with both magnetic and electric charges in nonabelian gauge theory}, Phys. Rev. D {\bf 11} (1975) 2227.



\end{thebibliography}
\end{document}